\documentclass[letterpaper,USenglish,numberwithinsect]{no-lipics}

\usepackage[utf8]{inputenc}
\usepackage[T1]{fontenc}

\usepackage{multicol}
\usepackage{mathtools}
\usepackage{xspace}
\usepackage{contour}

\newcommand{\A}{\mathcal{A}}
\newcommand{\Ham}{\mathsf{HD}}

\newcommand{\Z}{\mathbb{Z}}
\newcommand{\R}{\mathbb{R}}

\newcommand{\eps}{\varepsilon}

\newcommand{\AG}{\mathsf{AG}}

\newcommand{\No}{\textsc{No}\xspace}
\newcommand{\Yes}{\textsc{Yes}\xspace}
\newcommand{\poly}{\mathsf{poly}}

\newcommand{\wall}{%
    \begin{tikzpicture}[x=1ex, y=1ex]
        \foreach \x in {0,2,4} {
            \draw (\x,0) rectangle (\x+2,1);
        }
        \foreach \x in {1,3} {
            \draw (\x,1) rectangle (\x+2,2);
        }
    \end{tikzpicture}
}

\newcommand{\dd}{\mathinner{.\,.}}
\newcommand{\fragmentco}[2]{{[}\,#1\dd#2\,{)}}
\newcommand{\fragmentoc}[2]{{(}\,#1\dd#2\,{]}}

\newcommand{\fragmentcc}[2]{{[}\,#1\dd#2\,{]}}
\newcommand{\fragment}{\fragmentcc}

\newcommand{\per}{\mathsf{per}}
\newcommand{\DTW}{\mathsf{DTW}}
\newcommand{\PMDTW}{\mathsf{PM\text{\_}DTW}}
\newcommand{\kDTW}{k\textit{-}\mathsf{DTW}}
\newcommand{\ED}{\mathsf{ED}}
\newcommand{\kED}{k\textit{-}\mathsf{ED}}

\newcommand{\bigz}{\mathsf{big}}
\newcommand{\smallz}{\mathsf{small}}

\newcommand{\rle}{\mathsf{rle}}
\newcommand{\shrink}{\mathsf{shrink}}
\newcommand{\tshrink}{2\text{-}\mathsf{shrink}}

\newcommand{\half}{\mathsf{half}}
\newcommand{\unit}{\mathsf{unit}}

\newcommand{\dchara}{\mathtt{a}}
\newcommand{\dcharb}{\mathtt{b}}
\newcommand{\dchare}{\mathtt{e}}
\newcommand{\dcharf}{\mathtt{f}}
\newcommand{\dcharw}{\mathtt{w}}

\newcommand{\AGed}{\AG^\ED}
\newcommand{\ceil}[1]{\left\lceil{#1}\right\rceil}
\newcommand{\floor}[1]{\left\lfloor{#1}\right\rfloor}

\newcommand{\Oh}{\mathcal{O}}
\newcommand{\Ohtilde}{\mathpalette\OhtildeAux{}}
\newcommand{\OhtildeAux}[2]{%
  \raisebox{0.25ex}{$\mathsurround=0pt#1\widetilde{%
    \smash{\raisebox{-0.25ex}{$\mathsurround=0pt#1\Oh$}}}$}%
}

\newcommand{\Otild}{\Ohtilde}

\newcommand{\para}[1]{\subparagraph{#1}}

\newcommand{\dist}{\mathsf{dist}}
\newcommand{\cost}{\mathsf{cost}}

\newcommand{\ca}{\mathtt{a}}
\newcommand{\cb}{\mathtt{b}}
\newcommand{\cc}{\mathtt{c}}
\newcommand{\dol}{\texttt{\$}}
\newcommand{\hash}{\texttt{\#}}

\newcommand{\OV}{\mathsf{OV}}

\newcommand{\ainote}[1]{}

\newtheorem{observation}[theorem]{Observation}

\newcommand{\restatementwithproof}[1]{%
  \begingroup
  \let\lipicsRestatementEnd\relax
  #1%
  \endgroup
}

\crefname{claim}{Claim}{Claims}
\crefname{observation}{Observation}{Observations}

\title{Dynamic Time Warping in the Low-Distance Regime}
\subtitle{Tight Lower Bounds and Input-Sensitive Algorithms}
\titlerunning{Dynamic Time Warping in the Low-Distance Regime}

\author{Itai Boneh}{Institute of Computer Science, University of Wrocław, Wrocław, Poland\and \url{https://sites.google.com/view/itai-boneh}}{itai.boneh@cs.uni.wroc.pl}{https://orcid.org/0009-0007-8895-4069}{Partially supported by Israel Science Foundation grant 810/21. Partially supported by the Polish National Science Centre grant number 2023/51/B/ST6/01505}
\author{Shay Golan}{Department of Computer Science, Ariel University, Ariel, Israel\and \url{https://sites.google.com/view/shaygolan}}{shayg@ariel.ac.il}{https://orcid.org/0000-0001-8357-2802}{Partially supported by Israel Science Foundation grant 810/21.}
\author{Tomasz Kociumaka}{Max Planck Institute for Informatics, Saarland Informatics Campus, Saarbrücken, Germany\and \url{https://people.mpi-inf.mpg.de/~tkociuma/}}{tomasz.kociumaka@mpi-inf.mpg.de}{https://orcid.org/0000-0002-2477-1702}{}
\authorrunning{I. Boneh, S. Golan, and T. Kociumaka}

\date{}

\begin{document}

\maketitle

\begin{abstract}
Dynamic Time Warping (DTW) is a classical similarity measure for strings and time series, designed to compare sequences up to local stretching.
It is particularly natural when two sequences represent similar objects sampled or traversed at varying speeds.
Given non-empty strings $S$ and $T$ over an alphabet $\Sigma$ and a mismatch-cost function $\delta:\Sigma^2\to\R_{\ge 0}$, the dynamic time warping $\DTW_\delta(S,T)$ is the minimum value of $\sum_i \delta(S'[i],T'[i])$ over all pairs of equal-length strings $S'$ and $T'$ obtained by duplicating characters in $S$ and $T$, respectively.
For strings of length at most $n$, the value $\DTW_\delta(S,T)$ is computable in $\Oh(n^2)$ time via dynamic programming [Sakoe and Chiba; 1978], and the running time cannot be reduced beyond $n^{2-o(1)}$ under the Orthogonal Vectors Hypothesis (OVH), and hence also under the Strong Exponential Time Hypothesis [Abboud, Backurs, and Vassilevska Williams; 2015] [Bringmann and Künnemann; 2015].

In many settings, however, one is interested not in arbitrary pairs of sequences but in highly similar ones when the optimal expansions incur only a few mismatches. This leads to the low-distance regime, where the running time is expressed in terms of $n$ and an integer budget $k\ge1$ upper-bounding $\DTW_\delta(S,T)$.
Under this parameterization, the standard assumption is that $\delta(a,a)=0$ and $\delta(a,b)\ge 1$ for $a\ne b$; otherwise, the parameter $k$ can be made meaningless by rescaling $\delta$.
For several classical similarity measures, this low-distance regime admits much faster algorithms: a sequence of results beginning with unweighted edit distance [Landau and Vishkin; 1988] and recently extending to weighted edit distance and tree edit distance [Das, Gilbert, Hajiaghayi, Kociumaka, and Saha; 2023] yields running times of $\Oh(n+\poly(k))$. For DTW with metric costs, the best known bound is $\Oh(nk)$ [Kuszmaul; 2019], leaving open whether the dependence on $n$ and $k$ can be improved when $k$ is low.

We resolve this question negatively. Assuming the Orthogonal Vectors Hypothesis, we show that computing DTW requires $n^{1-o(1)}\cdot k$ time even for the discrete mismatch-cost function $\delta$ that assigns cost $1$ to every mismatch.
Our lower bound applies to the whole spectrum of thresholds $k$ between constant and linear in $n$.

The hardness proof is based on a new reduction from Orthogonal Vectors in which vector coordinates are encoded in the lengths of equal-character runs.
The resulting instances are very structured:
collapsing runs to single characters reveals long substrings with short periods.
We complement the lower bound with an algorithm that runs in $\Ohtilde(n+\poly(k))$ time whenever, after collapsing runs in the inputs, every substring with period $\Oh(k)$ has length $\poly(k)$.

Finally, we extend this lower bound to DTW pattern matching, which asks whether any non-empty substring of a length-$n$ text has DTW distance at most $k$ from a length-$m$ pattern. We prove that the classic $\Oh(nm)$-time dynamic-programming algorithm is near-optimal under OVH, even when the threshold $k$ is only logarithmic in~$n$.
\end{abstract}

\setcounter{page}{0}
\newpage

\section{Introduction}

Dynamic Time Warping (DTW) is a classical similarity measure for sequences, particularly well suited to time series recorded at varying speeds.
It was introduced in the context of speech recognition~\cite{Ita75,SC78} and has since found applications in a variety of areas, including data mining~\cite{RCMB12,WMDTSK13}, bioinformatics~\cite{AC01}, chemometrics~\cite{TBC04}, music processing~\cite{Mul07}, motion analysis~\cite{SMA11}, handwriting recognition~\cite{DRO08}, and signature verification~\cite{MP99}.

Dynamic Time Warping is defined for strings over an alphabet $\Sigma$ equipped with a mismatch-cost function $\delta : \Sigma^2 \to \mathbb{R}_{\ge 0}$.
For two strings $X,Y\in \Sigma^n$ of the same length $n$, the weighted Hamming distance is $\Ham_\delta(X,Y) \coloneqq \sum_{i=1}^n \delta(X[i],Y[i])$.
For non-empty strings $S,T\in \Sigma^+$, the DTW distance is defined as the minimum weighted Hamming distance between same-length expansions (time warps) of $S$ and $T$, obtained by arbitrarily duplicating characters in $S$ and $T$, respectively:
\[
\DTW_\delta(S,T)\coloneqq
\min_{
S',T'\;:\; \text{same-length expansions of } S\text{ and }T\;}
\Ham_\delta(S',T').
\]
For the discrete metric (unit mismatch-cost function), we omit the subscript and write $\DTW(S,T)$; for example,
$\DTW(\mathtt{aab},\mathtt{aba})=\Ham(\mathtt{aabb},\mathtt{aaba})=1.$

\begin{problem}[DTW]{Dynamic Time Warping (decision version)}\label{prob:dtw}
\PInput{Non-empty strings $S,T\in\Sigma^{\le n}$, a threshold $k\in\R_{\ge 0}$, and oracle access to a mismatch-cost function $\delta:\Sigma^2\to\R_{\ge 0}$.}
\PTask{Decide whether $\DTW_\delta(S,T)\le k$.}
\end{problem}

More generally, one may ask to compute $\DTW_\delta(S,T)$ and output optimal expansions of $S$ and~$T$; without loss of generality, their common length is at most $|S|+|T|$.
When $S,T\in \Sigma^{\le n}$ and $\delta$ can be evaluated in constant time (as we assume throughout), this can be done via the standard $\Oh(n^2)$-time dynamic program already presented in~\cite{SC78}.
More recently, Abboud, Backurs, and Vassilevska Williams~\cite{ABV14} as well as Bringmann and Künnemann~\cite{BK15} showed that, under the Strong Exponential Time Hypothesis (SETH), DTW does not admit an $n^{2-\Omega(1)}$-time algorithm.
Their lower bounds already hold under the Orthogonal Vectors Hypothesis (OVH), which asserts quadratic hardness for finding an orthogonal pair across two sets of binary vectors (see \cref{hyp:ovh}).
They hold even for simple costs: the discrete metric over a constant-size alphabet~\cite{ABV14} and the absolute-difference metric over constantly many integers~\cite{BK15}.

Only a few settings are known in which these lower bounds can be circumvented.
These include binary strings~\cite{ABV14,kuszmaul2021binary}, strings that are compressible with respect to run-length encoding~\cite{FJRW22,BGMW24}, and the low-distance regime~\cite{DBLP:conf/icalp/Kuszmaul19}, which is the main focus of this work, where the running time may depend on $k$ and is optimized for $k\ll n$.
Kuszmaul~\cite{DBLP:conf/icalp/Kuszmaul19} gave an $\Oh(n+nk)$-time algorithm for metric costs \emph{normalized} so that $\delta(a,a)=0$ and $\delta(a,b)\ge1$ for $a\ne b$, and used it to approximate DTW.

This $\Oh(n+nk)$ bound has not been improved even for the discrete metric despite the fact that the closely related edit distance~\cite{Lev65} has admitted a faster $\Oh(n+k^2)$-time algorithm for more than 35 years~\cite{LV88}.
More recently, the $\Oh(n+nk)$ bound for edit distance~\cite{Ukkonen85} has been improved for \emph{weighted} edit distance under normalized cost functions, first to $\Oh(n+k^5)$~\cite{DGHKS23}, then to $\Otild(n + \sqrt{nk^3})$~\cite{CKW23},\footnote{The $\Otild(\cdot)$ notation hides polylogarithmic factors in the input size, which is $n$ in the case of \ref{prob:dtw} and edit distance.} and, for integer weights bounded by a constant, to $\Otild(n+k^2)$~\cite{GK25}. These results raised hopes for analogous running times for DTW.
Our first main result rules out this possibility.

\begin{restatable}{mtheoremq}{mthmlb}\label{thm:lb-dtw}
Assuming the Orthogonal Vectors Hypothesis (\cref{hyp:ovh}), the following holds for every constant $0<\eps<1$.
Let $(k_n)_{n=1}^\infty$ be an integer sequence whose entries $k_n$ are computable in $\Oh(n)$ time and satisfy $1\le k_n \le \tfrac12 n$ for all sufficiently large $n$.
No deterministic algorithm for \cref{prob:dtw} can solve all instances with $|S|=|T|=n$ and $k=k_n$ in time $\Oh(n^{1-\eps}\cdot k_n)$.

The lower bound already holds for the unit mismatch-cost function $\delta$ over an alphabet of constant size.
Under randomized OVH, the same statement holds for bounded-error randomized algorithms.
\end{restatable}

In short, we show that $\Oh(n+nk)$ time is near-optimal even for the discrete (unit-cost) metric.
Our formal statement is stronger than what is usually proved in multivariate fine-grained complexity (see~\cite{BK18}): it rules out significant improvements for every threshold $k$ given as an efficiently computable function of $n$, rather than only for thresholds of the form $k\in \Theta(n^\kappa)$ with a fixed exponent $0 \le \kappa \le 1$.

\subparagraph*{Pattern Matching for Dynamic Time Warping}

Another setting where the low-distance regime is easier for edit distance is the pattern matching version of the problem.
Here, the task is to decide whether a pattern of length $m$ is similar to any non-empty substring of a text of length $n$. More generally, one may ask to locate all such substrings.
Dynamic Time Warping has been studied in this setting both in theoretical~\cite{Sakai2022,GDPS22} and applied contexts~\cite{KKS05,PKCP01,Wu2019}, under various names, including \emph{subsequence matching}, even though only substrings (contiguous subsequences) are sought.

\begin{problem}[PMDTW]{Pattern Matching for Dynamic Time Warping (decision version)}\label{prob:pmdtw}
\PInput{A non-empty pattern $P\in\Sigma^m$, a non-empty text $T\in\Sigma^n$, a threshold $k\in\R_{\ge 0}$, and oracle access to a mismatch-cost function $\delta:\Sigma^2\to\R_{\ge 0}$.}
\PTask{Decide whether there exists a non-empty substring $T\fragment{t}{t'}$ such that $\DTW_\delta(P,T\fragment{t}{t'})\le k$.}
\end{problem}

\Cref{prob:pmdtw} can be solved in $\Oh(nm)$ time by a straightforward adaptation of the textbook dynamic-programming solution to \cref{prob:dtw}, and the results of~\cite{ABV14,BK15} can be extended to show that no significantly faster algorithms exist under OVH.
Under the discrete metric, the zero-distance case is solvable in $\Oh(n+m)$ time, and the same bound is known for $k=1$~\cite{GDPS22}.

In contrast, the analogous problem for unit-cost edit distance admits the classical $\Oh(n+nk)$-time algorithm~\cite{LV89} and can even be solved in $\Oh(n+\ceil{n/m}k^4)$~\cite{CH02} or $\Otild(n+\ceil{n/m}k^{3.5})$~\cite{CKW22} time.
Recently, the first two running times have been recovered up to polylogarithmic factors for normalized weights; the third running time is recovered for integer metric weights bounded by a constant~\cite{CKW25}.

Our second main result shows that no analogous improvement is possible for DTW: the textbook $\Oh(nm)$ running time remains essentially optimal even when $k$ is logarithmic in the text length.

\begin{restatable}{mtheoremq}{mthmlbpm}\label{thm:lb-pmdtw}
Assuming the Orthogonal Vectors Hypothesis (\cref{hyp:ovh}), for every $0<\eps<1$, there exists $c_\eps>0$ such that the following holds.
Let $(m_n)_{n=1}^\infty$ and $(k_n)_{n=1}^\infty$ be integer sequences computable in $\Oh(n)$ time and satisfying $c_\eps \log n \le k_n \le \tfrac12 m_n \le \tfrac12 n$ for all sufficiently large $n$.
No deterministic algorithm for the \ref{prob:pmdtw} problem can solve all instances with $|P|=m_n$, $|T|=n$, and $k=k_n$ in time $\Oh(n^{1-\eps}\cdot m_n)$.

The lower bound already holds for the unit-cost function $\delta$ over an alphabet of constant size.
Under randomized OVH, the same statement holds for bounded-error randomized algorithms.
\end{restatable}

\subparagraph*{What Makes Dynamic Time Warping Hard?}
The instances we construct to prove \cref{thm:lb-dtw,thm:lb-pmdtw} not only establish hardness but also reveal a concrete structural obstruction for computing DTW fast.

To describe this obstruction, we use the \emph{shrunk} version $\shrink(X)$ of a string $X\in \Sigma^*$, obtained by collapsing each maximal run of consecutive equal characters into a single copy of that character; for example, $\shrink(\mathtt{aaaabba})=\mathtt{aba}$.
This notion is relevant for DTW because, under the discrete metric, $\DTW(X,Y)=0$ holds for strings $X,Y\in \Sigma^+$ if and only if $\shrink(X)=\shrink(Y)$.
We also use the notion of the \emph{period} $\per(X)$ of a string $X\in \Sigma^+$, defined as the smallest integer $p\in \fragment{1}{|X|}$ such that $X[i]=X[i+p]$ for all $i\in \fragment{1}{|X|-p}$; for example, $\per(\mathtt{abaabaabaa})=3$ because $\mathtt{abaabaabaa}=\mathtt{aba}\cdot \mathtt{aba}\cdot \mathtt{aba}\cdot \mathtt{a}$.

An interesting feature of the strings constructed to prove \cref{thm:lb-dtw,thm:lb-pmdtw} is that their shrunk versions contain very long substrings with very small periods. Concretely, each of the constructed length-$n$ strings contains, after shrinking, a substring of length $\Omega(n/\log n)$ and period $\Oh(\log n)$.

Motivated by this phenomenon, for an integer $p\ge1$, we define $L_p(X)$ to be the length of the longest substring of $\shrink(X)$ with period at most $p$, and set $L_0(X)\coloneqq0$.
We investigate how the complexity of \cref{prob:dtw} depends on this measure.
For a real threshold $k\ge0$, set $p\coloneqq24\ceil{k}$.
Our algorithm solves \cref{prob:dtw} in $\Otild(n+L_p(S)\cdot k^2)$ time for every normalized mismatch-cost function $\delta$.

\begin{restatable}{mtheoremq}{MainUB}\label{thm:main-UB}
There is an algorithm that, given an instance of \cref{prob:dtw} with a normalized mismatch-cost function and a real threshold $k\ge0$, solves it in time $\Otild(n+L_p(S)\cdot k^2)$, where $p\coloneqq24\ceil{k}$.
In case of a \Yes answer, the algorithm outputs an optimal alignment (that is, optimal expansions of $S$ and $T$).
\end{restatable}

In particular, inputs with $L_p(S)=\poly(k)$ admit an $\Otild(n+\poly(k))$-time algorithm.

\subparagraph*{Computational Model}
We use the word RAM with $w=\Theta(\log n)$-bit words, storing each input character in a separate word and assuming constant-time oracle access to the mismatch-cost function.
Our algorithm supports $\Oh(w)$-bit integer weights and extends to real weights in the real RAM model~\cite{EHM24}; precise conventions appear in \cref{sec:computational-model}.
Our deterministic lower bounds assume OVH, and their bounded-error randomized counterparts assume randomized OVH, defined alongside \cref{hyp:ovh}.


\section{Technical Overview}

In this section, we provide a high-level description of the ideas and techniques behind our results.

We start with a couple of standard definitions to set the language.
As discussed in the introduction, we make the routine assumption in the low-distance regime that the mismatch-cost function $\delta : \Sigma^2\to \mathbb{R}_{\ge 0}$ is \emph{normalized} so that $\delta(a,a)=0$ and $\delta(a,b)\ge 1$ holds for all distinct characters $a,b\in \Sigma$.
All our lower bounds are valid already for the \emph{discrete metric} $\delta$, where $\delta(a,a)=0$ and $\delta(a,b)=1$ for all distinct characters $a,b\in \Sigma$.

While we defined $\DTW_\delta(S,T)$ as the minimum weighted Hamming distance between same-length \emph{expansions} (time warps) of $S$ and $T$, it is often more convenient to interpret this value as the distance from the top-left corner to the bottom-right corner of a node-weighted\footnote{In comparison, alignment graphs constructed for edit distance and longest common subsequence are edge-weighted.} \emph{alignment graph} $\AG_\delta(S,T)$.
The vertex set of this graph is $[|S|]\times [|T|]$, where $[r] \coloneqq \{1,\ldots,r\}$ for a positive integer $r$, and the weight (cost) of the vertex $(i,j)$ is $\delta(S[i],T[j])$.
The graph has directed edges of three types: vertical edges $(i,j)\to (i+1,j)$, horizontal edges $(i,j)\to (i,j+1)$, and diagonal edges $(i,j)\to (i+1,j+1)$, each present whenever both endpoints are vertices of $\AG_\delta(S,T)$.
In this interpretation, $\DTW_{\delta}(S,T)$ is the minimum cost of a path from $(1,1)$ to $(|S|,|T|)$; see \cref{fig:ag} for an illustration on a concrete example.
We write $\DTW(S,T)$ and $\AG(S,T)$ without subscripts only for unit costs.

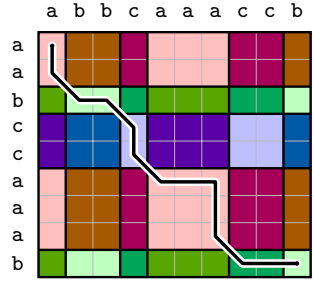
\begin{figure}[htb]
\begin{center}
\begin{tikzpicture}[x=0.36cm,y=0.36cm]
    \begin{scope}[yscale=-1]
        \fill[red!25!white]       (0,0) rectangle (1,2);
        \fill[red!65!green]       (1,0) rectangle (3,2);
        \fill[red!65!blue]        (3,0) rectangle (4,2);
        \fill[red!25!white]       (4,0) rectangle (7,2);
        \fill[red!65!blue]        (7,0) rectangle (9,2);
        \fill[red!65!green]       (9,0) rectangle (10,2);

        \fill[green!65!red]       (0,2) rectangle (1,3);
        \fill[green!25!white]     (1,2) rectangle (3,3);
        \fill[green!65!blue]      (3,2) rectangle (4,3);
        \fill[green!65!red]       (4,2) rectangle (7,3);
        \fill[green!65!blue]      (7,2) rectangle (9,3);
        \fill[green!25!white]     (9,2) rectangle (10,3);

        \fill[blue!65!red]        (0,3) rectangle (1,5);
        \fill[blue!65!green]      (1,3) rectangle (3,5);
        \fill[blue!25!white]      (3,3) rectangle (4,5);
        \fill[blue!65!red]        (4,3) rectangle (7,5);
        \fill[blue!25!white]      (7,3) rectangle (9,5);
        \fill[blue!65!green]      (9,3) rectangle (10,5);

        \fill[red!25!white]       (0,5) rectangle (1,8);
        \fill[red!65!green]       (1,5) rectangle (3,8);
        \fill[red!65!blue]        (3,5) rectangle (4,8);
        \fill[red!25!white]       (4,5) rectangle (7,8);
        \fill[red!65!blue]        (7,5) rectangle (9,8);
        \fill[red!65!green]       (9,5) rectangle (10,8);

        \fill[green!65!red]       (0,8) rectangle (1,9);
        \fill[green!25!white]     (1,8) rectangle (3,9);
        \fill[green!65!blue]      (3,8) rectangle (4,9);
        \fill[green!65!red]       (4,8) rectangle (7,9);
        \fill[green!65!blue]      (7,8) rectangle (9,9);
        \fill[green!25!white]     (9,8) rectangle (10,9);

        \draw[gray!55, line width=0.2pt,step=1] (0,0) grid (10,9);

        \draw[line width=0.9pt] (0,0) rectangle (10,9);
        \foreach \x in {1,3,4,7,9} {
            \draw[line width=0.9pt] (\x,0) -- (\x,9);
        }
        \foreach \y in {2,3,5,8} {
            \draw[line width=0.9pt] (0,\y) -- (10,\y);
        }

        \draw[white, line width=3.2pt, line cap=round, line join=round]
            (0.5,0.5) -- (0.5,1.5) -- (1.5,2.5) -- (2.5,2.5) -- (3.5,3.5)
            -- (3.5,4.5) -- (4.5,5.5) -- (5.5,5.5) -- (6.5,5.5) -- (6.5,6.5)
            -- (6.5,7.5) -- (7.5,8.5) -- (8.5,8.5) -- (9.5,8.5);
        \draw[black, line width=1.3pt, line cap=round, line join=round]
            (0.5,0.5) -- (0.5,1.5) -- (1.5,2.5) -- (2.5,2.5) -- (3.5,3.5)
            -- (3.5,4.5) -- (4.5,5.5) -- (5.5,5.5) -- (6.5,5.5) -- (6.5,6.5)
            -- (6.5,7.5) -- (7.5,8.5) -- (8.5,8.5) -- (9.5,8.5);
        \fill[black] (0.5,0.5) circle (0.08);
        \fill[black] (9.5,8.5) circle (0.08);
    \end{scope}


    \foreach \x/\lab in {0.5/$\ca$,1.5/$\cb$,2.5/$\cb$,3.5/$\cc$,4.5/$\ca$,5.5/$\ca$,6.5/$\ca$,7.5/$\cc$,8.5/$\cc$,9.5/$\cb$} {
        \node[above=2pt] at (\x,0) {\footnotesize \lab};
    }
    \foreach \y/\lab in {-0.5/$\ca$,-1.5/$\ca$,-2.5/$\cb$,-3.5/$\cc$,-4.5/$\cc$,-5.5/$\ca$,-6.5/$\ca$,-7.5/$\ca$,-8.5/$\cb$} {
        \node[left=2pt] at (0,\y) {\footnotesize \lab};
    }

\end{tikzpicture}
\end{center}
\caption{Alignment graph $\AG_\delta(S,T)$ for strings $S\coloneqq\ca\ca\cb\cc\cc\ca\ca\ca\cb$ (corresponding to rows from top to bottom) and $T\coloneqq\ca\cb\cb\cc\ca\ca\ca\cc\cc\cb$ (corresponding to columns from left to right).  Each grid cell represents one vertex and is colored according to the ordered pair $(S[i],T[j])$: light cells are matches of cost $0$, dark cells are mismatches of cost at least $1$.
Light gray grid lines mark cell boundaries, with block boundaries drawn in a darker shade of gray. The black polyline is an optimal alignment for the discrete metric; its total cost is $2$, accrued for aligning $S\fragment{9}{9}=\cb$ with $T\fragment{8}{9}=\cc\cc$.}
\label{fig:ag}
\end{figure}

Dynamic time warping is closely related to \emph{run-length} encoding, where each maximal run of consecutive equal characters is represented by the character itself and the number of copies. Formally, $\rle(X) \coloneqq (a_1,e_1)(a_2,e_2)\cdots (a_r,e_r)$ if and only if $X = a_1^{e_1}a_2^{e_2}\cdots a_r^{e_r}$ holds for some exponents $e_i \in \Z_{\ge 1}$ with $i\in \fragment{1}{r}$ and characters $a_i\in\Sigma$ with $a_i\ne a_{i-1}$ for $i\in \fragment{2}{r}$.
For example, $\rle(\mathtt{aaabba})=\rle(\dchara^3 \dcharb^2 \dchara^1)=(\dchara,3)(\dcharb,2)(\dchara,1)$.

For $a\in \Sigma$, we call $X\fragment{x}{x'}$ an \emph{$a$-run} if $X[\hat{x}]=a$ for all $\hat{x}\in \fragment{x}{x'}$ and neither $X[x-1]=a$ nor $X[x'+1]=a$, possibly because these positions are out of bounds.
For every $a$-run $S\fragment{s}{s'}$ in $S$ and $b$-run $T\fragment{t}{t'}$ in $T$, we refer to $\fragment{s}{s'}\times \fragment{t}{t'}$ as a \emph{block} in $\AG_\delta(S,T)$.
All vertices in the block have cost $\delta(a,b)$.

\subsection{Lower Bound for \ref{prob:dtw} (Section~\ref{sec:lb})}
Recall that, in the Orthogonal Vectors (\textsf{OV}) problem, the input consists of two sets $X,Y\subseteq \{0,1\}^d$ of $d$-dimensional binary vectors, and the task is to decide if there are vectors $\vec{x}\in X$ and $\vec{y}\in Y$ whose inner product $\vec{x}\cdot \vec{y} \coloneqq \sum_{i=1}^d \vec{x}[i]\cdot \vec{y}[i]$ is zero (this inner product is computed in $\Z$, not in $\Z_2$).
Although OVH concerns two sets of equal size, our reduction allows $X$ and $Y$ to have different sizes.

Our lower bound for \cref{prob:dtw} (in \cref{thm:lb-dtw}) primarily relies on the following reduction:%

\begin{restatable}{theorem}{thmovdtw}\label{thm:ov-to-dtw}
Given an instance of $\OV$ with non-empty sets $X,Y\subseteq\{0,1\}^d$ for an integer $d\ge1$, one can in $\Oh((|X|+|Y|)d^2)$ time construct an equivalent instance $(S,T,k)$ of the unit-cost \ref{prob:dtw} problem over an alphabet of size 5, with $|S|,|T|=\Oh((|X|+|Y|)d^2)$ and an integer threshold $k=\Oh(|Y|d)$.
\end{restatable}

\subparagraph*{Encoding Vector Orthogonality}
The starting point of our reduction is the observation that, for $x,y\in \Z_{\ge 0}$, we have
$\DTW(\dchara^{1+x}\dcharb,\dcharb^{1+y}\dchara) = 2 + \min(x,y)$.
When $x,y\in \{0,1\}$, this corresponds to $2+x\cdot y$ and encodes the inner product of one-dimensional vectors.
To generalize this construction to $d$ dimensions, we interleave the $\dchara$- and $\dcharb$-runs with long $\dchare$- and $\dcharf$-runs.
More precisely, we use the following \emph{coordinate gadgets}:
\[
E_X(x)\coloneqq\mathtt e^{4d}\mathtt a^{1+x}\mathtt b\mathtt f^{4d},
\qquad\text{and}\qquad
E_Y(y)\coloneqq\mathtt e^{4d}\mathtt b^{1+y}\mathtt a\mathtt f^{4d}.
\]
Concatenating these coordinate gadgets gives \emph{vector gadgets} $G_X(\vec u)$ and $G_Y(\vec v)$; see \cref{fig:vector-gadget-alignment} for an example.

\begin{figure}[htbp]
\begin{center}
    \vspace{-.2cm}
\begin{tikzpicture}[x=0.18cm,y=0.18cm]
    \begin{scope}[yscale=-1]

        \fill[red!25!white]       (0,0) rectangle (8,8);
        \fill[red!65!blue]        (8,0) rectangle (9,8);
        \fill[red!65!green]       (9,0) rectangle (10,8);
        \fill[red!65!orange]      (10,0) rectangle (18,8);
        \fill[red!25!white]       (18,0) rectangle (26,8);
        \fill[red!65!blue]        (26,0) rectangle (28,8);
        \fill[red!65!green]       (28,0) rectangle (29,8);
        \fill[red!65!orange]      (29,0) rectangle (37,8);

        \fill[green!65!red]       (0,8) rectangle (8,10);
        \fill[green!65!blue]      (8,8) rectangle (9,10);
        \fill[green!25!white]     (9,8) rectangle (10,10);
        \fill[green!65!orange]    (10,8) rectangle (18,10);
        \fill[green!65!red]       (18,8) rectangle (26,10);
        \fill[green!65!blue]      (26,8) rectangle (28,10);
        \fill[green!25!white]     (28,8) rectangle (29,10);
        \fill[green!65!orange]    (29,8) rectangle (37,10);

        \fill[blue!65!red]        (0,10) rectangle (8,11);
        \fill[blue!25!white]      (8,10) rectangle (9,11);
        \fill[blue!65!green]      (9,10) rectangle (10,11);
        \fill[blue!65!orange]     (10,10) rectangle (18,11);
        \fill[blue!65!red]        (18,10) rectangle (26,11);
        \fill[blue!25!white]      (26,10) rectangle (28,11);
        \fill[blue!65!green]      (28,10) rectangle (29,11);
        \fill[blue!65!orange]     (29,10) rectangle (37,11);

        \fill[orange!65!red]      (0,11) rectangle (8,19);
        \fill[orange!65!blue]     (8,11) rectangle (9,19);
        \fill[orange!65!green]    (9,11) rectangle (10,19);
        \fill[orange!25!white]    (10,11) rectangle (18,19);
        \fill[orange!65!red]      (18,11) rectangle (26,19);
        \fill[orange!65!blue]     (26,11) rectangle (28,19);
        \fill[orange!65!green]    (28,11) rectangle (29,19);
        \fill[orange!25!white]    (29,11) rectangle (37,19);

        \fill[red!25!white]       (0,19) rectangle (8,27);
        \fill[red!65!blue]        (8,19) rectangle (9,27);
        \fill[red!65!green]       (9,19) rectangle (10,27);
        \fill[red!65!orange]      (10,19) rectangle (18,27);
        \fill[red!25!white]       (18,19) rectangle (26,27);
        \fill[red!65!blue]        (26,19) rectangle (28,27);
        \fill[red!65!green]       (28,19) rectangle (29,27);
        \fill[red!65!orange]      (29,19) rectangle (37,27);

        \fill[green!65!red]       (0,27) rectangle (8,29);
        \fill[green!65!blue]      (8,27) rectangle (9,29);
        \fill[green!25!white]     (9,27) rectangle (10,29);
        \fill[green!65!orange]    (10,27) rectangle (18,29);
        \fill[green!65!red]       (18,27) rectangle (26,29);
        \fill[green!65!blue]      (26,27) rectangle (28,29);
        \fill[green!25!white]     (28,27) rectangle (29,29);
        \fill[green!65!orange]    (29,27) rectangle (37,29);

        \fill[blue!65!red]        (0,29) rectangle (8,30);
        \fill[blue!25!white]      (8,29) rectangle (9,30);
        \fill[blue!65!green]      (9,29) rectangle (10,30);
        \fill[blue!65!orange]     (10,29) rectangle (18,30);
        \fill[blue!65!red]        (18,29) rectangle (26,30);
        \fill[blue!25!white]      (26,29) rectangle (28,30);
        \fill[blue!65!green]      (28,29) rectangle (29,30);
        \fill[blue!65!orange]     (29,29) rectangle (37,30);

        \fill[orange!65!red]      (0,30) rectangle (8,38);
        \fill[orange!65!blue]     (8,30) rectangle (9,38);
        \fill[orange!65!green]    (9,30) rectangle (10,38);
        \fill[orange!25!white]    (10,30) rectangle (18,38);
        \fill[orange!65!red]      (18,30) rectangle (26,38);
        \fill[orange!65!blue]     (26,30) rectangle (28,38);
        \fill[orange!65!green]    (28,30) rectangle (29,38);
        \fill[orange!25!white]    (29,30) rectangle (37,38);

        \draw[gray!45, line width=0.18pt, step=1] (0,0) grid (37,38);

        \draw[line width=0.75pt] (0,0) rectangle (37,38);
        \foreach \x in {8,9,10,18,26,28,29} {
            \draw[line width=0.75pt] (\x,0) -- (\x,38);
        }
        \foreach \y in {8,10,11,19,27,29,30} {
            \draw[line width=0.75pt] (0,\y) -- (37,\y);
        }

        \draw[line width=1.15pt] (18,0) -- (18,38);
        \draw[line width=1.15pt] (0,19) -- (37,19);

        \draw[white, line width=2.8pt, line cap=round, line join=round]
            (0.5,0.5) -- (9.5,9.5) -- (9.5,10.5) -- (36.5,37.5);
        \draw[black, line width=1.15pt, line cap=round, line join=round]
            (0.5,0.5) -- (9.5,9.5) -- (9.5,10.5) -- (36.5,37.5);
        \fill[black] (0.5,0.5) circle (0.12);
        \fill[black] (36.5,37.5) circle (0.12);
    \end{scope}

    \foreach \x in {0.5,1.5,2.5,3.5,4.5,5.5,6.5,7.5,18.5,19.5,20.5,21.5,22.5,23.5,24.5,25.5} {
        \node[font=\tiny, inner sep=0.5pt, anchor=south] at (\x,0.35) {$\mathtt e$};
    }
    \foreach \x in {8.5,26.5,27.5} {
        \node[font=\tiny, inner sep=0.5pt, anchor=south] at (\x,0.35) {$\mathtt b$};
    }
    \foreach \x in {9.5,28.5} {
        \node[font=\tiny, inner sep=0.5pt, anchor=south] at (\x,0.35) {$\mathtt a$};
    }
    \foreach \x in {10.5,11.5,12.5,13.5,14.5,15.5,16.5,17.5,29.5,30.5,31.5,32.5,33.5,34.5,35.5,36.5} {
        \node[font=\tiny, inner sep=0.5pt, anchor=south] at (\x,0.35) {$\mathtt f$};
    }

    \foreach \y in {-0.5,-1.5,-2.5,-3.5,-4.5,-5.5,-6.5,-7.5,-19.5,-20.5,-21.5,-22.5,-23.5,-24.5,-25.5,-26.5} {
        \node[font=\tiny, anchor=east] at (-0,\y) {$\mathtt e$};
    }
    \foreach \y in {-8.5,-9.5,-27.5,-28.5} {
        \node[font=\tiny, anchor=east] at (-0,\y) {$\mathtt a$};
    }
    \foreach \y in {-10.5,-29.5} {
        \node[font=\tiny, anchor=east] at (-0,\y) {$\mathtt b$};
    }
    \foreach \y in {-11.5,-12.5,-13.5,-14.5,-15.5,-16.5,-17.5,-18.5,-30.5,-31.5,-32.5,-33.5,-34.5,-35.5,-36.5,-37.5} {
        \node[font=\tiny, anchor=east] at (-0,\y) {$\mathtt f$};
    }
\end{tikzpicture}
\captionof{figure}{Alignment graph $\AG(G_X(\vec u),G_Y(\vec v))$ for $d=2$, $\vec u\coloneqq(1,1)$, and $\vec v\coloneqq(0,1)$, shown in the same cell-based style as \cref{fig:ag}. The black polyline is an optimal alignment for the unit-cost function, and its total cost is $5=2d+\vec u\cdot \vec v$.}
\label{fig:vector-gadget-alignment}
\end{center}
\end{figure}
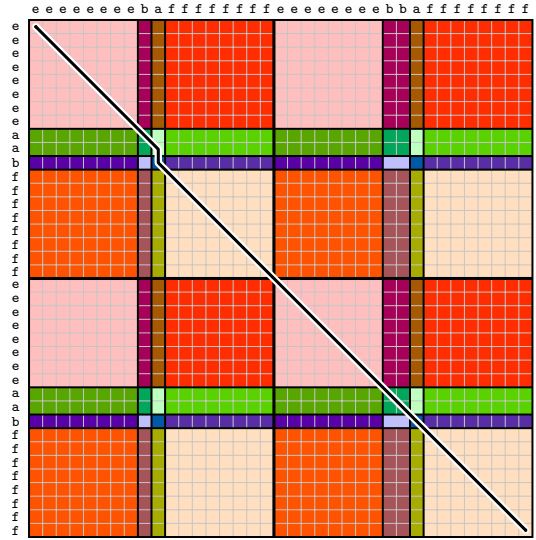

As proved in \cref{lem:dist_vect_gadget}, we have $\DTW(G_X(\vec u),G_Y(\vec v))=2d + \vec{u}\cdot \vec{v}$.
This is already stronger than what is typical in edit-distance or LCS reductions: the gadget does not merely separate orthogonal from non-orthogonal pairs but recovers the inner product exactly.

The upper bound  $\DTW(G_X(\vec u),G_Y(\vec v))\le 2d + \vec{u}\cdot \vec{v}$ is obtained by aligning the gadgets coordinate by coordinate.
For the lower bound, consider any alignment of cost below $2d+\vec u\cdot \vec v$.
Since the length of each $\dchare$-run and each $\dcharf$-run exceeds the cost of the considered alignment, the alignment must visit a cost-$0$ block for every $\dchare$-run in $G_X(\vec{u})$ and $G_Y(\vec{v})$, and likewise for every $\dcharf$-run.
The alignment also cannot visit two such blocks for the same $\mathtt e$- or $\mathtt f$-run, because the gaps between consecutive $\mathtt e$-runs, and likewise between consecutive $\mathtt f$-runs, in the other string are already too large.
By monotonicity, the only possible visited separator blocks are therefore the diagonal ones.
This forces the alignment to pass through the gadget coordinate by coordinate, and the total cost is at least the sum of the $d$ middle contributions, namely $2+\vec u[i]\cdot \vec v[i]$ for each coordinate $i$.

\subparagraph*{From One Pair to Many Pairs}
Now let
\[
X\coloneqq(\vec u_0,\ldots,\vec u_{n-1}),\qquad Y\coloneqq(\vec v_0,\ldots,\vec v_{m-1}),
\]
and suppose we want to decide whether there is an orthogonal pair $(\vec u_i,\vec v_j)$.
We already know how to test vector orthogonality using $\AG(G_X(\vec u_i),G_Y(\vec v_j))$.
The main task is to embed all these gadgets into one large alignment graph and to connect them by a network of cheap paths, including long cost-$0$ \emph{freeways}, so that an optimal path can benefit from the orthogonality of any pair $(\vec u_i,\vec v_j)$ while incurring a total cost of $\Oh(md)$.

The only way to introduce freeways is to use strings with DTW distance 0, i.e., with equal \emph{shrunk versions}, in which all run exponents are reduced to $1$.
 For this, we introduce \emph{narrow gadgets}
\[
\bot_X \coloneqq (\dchare\dchara\dcharb\dcharf)^d=\shrink(G_X(\vec u))
\qquad\text{and}\qquad
\bot_Y \coloneqq (\dchare\dcharb\dchara\dcharf)^d=\shrink(G_Y(\vec v)),
\]
which are the shrunk versions of the vector gadgets.
We also use thick \emph{walls} $\wall \coloneqq \dcharw^{12d}$.
The global strings alternate vector gadgets, narrow gadgets, and walls.
Using $\bigodot$ to denote concatenation, we define
\[\hat X\coloneqq\wall\,\bigodot_{i=0}^{n-1} \left(G_X(\vec u_i)\,\wall\,\bot_Y\,\wall\right)\ \quad\text{and}\quad \hat Y\coloneqq\wall\,\bigodot_{i=0}^{n+m-1}\left(\bot_X\,\wall\,G_Y(\vec v_{i \bmod m})\,\wall\,\right).\]

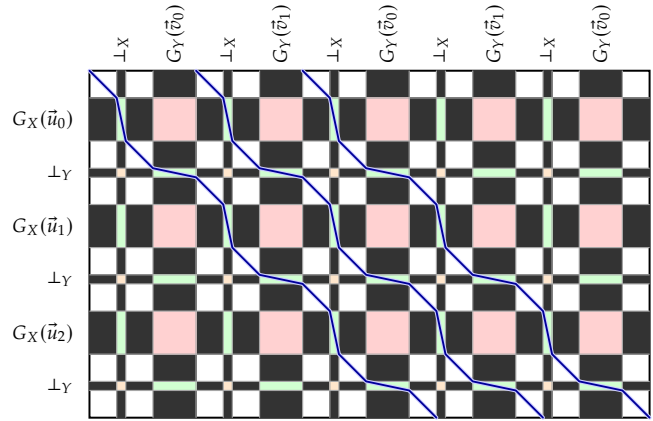
\begin{figure}[ht]
\begin{center}
\begin{tikzpicture}[x=0.015cm,y=0.015cm]
    \begin{scope}[yscale=-1]
        \draw[line width=0.75pt] (0,0) rectangle (494,306);

        \foreach \y/\h in {0/24,62/24,94/24,156/24,188/24,250/24,282/24} {
            \foreach \x/\w in {24/8,56/38,118/8,150/38,212/8,244/38,306/8,338/38,400/8,432/38} {
                \fill[black!80] (\x,\y) rectangle ++(\w,\h);
            }
        }
        \foreach \y/\h in {24/38,86/8,118/38,180/8,212/38,274/8} {
            \foreach \x/\w in {0/24,32/24,94/24,126/24,188/24,220/24,282/24,314/24,376/24,408/24,470/24} {
                \fill[black!80] (\x,\y) rectangle ++(\w,\h);
            }
        }

        \foreach \y in {24,118,212} {
            \foreach \x in {24,118,212,306,400} {
                \fill[green!18] (\x,\y) rectangle ++(8,38);
            }
        }
        \foreach \y in {86,180,274} {
            \foreach \x in {56,150,244,338,432} {
                \fill[green!18] (\x,\y) rectangle ++(38,8);
            }
        }

        \foreach \y in {86,180,274} {
            \foreach \x in {24,118,212,306,400} {
                \fill[orange!20] (\x,\y) rectangle ++(8,8);
            }
        }

        \foreach \y in {24,118,212} {
            \foreach \x in {56,150,244,338,432} {
                \fill[red!18] (\x,\y) rectangle ++(38,38);
            }
        }

        \foreach \x in {24,32,56,94,118,126,150,188,212,220,244,282,306,314,338,376,400,408,432,470} {
            \draw[gray!70, line width=0.45pt] (\x,0) -- (\x,306);
        }

        \foreach \y in {24,62,86,94,118,156,180,188,212,250,274,282} {
            \draw[gray!70, line width=0.45pt] (0,\y) -- (494,\y);
        }

        \draw[blue!20, line width=1.5pt, line cap=round, line join=round]
            (0,0) -- (24,24) -- (32,62) -- (56,86) -- (94,94) -- (118,118)
            -- (126,156) -- (150,180) -- (188,188) -- (212,212) -- (220,250)
            -- (244,274) -- (282,282) -- (306,306);
        \draw[blue!60!black, line width=0.7pt, line cap=round, line join=round]
            (0,0) -- (24,24) -- (32,62) -- (56,86) -- (94,94) -- (118,118)
            -- (126,156) -- (150,180) -- (188,188) -- (212,212) -- (220,250)
            -- (244,274) -- (282,282) -- (306,306);

        \draw[blue!20, line width=1.5pt, line cap=round, line join=round]
            (94,0) -- (118,24) -- (126,62) -- (150,86) -- (188,94) -- (212,118)
            -- (220,156) -- (244,180) -- (282,188) -- (306,212) -- (314,250)
            -- (338,274) -- (376,282) -- (400,306);
        \draw[blue!60!black, line width=0.7pt, line cap=round, line join=round]
            (94,0) -- (118,24) -- (126,62) -- (150,86) -- (188,94) -- (212,118)
            -- (220,156) -- (244,180) -- (282,188) -- (306,212) -- (314,250)
            -- (338,274) -- (376,282) -- (400,306);

        \draw[blue!20, line width=1.5pt, line cap=round, line join=round]
            (188,0) -- (212,24) -- (220,62) -- (244,86) -- (282,94) -- (306,118)
            -- (314,156) -- (338,180) -- (376,188) -- (400,212) -- (408,250)
            -- (432,274) -- (470,282) -- (494,306);
        \draw[blue!60!black, line width=0.7pt, line cap=round, line join=round]
            (188,0) -- (212,24) -- (220,62) -- (244,86) -- (282,94) -- (306,118)
            -- (314,156) -- (338,180) -- (376,188) -- (400,212) -- (408,250)
            -- (432,274) -- (470,282) -- (494,306);
    \end{scope}

    \node[font=\scriptsize, rotate=90, anchor=west, inner sep=0.5pt] at (28,5) {$\bot_X$};
    \node[font=\scriptsize, rotate=90, anchor=west, inner sep=0.5pt] at (75,5) {$G_Y(\vec v_0)$};
    \node[font=\scriptsize, rotate=90, anchor=west, inner sep=0.5pt] at (122,5) {$\bot_X$};
    \node[font=\scriptsize, rotate=90, anchor=west, inner sep=0.5pt] at (169,5) {$G_Y(\vec v_1)$};
    \node[font=\scriptsize, rotate=90, anchor=west, inner sep=0.5pt] at (216,5) {$\bot_X$};
    \node[font=\scriptsize, rotate=90, anchor=west, inner sep=0.5pt] at (263,5) {$G_Y(\vec v_0)$};
    \node[font=\scriptsize, rotate=90, anchor=west, inner sep=0.5pt] at (310,5) {$\bot_X$};
    \node[font=\scriptsize, rotate=90, anchor=west, inner sep=0.5pt] at (357,5) {$G_Y(\vec v_1)$};
    \node[font=\scriptsize, rotate=90, anchor=west, inner sep=0.5pt] at (404,5) {$\bot_X$};
    \node[font=\scriptsize, rotate=90, anchor=west, inner sep=0.5pt] at (451,5) {$G_Y(\vec v_0)$};

    \node[font=\scriptsize, anchor=east] at (-6,-43) {$G_X(\vec u_0)$};
    \node[font=\scriptsize, anchor=east] at (-6,-90) {$\bot_Y$};
    \node[font=\scriptsize, anchor=east] at (-6,-137) {$G_X(\vec u_1)$};
    \node[font=\scriptsize, anchor=east] at (-6,-184) {$\bot_Y$};
    \node[font=\scriptsize, anchor=east] at (-6,-231) {$G_X(\vec u_2)$};
    \node[font=\scriptsize, anchor=east] at (-6,-278) {$\bot_Y$};

\end{tikzpicture}
\end{center}
\caption{Schematic picture of $\AG(\hat X,\hat Y)$ for $n=3$ and $m=2$. White regions are cost-$0$ wall--wall blocks, called portals, and dark gray regions have uniform cost $1$ (wall--gadget interactions). Light green regions correspond to vector--narrow and narrow--vector interactions, light orange regions to narrow--narrow interactions, and light red regions to vector--vector interactions. The three blue polylines are the freeways, i.e., the $m+1$ zero-cost paths stretching the entire height of the alignment graph.}
\label{fig:hatx-haty-schematic}
\end{figure}

After shrinking, both $\hat X$ and $\hat Y$ become prefixes of the same infinite string $(\dcharw (\dchare \dchara \dcharb \dcharf)^d \dcharw (\dchare \dcharb \dchara \dcharf)^d)^\infty$, which leads to multiple freeways spanning the entire height of the alignment graph $\AG(\hat X, \hat Y)$.
Formally, these freeways are zero-cost block diagonals, meaning sequences of cost-$0$ blocks that advance by one run in each string; see \cref{fig:hatx-haty-schematic}.
The top-left corner lies on the leftmost freeway, the bottom-right corner lies on the rightmost one, and the alignments must change freeways at least $m$ times.

\begin{figure}[hbt]
\begin{center}
\begin{tikzpicture}[x=0.04cm,y=0.04cm]
    \begin{scope}[yscale=-1]

        \draw[line width=0.75pt] (0,0) rectangle (150,86);

        \foreach \y/\h in {0/24,62/24} {
            \foreach \x/\w in {24/8,56/38,118/8} {
                \fill[black!80] (\x,\y) rectangle ++(\w,\h);
            }
        }
        \foreach \y/\h in {24/38} {
            \foreach \x/\w in {0/24,32/24,94/24,126/24} {
                \fill[black!80] (\x,\y) rectangle ++(\w,\h);
            }
        }

        \fill[green!18] (24,24) rectangle ++(8,38);
        \fill[green!18] (118,24) rectangle ++(8,38);

        \fill[red!18] (56,24) rectangle ++(38,38);

        \foreach \x in {24,32,56,94,118,126} {
            \draw[gray!70, line width=0.45pt] (\x,0) -- (\x,86);
        }
        \foreach \y in {24,62} {
            \draw[gray!70, line width=0.45pt] (0,\y) -- (150,\y);
        }

        \draw[blue!20, line width=1.5pt, line cap=round, line join=round]
            (0,0) -- (24,24) -- (32,62) -- (56,86);
        \draw[blue!60!black, line width=0.7pt, line cap=round, line join=round]
            (0,0) -- (24,24) -- (32,62) -- (56,86);

        \draw[blue!20, line width=1.5pt, line cap=round, line join=round]
            (94,0) -- (118,24) -- (126,62) -- (150,86);
        \draw[blue!60!black, line width=0.7pt, line cap=round, line join=round]
            (94,0) -- (118,24) -- (126,62) -- (150,86);

        \draw[white, line width=2.8pt, line cap=round, line join=round]
            (12,12) -- (44,12) -- (106,74) -- (138,74);
        \draw[red!70!black, line width=1.15pt, line cap=round, line join=round]
            (12,12) -- node[font=\scriptsize\bfseries, text=red!70!black,above]{\contour{white}{$4d$}} (44,12) -- node[font=\scriptsize\bfseries, text=red!70!black,above,rotate=-45]{$2d+\vec u\!\cdot\!\vec v$} (106,74) -- node[font=\scriptsize\bfseries, text=red!70!black,above]{\contour{white}{$4d$}} (138,74);

        \fill[red!70!black] (12,12) circle (1.2);
        \fill[red!70!black] (138,74) circle (1.2);
    \end{scope}

    \node[font=\scriptsize, anchor=south] at (12,2) {$\wall$};
    \node[font=\scriptsize, anchor=south] at (28,2) {$\bot_X$};
    \node[font=\scriptsize, anchor=south] at (44,2) {$\wall$};
    \node[font=\scriptsize, anchor=south] at (75,2) {$G_Y(\vec v)$};
    \node[font=\scriptsize, anchor=south] at (106,2) {$\wall$};
    \node[font=\scriptsize, anchor=south] at (122,2) {$\bot_X$};
    \node[font=\scriptsize, anchor=south] at (138,2) {$\wall$};

    \node[font=\scriptsize, anchor=east] at (-2,-12) {$\wall$};
    \node[font=\scriptsize, anchor=east] at (-2,-43) {$G_X(\vec u)$};
    \node[font=\scriptsize, anchor=east] at (-2,-74) {$\wall$};
\end{tikzpicture}
\end{center}
\caption{The canonical way to switch from one freeway to the next, both depicted by blue polylines. The red path leaves the left freeway, crosses a narrow gadget for cost $4d$, then crosses one vector-gadget pair diagonally for cost $2d+\vec u\cdot \vec v$, and finally crosses another narrow gadget for cost $4d$ before joining the next freeway.}
\label{fig:highway-switch-schematic}
\end{figure}
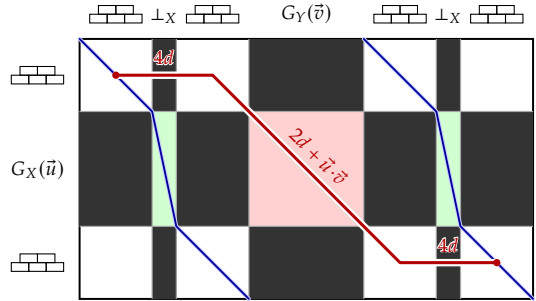

We call the block at the intersection of two $\dcharw$-runs a \emph{portal}.
A canonical way to switch from a given freeway to the next one is depicted in \cref{fig:highway-switch-schematic}:
it first crosses a narrow gadget for cost $4d$, then it traverses one pair $G_X(\vec u_i),G_Y(\vec v_j)$ diagonally for cost $2d+\vec u_i\cdot \vec v_j$, and finally it crosses another narrow gadget for another $4d$.
Hence, each switch can be executed with a cost of exactly
\[
10d+\vec u_i\cdot \vec v_j.
\]

To turn this local picture into a correct global reduction, we first preprocess the OV instance.
We add one extra coordinate that is $0$ on every vector of $X$ and $1$ on every vector of $Y$, and we prepend and append $m$ dummy vectors $(0,\ldots,0,1)$ to $X$.
These dummy vectors have inner product $1$ with every vector of $Y$.
We apply the construction to the modified lists, now writing $d$ for their dimension and $n$ for the length of $X$.
As a result, all boundary freeway changes cost exactly $10d+1$, while an orthogonal pair, if it exists, gives one interior freeway change of cost $10d$.
This yields an alignment of total cost $(10d+1)m-1$ in the \Yes case.

The main work is to show that, in a \No instance, no cheaper alignment exists.
We extract $m$ successive crossing subpaths between consecutive zero diagonals and analyze each separately.
Each crossing's endpoints can be extended along the zero diagonals to portals without increasing the cost.
We number the walls from $0$ in each string and identify a portal by the indices of its two walls.
For a crossing of cost at most $10d$, consecutive portal visits cannot skip a wall: doing so would mismatch all $12d$ copies of $\dcharw$ in that wall.
Thus, each wall index changes by at most one between consecutive portal visits.
This lets us isolate a part of the crossing that makes two horizontal moves between portals, with only diagonal moves between them.
Each horizontal move must cross a narrow gadget, costing $4d$, because crossing a vector gadget would exceed the budget.
The alternating gadget types then force an intervening vector--vector comparison, costing at least $2d+1$ in a \No instance.
These three costs already sum to $10d+1$, a contradiction.
Every crossing therefore costs at least $10d+1$, giving the required total of $(10d+1)m$.
This proves \cref{thm:ov-to-dtw}.

\subparagraph*{From \Cref{thm:ov-to-dtw} to \Cref{thm:lb-dtw}}
For every fixed $\eps>0$, the reduction above already rules out $\Oh(n^{1-\eps}k)$-time algorithms for many choices of the threshold $k$, but \cref{thm:lb-dtw} is stronger: it rules out this running time for \emph{every} efficiently computable threshold sequence $(k_n)_{n=1}^\infty$ with $1\le k_n\le n/2$.
To obtain this stronger statement, we combine the reduction with batching and padding.

Fix an OV instance with $|X|=|Y|=s$ and dimension $d=\Theta(\log s)$, and choose a target length $n=\Theta(sd^2)$.
Assume for contradiction that there is an algorithm running in time $\Oh(n^{1-\eps}k_n)$ on length-$n$ instances with threshold $k_n$.
A simple input-reading bound implies that such an algorithm must take $\Omega(n)$ time even on very simple binary instances.
This is consistent with our constraint only when $k_n=\Omega(n^\eps)$.
In particular, $k_n$ is large enough for us to group $Y$ into batches of size at most
$r = \Theta(\min\{s,k_n/d\})$.
Let $Y_1,\ldots,Y_t$ be the resulting partition, where $t=\Theta(\lceil s/r\rceil)=\Oh(1+sd/k_n)$.

For each batch $Y_i$, we apply \cref{thm:ov-to-dtw} to $(X,Y_i)$ and obtain an equivalent DTW instance $(S_i,T_i,\kappa_i)$ with $\kappa_i\le k_n/2$ and $|S_i|,|T_i|<n/4$.
We then pad both strings with fresh symbols in two stages.
First, we prepend a common fresh symbol enough times to bring both strings to length $n-(k_n-\kappa_i)$ without changing their distance.
Second, we prepend distinct fresh symbols to add exactly $k_n-\kappa_i$ to the DTW and bring both strings to length $n$.
Thus, each batch yields an equivalent length-$n$ instance $(\widehat S_i,\widehat T_i,k_n)$.
Running the hypothetical algorithm on all $t$ batches decides the original OV instance.

Up to logarithmic factors for error amplification, the total running time becomes
$
\Ohtilde\bigl(t\cdot n^{1-\eps}k_n\bigr)
=
\Ohtilde\!\left(n^{1-\eps}(n+sd)\right)
$.
Since $n=\Theta(sd^2)$ and $d=\Theta(\log s)$, this simplifies to $\Ohtilde(s^{2-\eps})$, contradicting the OV Hypothesis.
This completes the proof of \cref{thm:lb-dtw}.

\subparagraph*{Comparison with Previous Lower Bounds}
Even for $k=\Theta(n)$, our hard instances have a very different structure compared to those constructed by Abboud, Backurs, and Vassilevska Williams~\cite{ABV14} and Bringmann and K\"unnemann~\cite{BK15}.
In particular, we do not rely on the hard edit-distance instances used in~\cite{ABV14}; instead, we isolate the feature that makes DTW strictly harder.
We also do not follow the recipe of~\cite{BK15} (designed to work for many sequence similarity measures), where, in the structured alignments, \emph{all} vector gadgets in the shorter string are aligned against vector gadgets in the longer string.
Most importantly, the necessity to implement cost-$0$ freeways imposes a structural constraint that, on the one hand, limits our freedom but, on the other hand, leads to a particularly simple and elegant construction.

\subsection{Lower Bound for \ref{prob:pmdtw} (Section~\ref{sec:pmlb})}

Our lower bound for \cref{prob:pmdtw} (\cref{thm:lb-pmdtw}) relies on the following reduction from \textsf{OV}.
\begin{restatable}{theorem}{thmovpm}\label{thm:ov-to-pmdtw}
Given an instance of $\OV$ with non-empty sets $X,Y\subseteq\{0,1\}^d$ for an integer $d\ge1$, one can in $\Oh((|X|+|Y|)d^2)$ time construct an equivalent instance $(P,T,k)$ of the unit-cost \ref{prob:pmdtw} problem over an alphabet of size 7, with $|P|=\Oh(|Y|d^2)$, $|T|=\Oh((|X|+|Y|)d^2)$, and an integer threshold $k=\Oh(d)$.
In this instance, the first $k+1$ characters of $P$ are equal.
\end{restatable}

Compared to \cref{thm:ov-to-dtw}, the main difference is that the threshold is only $\Oh(d)$, which can be as small as logarithmic in $|T|$, while $|Y|$ now contributes to the pattern length rather than to the threshold.

The starting point for our reduction is the construction behind \cref{thm:ov-to-dtw}.
However, \ref{prob:pmdtw} asks for the minimum-cost path from anywhere in the top row to anywhere in the bottom row of the alignment graph.
Hence, if we used the DTW construction unchanged, then every freeway would already yield a cost-$0$ occurrence, so the resulting instance would always be a \Yes instance of \cref{prob:pmdtw}.

For the construction below, $n,m,d$ denote the list lengths and dimension after the preprocessing described below.
The pattern $\hat P$ alternates narrow gadgets $\bot_X$ with vector gadgets $G_Y(\vec v_j)$, while the text $\hat T$ alternates $G_X(\vec u_i)$ with $\bot_Y$; walls separate consecutive gadgets.
At an idealized level, we would like to force every affordable path to start at the top of one freeway and end at the bottom of an adjacent freeway.
Then, an orthogonal pair would yield one such path of cost $10d$, whereas in a \No instance every such path would cost at least $10d+1$.
Enforcing this exact behavior is unrealistic, but we can come close.

To implement this idea, we set $k\coloneqq10d+8$ and modify the wall gadgets by inserting two fresh symbols $\dol_1,\dol_2$ in the middle.
Specifically, walls are now sufficiently long $\mathtt w$-runs surrounding a middle part with fresh characters: the small wall $W_{\mathrm s}$ has middle $\dol_1\dol_2$, the large wall $W_{\mathrm L}$ has middle $\dol_1^{2k}\dol_2^{2k}$, and the reversed wall $W_{\mathrm{rev}}$ has middle $\dol_2^{2k}\dol_1$.
This way, $\DTW(W_{\mathrm{rev}},W_{\mathrm L})=2k+1$, $\DTW(W_{\mathrm{rev}},W_{\mathrm s})=2$, and $\DTW(W_{\mathrm s},W_{\mathrm L})=0$.
The wall placement is asymmetric; see \cref{fig:pmdtw-wall-layout}.

With walls indexed from $0$ in each string, every cheap occurrence can be adjusted, without increasing its cost, to start and end inside small text walls, whose indices are divisible by $6$.
When $m\equiv1\pmod3$, the two endpoint portals have distinct even \emph{wall offsets}, where the wall offset is the pattern wall index minus the text wall index.
This restriction leaves only one third of all subgraphs $\AG(G_Y(\vec v),G_X(\vec u))$ within reach.

\begin{figure}[t!]
\input{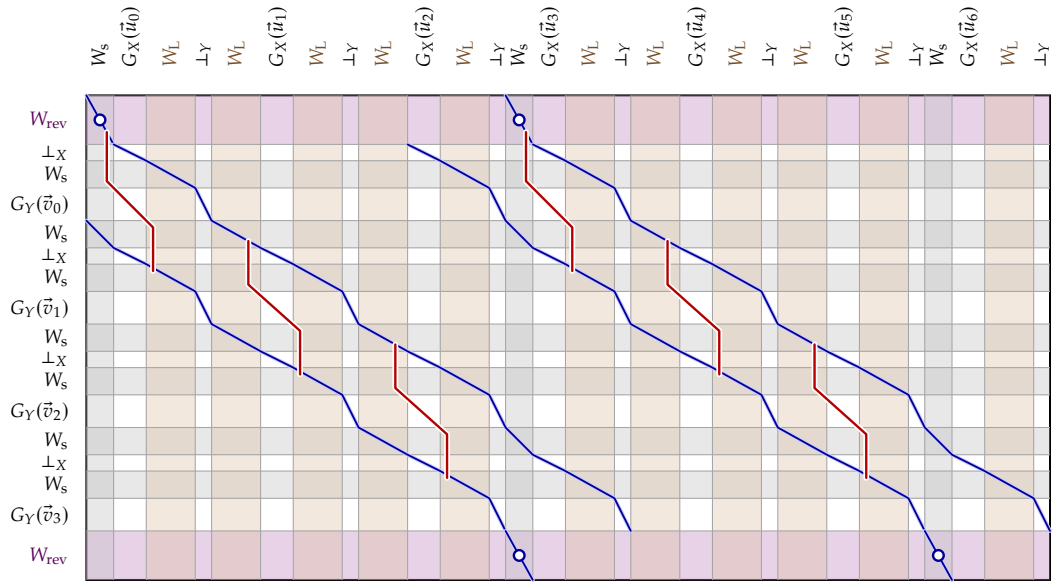}
\caption{
Schematic illustration of the wall layout in the \ref{prob:pmdtw} reduction, shown for $m=4$ and $n=7$. In the pattern $\hat P$, the boundary walls are reversed walls, and every internal wall is a small wall.
In the text $\hat T$, the wall immediately after each $G_X(\vec u_i)$ is always a large wall, while the wall immediately before $G_X(\vec u_i)$ is small when $i\bmod 3 = 0$ and large otherwise.
The blue polylines are the freeways that affordable alignments may use.
Note that these freeways do not span the entire height of the alignment graph since aligning a reversed wall against a large wall is expensive; even aligning a reversed wall against a small wall costs $2$ units, which we indicate by a small hollow blue circle.
The red polylines show canonical switches between freeways; compared to \cref{fig:highway-switch-schematic}, the switch cost increases by $4$ units since each switch aligns the middle parts $\dol_1\dol_2$ of two small walls in $\hat P$ against $\dcharw$ characters in $\hat T$.}
\label{fig:pmdtw-wall-layout}
\end{figure}

To ensure that a \Yes instance of OV yields a path of cost $k$, we use three cyclic shifts of a suitably padded $X$ and, additionally, pad $X$ with $m$ sentinel $(1,1,\ldots,1)$ vectors (which are not orthogonal to any non-zero vector in $Y$) both at the beginning and at the end.
We also pad $Y$ to place its original vectors in the interior and ensure $m\equiv1\pmod3$.
This way, every orthogonal vector pair $(\vec u,\vec v)\in X\times Y$ corresponds to a subgraph $\AG(G_Y(\vec{v}),G_X(\vec{u}))$ that can be traversed using an affordable path.

For \No instances, the proof follows the same high-level plan as for \cref{thm:ov-to-dtw}, but several new technical issues arise.
Among other issues, walls now contain two $\dcharw$-runs instead of one, and a sought path need not start and end within portals.
To make these complications manageable, we repeat every coordinate nine times, which lets us work with coarser cost estimates while preserving the equivalence.

\subparagraph*{From \Cref{thm:ov-to-pmdtw} to \Cref{thm:lb-pmdtw}}
\Cref{thm:lb-pmdtw} strengthens \cref{thm:ov-to-pmdtw} to arbitrary threshold sequences $(k_n)$ and pattern lengths $(m_n)$.
Here the time complexity depends on the pattern length rather than the threshold.
Still, a simple lower bound using Yao's principle implies that an algorithm running in time $\Oh(n^{1-\eps}m_n)$ is only possible when $m_n=\Omega(n^\eps)$; otherwise, its running time would be sublinear in $n$.

We therefore start from an OV instance with sets of size $s$, choose a target text length $n=\Theta(sd^2)$, and partition $Y$ into groups of size roughly $r = \Theta(m_n/d^2)$.
Applying \cref{thm:ov-to-pmdtw} to each group yields \ref{prob:pmdtw} instances $(P_i,T_i,\kappa_i)$ with threshold $\kappa_i=\Theta(d)$, text length at most $n$, and pattern length well below $m_n$.
At this point, we use three simple padding facts.
Prepending a fresh symbol to the pattern increases the minimum distance by exactly the number of added copies, repeating the first text symbol preserves the minimum distance, and prepending more copies of the first pattern symbol preserves the answer at threshold $\kappa_i$ because the first $\kappa_i+1$ characters of $P_i$ are identical.
Together, these allow us to pad the text to length exactly $n$, pad the pattern to length exactly $m_n$, and raise the threshold from $\kappa_i$ to the target value $k_n$, as formalized in \cref{lem:pmdtw-padding}.
Running the hypothetical \ref{prob:pmdtw} algorithm on all groups would solve OV in truly subquadratic time.

\subsection{Upper Bound (\cref{sec:UB})}

The high-level structure of the algorithm behind \cref{thm:main-UB} follows the blueprint of the state-of-the-art algorithm for weighted edit distance~\cite{GK25}, which itself builds on~\cite{DGHKS23,CKW23}.
Its main ingredient is a recursive divide-and-conquer procedure that, given a cost-$\Oh(k)$ alignment $\A : X\to Y$, recovers an optimal alignment $\A^*$ of the same strings.
In the edit-distance setting, the bottleneck is the \emph{conquer} phase, which solves the original task under an extra assumption that $X$ is the concatenation of $\Oh(k)$ strings, each of period $\Oh(k)$; this automatically implies an analogous structure for $Y$.
Such instances can be solved in $\Oh(\poly(k))$ time because, as observed already in \cite{DGHKS23}, once a period is repeated $\omega(k)$ times, inserting or deleting one more repetition in both strings does not materially change alignments of cost $\Oh(k)$.

For DTW, the same high-level plan still works, albeit with a larger running time because we cannot reduce the number of period repetitions.
The basic structural fact is that cheap alignments must intersect often.
More precisely, if two cost-$\Oh(k)$ alignments map the same fragment of $X$ to nearby fragments of $Y$, then they cannot stay disjoint for long.
Away from the $\Oh(k)$ runs that contribute mismatches, both paths move through cost-$0$ blocks.
If they keep visiting distinct cost-$0$ blocks for many consecutive runs, then they trace two parallel diagonals of cost-$0$ blocks.
This is possible only when the corresponding substring of $\shrink(X)$ has period $\Oh(k)$, because two cost-$\Oh(k)$ alignments cannot drift farther apart than $\Oh(k)$.
By definition of $L_p(X)$ as the maximum length of a substring of $\shrink(X)$ with period at most $p$, such a disjoint stretch has length at most $L_{\Oh(k)}(X)$.
Combining these periodic stretches with the $\Oh(k)$ runs involved in mismatches yields an overall separation bound of $\Oh(kL_{\Oh(k)}(X))$, measured in the number of runs.

Suppose we are given a cost-$\Oh(k)$ alignment $\A:X\to Y$.
We cut out central substrings $X_c$ and $Y_c$ of run-length size $\Oh(kL_{\Oh(k)}(X))$ around the middle of $X$ so that $\A$ aligns $X_c$ with $Y_c$ and compute an optimal alignment $\A_c^*:X_c\to Y_c$ for this middle instance (see the next paragraph).
Combining $\A_c^*$ with suitable prefix and suffix paths derived from $\A$ gives a cost-$\Oh(k)$ alignment $\bar{\A}$ of the entire instance.
We select a cost-$0$ block $B$ visited by $\A_c^*$ near the middle of $X$.
Using the frequent-intersection property, we compare $\A_c^*$ with an arbitrary optimal alignment of $X$ and $Y$: they have common cost-$0$ blocks on both sides of $B$.
Replacing the subpath of the optimal alignment between these two common blocks by the corresponding subpath of $\A_c^*$ shows that there exists an optimal alignment that visits $B$.
We therefore split $\bar{\A}$ at $B$ into two cost-$\Oh(k)$ alignments and recurse on the two corresponding instances, which overlap at $B$.
The optimal alignments returned by the two recursive calls can be combined through the zero-cost block $B$; since some global optimum visits $B$, their combination is an optimal alignment of $X$ and $Y$.
This is the DTW analogue of the edit-distance divide-and-conquer step.

The remaining issue is therefore how to solve this middle instance $(X_c,Y_c)$.
Here the analogy between edit distance and DTW breaks: unlike in edit distance, we cannot hope to shorten long periodic regions by deleting excess repetitions.
At first sight, this seems to force the $\Oh(|X_c|k)$-time algorithm of Kuszmaul~\cite{DBLP:conf/icalp/Kuszmaul19}, which we proved to be near-optimal.
However, we can improve this to $\Ohtilde(|\shrink(X_c)|k)=\Ohtilde(k^2 L_{\Oh(k)}(X))$ using the recent DTW algorithm of Boneh, Golan, Mozes, and Weimann~\cite{BGMW24}, which exploits the run-length encoding of the input.
The only caveat is that their algorithm does not recover a witness alignment, so we show how to do this with a short argument that treats most of the machinery of~\cite{BGMW24} as a black box.
Since the recursive calls always work on substrings of the original strings $S$ and $T$ and the measure $L_p(X)$ is monotone, each middle instance has only $\Oh(kL_{\Oh(k)}(S))$ runs, which makes this subroutine fast enough for the conquer step.

With the main divide-and-conquer procedure adapted, the remaining task is to obtain an $\Oh(k)$-cost DTW alignment of the original strings $S$ and $T$.
This is easiest under the discrete metric and the additional assumption that every run in $S$ and $T$ has exponent $1$.
In that regime, we prove that the edit distance $\ED(\shrink(S),\shrink(T))$ is a constant-factor approximation of $\DTW(S,T)$, and that an optimal edit-distance alignment translates into a DTW alignment of no greater cost.
In general, the approximation ratio grows linearly with the largest exponent in $\rle(S)$ and $\rle(T)$.
Instead of reducing to edit distance in one shot, we therefore proceed gradually.
Let $S'$ and $T'$ be obtained from $S$ and $T$ by replacing each run $a^e$ with $a^{\ceil{e/2}}$.
Then
\[
\DTW(S',T') \le \DTW(S,T) \le 2\DTW(S',T'),
\]
and witness alignments translate naturally in both directions.
This yields a second recursion whose successive levels halve the largest run exponent.
At each level, to compute $\DTW(S,T)$, we first recurse on $(S',T')$, then translate the resulting optimal alignment of $S'$ and $T'$ into an approximate alignment of $S$ and $T$, and finally refine it to an optimal alignment using the divide-and-conquer procedure described above.
Since every run exponent is at most $n$, this recursion has depth
$O(\log n)$.

For weighted DTW, we can proceed  with the mismatch-cost function as follows: cap all costs $\delta(a,b)$ at $k+1$, replace $\delta$ by $\delta'(a,b)\coloneqq \max\left\{\delta(a,b)/2,1\right\}$ (for all unequal symbols $a\ne b$), and repeat until all positive costs become $1$.
This takes $\Oh(\log k)$ levels, just as in~\cite{GK25} for large integer weights.
To avoid paying separately for two top-level recursions, one can even reduce run exponents and mismatch costs simultaneously.

\section{Preliminaries}\label{lem:prelim}

\para{Computational model.}\label{sec:computational-model}
We use a word RAM with word size $w=\Theta(\log n)$, where $n$ denotes the input size.
Each input character is represented by one word and stored in a separate memory cell; accessing and comparing characters takes constant time.
The input weights and threshold are nonnegative integers represented using $\Oh(w)$ bits, and a query to the mismatch-cost function $\delta$ takes constant time.

For the upper bound with real weights and thresholds, we use the real RAM of Erickson, van der Hoog, and Miltzow~\cite[Section~6.1]{EHM24}, with the same word size and character representation.
In addition to word registers, this model has real registers supporting exact addition, subtraction, multiplication, division, and comparison in constant time.
Words may be converted to reals, but the model provides neither conversion of arbitrary reals to integers nor access to their binary representations.
Bounded rounding operations needed by our algorithm are implemented using comparisons, as explained in Section~\ref{sec:UB}.

All running-time bounds are worst-case bounds, including over the random choices of randomized algorithms.
A bounded-error randomized algorithm must return the correct answer with probability at least $2/3$ on every input.
The lower bounds are in the word RAM model; the real RAM extension concerns the upper bound.

\para{Dynamic time warping.}

A mismatch-cost function $\delta:\Sigma^2\to\R_{\ge 0}$ is \emph{normalized} if $\delta(\sigma,\sigma)=0$ for every $\sigma\in\Sigma$ and $\delta(\sigma,\sigma')\ge 1$ for every $\sigma\ne\sigma'$.
Throughout the paper, we assume that mismatch-cost functions are normalized.
This normalization ensures that an alignment of cost at most $k$ contains at most $\lfloor k\rfloor$ mismatches.
In particular, every alignment has cost $0$ or at least $1$.

For two non-empty strings $S$ and $T$, and a mismatch-cost function $\delta$, we define the corresponding alignment graph $\AG_\delta(S,T)$, which is a directed vertex-weighted graph, as follows.
The set of vertices is $[|S|]\times [|T|]$.
From every vertex $(i,j)$ there are (at most) three outgoing edges to $(i+1,j)$, $(i,j+1)$, and $(i+1,j+1)$, whenever these target vertices exist.
Every vertex $v=(i,j)$ has cost $c(v)\coloneqq\delta(S[i],T[j])$.
The cost of a directed path $\A$ in $\AG_\delta(S,T)$ is $\cost(\A)\coloneqq\sum_{v\in\A}c(v)$, counting each visited vertex once, including both endpoints.
We allow paths consisting of a single vertex $v$, whose cost is $c(v)$.
For vertices $u$ and $v$, we denote by $\dist(u,v)$ the minimum cost of a directed path from $u$ to $v$, or $\infty$ if $v$ is unreachable from $u$.
Every path $\A$ from $(1,1)$ to $(|S|,|T|)$ in $\AG_\delta(S,T)$ is called an \emph{alignment} from $S$ to $T$.
We define $\DTW_\delta(S,T)$ to be the minimum cost over all alignments in $\AG_\delta(S,T)$.

A particularly important special case is the \emph{unit-cost function} $\unit$, defined by
\[
\unit(\sigma,\sigma')\coloneqq
\begin{cases}
0 & \text{if }\sigma=\sigma',\\
1 & \text{if }\sigma\ne\sigma',
\end{cases}
\]
which corresponds to the classical \emph{unweighted} DTW problem.
We write $\DTW(S,T)\coloneqq\DTW_\unit(S,T)$ and $\AG(S,T)\coloneqq\AG_\unit(S,T)$ for unit costs and retain the cost-function subscript otherwise.
All our lower bounds (in \cref{sec:lb,sec:pmlb}) hold for unit costs.

We also view the alignment graph as embedded in the plane with the $i$-axis directed downward and the $j$-axis directed to the right.
In this orientation, $(i,j)$ is above $(i+1,j)$ and to the left of $(i,j+1)$.

We make the following observations:

\begin{observation}\label{obs:sub_DTW}
    Let $(i_1,j_1),(i_2,j_2)$ be two vertices in $\AG_\delta(S,T)$ such that $(i_2,j_2)$ is reachable from $(i_1,j_1)$.
    Then, $\dist((i_1,j_1),(i_2,j_2))=\DTW_\delta(S\fragment{i_1}{i_2},T\fragment{j_1}{j_2})$.
\lipicsEnd
\end{observation}

\begin{observation}\label{obs:split_by_A*}
    If a shortest path $\A^*$ from $(1,1)$ to $(|S|,|T|)$ visits some vertex $(i,j)$ then \[\DTW_\delta(S,T)=\DTW_\delta(S\fragment{1}{i},T\fragment{1}{j})+\DTW_\delta(S\fragment{i}{|S|},T\fragment{j}{|T|})-\delta(S[i],T[j]).\]
\lipicsEnd
\end{observation}

\begin{fact}\label{fact:unary_DTW}
    Let $S=\sigma^{|S|}$ for some $\sigma\in\Sigma$ and let $T\in\Sigma^*\sigma \Sigma^*$.
    Then, $\DTW(S,T)=|\{i\mid T[i]\ne \sigma\}|$.
\lipicsEnd
\end{fact}

\begin{fact}\label{fact:dtw_partition}
    For any four strings $A_1,A_2,B_1,B_2$ we have $\DTW_\delta(A_1\cdot A_2,B_1\cdot B_2)\le \DTW_\delta(A_1,B_1)+\DTW_\delta(A_2,B_2)$.
\lipicsEnd
\end{fact}

\begin{lemma}\label{lem:sig_antsig_sig}
    Let $\hat X,\hat Y\in (\Sigma\setminus\{\sigma,\alpha\})^+$ and let $X\coloneqq \sigma^a\hat X\alpha^b$ and $Y\coloneqq\sigma^c\hat Y\alpha^d$ for some positive integers $a,b,c,d$.
    Then, $\DTW(X,Y)=\DTW(\hat X,\hat Y)$.
    The symbols $\sigma$ and $\alpha$ need not be distinct.
\end{lemma}
\begin{proof}
Concatenating an optimal alignment of $\hat X$ and $\hat Y$ with zero-cost alignments of the corresponding guards gives
\[
\DTW(X,Y)\le \DTW(\hat X,\hat Y)
\]
by \cref{fact:dtw_partition}.

For the reverse inequality, let $\A$ be any unit-cost alignment from $X$ to $Y$.
Put $r\coloneqq|\hat X|$, $s\coloneqq|\hat Y|$, and let
\[
 I\coloneqq\fragment{a+1}{a+r},\qquad J\coloneqq\fragment{c+1}{c+s},\qquad Q\coloneqq I\times J
\]
be the interior intervals and their alignment rectangle.
If $\A$ never visits $Q$, every position of $\hat X$ and every position of $\hat Y$ is aligned against a guard symbol, hence incurs a mismatch.
These contributions are distinct, since no visited vertex has both coordinates in the interior intervals.
Thus $\cost(\A)\ge r+s\ge\DTW(\hat X,\hat Y)$.

Otherwise, let $\A_0$ be the subpath from the first to the last vertex of $\A$ having at least one coordinate in its respective interior interval.
Project every vertex $(i,j)$ of $\A_0$ onto $Q$ by replacing it with
\[
 \bigl(\min\{a+r,\max\{a+1,i\}\},
       \min\{c+s,\max\{c+1,j\}\}\bigr).
\]
Since $\A$ visits $Q$ and its coordinates are nondecreasing, the projected sequence starts at $(a+1,c+1)$ and ends at $(a+r,c+s)$.
Each coordinate still increases by at most one per step.
After removing consecutive repetitions of a vertex, the sequence is therefore an alignment of $\hat X$ and $\hat Y$, with coordinates shifted by $(a,c)$.

Vertices already in $Q$ are unchanged.
Every vertex of $\A_0$ outside $Q$ has exactly one coordinate in its interior interval: by monotonicity, a vertex with neither coordinate in its interior interval cannot lie between the chosen endpoints of $\A_0$ when $\A$ visits $Q$.
Such a vertex aligns an interior symbol against a guard and has cost $1$, while its projection has cost at most $1$.
Removing repeated vertices cannot increase the cost either.
The resulting interior alignment consequently has cost at most $\cost(\A_0)\le\cost(\A)$.
Minimizing over $\A$ proves the reverse inequality.
\end{proof}

The alignment grid is naturally partitioned into blocks, where a block is the induced subgraph corresponding to one run of $S$ and one run of $T$.
For a run $\fragment{i_1}{i_2}$ of $S$ and a run $\fragment{j_1}{j_2}$ of $T$, the corresponding block, $B$, contains all vertices in $\fragment{i_1}{i_2}\times \fragment{j_1}{j_2}$.
The \emph{first vertex} of $B$ is $(i_1,j_1)$ and the \emph{last vertex} of $B$ is $(i_2,j_2)$.
A block is said to be a $0$-block if the corresponding runs in $S$ and $T$ have the same character.
Let $S=s_1^{e_1}s_2^{e_2}\ldots s_n^{e_n}$ and $T= t_1^{\ell_1}t_2^{\ell_2}\ldots t_m^{\ell_m}$ be the representations of $S$ and $T$ as concatenations of runs.
For $(i,j)\in [n]\times [m]$, we denote by $B_{i,j}$ the block corresponding to the $i$th run of $S$ and the $j$th run of $T$ in the alignment graph.
Formally, $B_{i,j} \coloneqq \fragment{i_1}{i_2}\times\fragment{j_1}{j_2}$ with $i_1 \coloneqq 1+\sum_{k=1}^{i-1}e_k$, $i_2 \coloneqq i_1+e_i-1 $, $j_1 \coloneqq 1+ \sum_{k=1}^{j-1}\ell_k$, and $j_2 \coloneqq j_1+\ell_j-1$.

For vertices $v=(i,j)$ and $u=(i',j')$, we write $v\preceq u$ if $i\le i'$ and $j\le j'$, equivalently if $u$ is reachable from $v$.

\begin{lemma}\label{lem:first_and_last_in_0_block}
Let $B$ be a $0$-block with first vertex $f$ and last vertex $\ell$, and let $w\in B$.
For every vertex $v\preceq f$, there is a shortest path from $v$ to $w$ that visits $f$.
Similarly, for every vertex $u$ with $\ell\preceq u$, there is a shortest path from $w$ to $u$ that visits $\ell$.
\end{lemma}
\begin{proof}
We prove the second claim; the first one is symmetric.
Write $B=\fragment{i_2}{i_3}\times\fragment{j_2}{j_3}$, so $\ell=(i_3,j_3)$.
Let $\sigma$ be the common character of the runs $S\fragment{i_2}{i_3}$ and $T\fragment{j_2}{j_3}$, and let $\A$ be a shortest path from $w$ to $u$.
If $\A$ already visits $(i_3,j_3)$, in particular if $u=(i_3,j_3)$, there is nothing to prove.
Otherwise, $u$ lies outside $B$, so $\A$ must leave $B$.

Let $(x,y)$ be the last vertex of $\A$ that belongs to $B$.
Since $\A$ does not visit $(i_3,j_3)$, exactly one of the equalities $x=i_3$ and $y=j_3$ holds.
By symmetry, we may assume that $x=i_3$ and $j'\coloneqq y<j_3$.
Let $(i',j_3)$ be the last vertex of $\A$ whose second coordinate is $j_3$, and let $\A'$ be the subpath of $\A$ from $(i_3,j')$ to $(i',j_3)$.

For every $\hat\imath \in \fragmentoc{i_3}{i'}$, the subpath $\A'$ visits at least one vertex $(\hat\imath,\hat\jmath)$ with $\hat\jmath\in \fragment{j'}{j_3}\subseteq \fragment{j_2}{j_3}$, at the cost of $\delta(S[\hat\imath],\sigma)$.
All other vertices of $\A'$ have nonnegative cost, therefore
\[
\cost(\A')\ge \sum_{\hat\imath=i_3+1}^{i'} \delta(S[\hat\imath],\sigma).
\]

Now let $\A^*$ be the path consisting of the horizontal segment from $(i_3,j')$ to $(i_3,j_3)$ followed by the vertical segment from $(i_3,j_3)$ to $(i',j_3)$.
The horizontal segment has cost $0$, because both coordinates stay inside the $0$-block $B$.
Along the vertical segment, every visited vertex except $(i_3,j_3)$ has the form $(\hat\imath,j_3)$ with $\hat \imath \in \fragmentoc{i_3}{i'}$, and hence contributes exactly $\delta(S[\hat\imath],\sigma)$.
Consequently,
\[
\cost(\A^*)=\sum_{\hat\imath=i_3+1}^{i'} \delta(S[\hat\imath],\sigma)\le \cost(\A').
\]
Replacing $\A'$ by $\A^*$ in $\A$ therefore yields another shortest path from $w$ to $u$ that visits $(i_3,j_3)$.
\end{proof}

The following is a direct consequence of \cref{lem:first_and_last_in_0_block}.
\begin{lemma}\label{lem:corners_0_block}
Let $B$ be a $0$-block with first vertex $f$ and last vertex $\ell$, and let $v,u$ be vertices with $v\preceq f$ and $\ell\preceq u$.
If a shortest path from $v$ to $u$ visits $B$, then some shortest path from $v$ to $u$ visits both $f$ and $\ell$.
\end{lemma}
\begin{proof}
Let $\A$ be a shortest path from $v$ to $u$ that visits a vertex $w\in B$.
Let $\A_1$ be the subpath of $\A$ from $v$ to $w$ and $\A_2$ be the subpath from $w$ to $u$.
The claim follows by applying \cref{lem:first_and_last_in_0_block} on both $\A_1$ and $\A_2$ and considering the concatenation of the two obtained paths.
\end{proof}

\para{An equivalent definition of DTW.}
Recall that one can describe DTW using the terminology of string expansions.
For a string $X=X[1]X[2]\dots X[|X|]$, an expansion of $X$ is a string obtained by local duplications of $X$'s characters.
Formally, $X'$ is an expansion of $X$ if
\[
X'=X[1]^{e_1}X[2]^{e_2}\cdots X[|X|]^{e_{|X|}},
\]
where $e_i\in\Z_{\ge1}$ for every $i\in[|X|]$.
In other words, each character of $X$ may be repeated an arbitrary positive number of times, but the relative order of the characters must be preserved.

Using this terminology, $\DTW_\delta$ is the minimum weighted Hamming distance between same-length expansions of the two strings (see, e.g., \cite[Section 4.3]{SankoffKruskal83}).
Formally,
\[
\DTW_\delta(S,T)=
\min_{\substack{
S' \text{ is an expansion of } S\\
T' \text{ is an expansion of } T\\
|S'|=|T'|
}}
\sum_{i=1}^{|S'|}\delta(S'[i],T'[i]).
\]

The following lemma describes how to move between these two views.

\begin{lemma}\label{lem:expensions-alignments}
Let $S$ and $T$ be two non-empty strings.

\begin{enumerate}
\item Given same-length expansions $S'$ and $T'$ of $S$ and $T$, respectively, one can compute in $\Oh(|S'|)$ time an alignment $\A$ in $\AG_\delta(S,T)$ whose cost is at most $\sum_{i=1}^{|S'|}\delta(S'[i],T'[i])$.

\item Given an alignment $\A$ in $\AG_\delta(S,T)$, one can compute in $\Oh(|S|+|T|)$ time same-length expansions $S'$ and $T'$ of $S$ and $T$, respectively, with $\sum_{i=1}^{|S'|}\delta(S'[i],T'[i])=\cost(\A)$.
\end{enumerate}
\end{lemma}

\begin{proof}
(1)
Let $S'$ and $T'$ be same-length expansions of $S$ and $T$.
For every position $k\in[|S'|]$, let $p(k)$ be the index in $S$ from which $S'[k]$ was duplicated, and let $q(k)$ be the index in $T$ from which $T'[k]$ was duplicated.
Notice that $p(\cdot)$ (and similarly $q(\cdot)$) can be defined greedily by setting $p(1)\coloneqq1$ and, for $i>1$, letting $p(i)\coloneqq p(i-1)+1$ if $p(i-1)<|S|$ and $S'[i]=S[p(i-1)+1]$, and $p(i)\coloneqq p(i-1)$ otherwise.
Then
\[
1=p(1)\le p(2)\le\dots\le p(|S'|)=|S|
\qquad
1=q(1)\le q(2)\le\dots\le q(|T'|)=|T|.
\]

We construct $\A$ in one left-to-right pass over the sequence $((p(k),q(k)))_{k=1}^{|S'|}$, keeping the first pair and appending each subsequent pair only if it differs from the last retained pair.
Removing consecutive repetitions preserves the first and last pairs and ensures that at least one coordinate increases at each step.

Since each coordinate increases by at most one between consecutive positions, every pair of consecutive vertices of $\A$ is connected by an edge of $\AG_\delta(S,T)$.
Moreover, $\A$ starts at $(1,1)$ and ends at $(|S|,|T|)$, so it is a valid alignment.

For every $k$ the vertex $(p(k),q(k))$ has cost
\[
c((p(k),q(k)))=\delta(S[p(k)],T[q(k)])=\delta(S'[k],T'[k]).
\]
Hence the cost of $\A$ is at most $\sum_{k=1}^{|S'|}\delta(S'[k],T'[k])$.
Clearly, one can compute $\A$ in $\Oh(|S'|)$ time.

\medskip
(2)
Let $\A\coloneqq ((i_1,j_1),(i_2,j_2),\ldots,(i_\ell,j_\ell))$
be an alignment in $\AG_\delta(S,T)$.
Define strings $S'$ and $T'$ of length $\ell$ by
\[
S'[k]\coloneqq S[i_k],
\qquad
T'[k]\coloneqq T[j_k].
\]

Since $\A$ starts at $(1,1)$, ends at $(|S|,|T|)$, and each step increases every coordinate by at most one, the sequence
$
i_1,i_2,\dots,i_\ell
$
contains each index of $S$ in non-decreasing order and each index appears at least once.
Thus $S'$ is an expansion of $S$.
The same argument shows that $T'$ is an expansion of $T$.

Moreover,
\[
\sum_{k=1}^{\ell}\delta(S'[k],T'[k])
=
\sum_{k=1}^{\ell}\delta(S[i_k],T[j_k])
=
\sum_{k=1}^{\ell} c((i_k,j_k)),
\]
which equals the cost of $\A$.
Every step of $\A$ increases at least one coordinate, so $\ell\le |S|+|T|-1$.
Thus, one can compute $S'$ and $T'$ in $\Oh(\ell)=\Oh(|S|+|T|)$ time.
\end{proof}

\para{Edit Distance.}
The edit distance of two strings $S$ and $T$ is the minimum number of character insertions, deletions, and substitutions required to transform $S$ into $T$.
Equivalently, one can define the edit-distance alignment graph $\AGed(S,T)$, which is a directed edge-weighted graph, as follows.
Its vertex set is $\fragment{0}{|S|}\times\fragment{0}{|T|}$.
Every vertex $(i,j)$ has up to three incoming edges, from $(i-1,j)$, $(i,j-1)$, and $(i-1,j-1)$, whenever these vertices exist.
The edges from $(i-1,j)$ and $(i,j-1)$ have weight $1$, while the edge from $(i-1,j-1)$ has weight $0$ if $S[i]=T[j]$ and weight $1$ otherwise.
Then $\ED(S,T)$ is exactly the distance from $(0,0)$ to $(|S|,|T|)$ in $\AGed(S,T)$.
As with the DTW alignment graph $\AG_\delta(S,T)$, we view $\AGed(S,T)$ as embedded in the plane with the $i$-axis directed downward and the $j$-axis directed to the right.
In this orientation, $(i,j)$ lies above $(i+1,j)$ and to the left of $(i,j+1)$.

\para{Periodicity.}
For a string $S\fragment{1}{n}$, an integer $p\ge 1$ is a period of $S$ if for every $i\in \fragment{1}{n-p}$ it holds that $S[i]=S[i+p]$.
The minimal period of $S$ is called the period of $S$ and is denoted $\per(S)$.
The following is a straightforward observation.
\begin{observation}\label{obs:two-occs-period}
    Let $S$ and $T$ be two strings such that $S=T\fragmentco{i}{i+|S|}$ and $S=T\fragmentco{j}{j+|S|}$.
    Then, either $i=j$ or $|i-j|$ is a period of $S$.
\end{observation}
\begin{proof}
If $i=j$ there is nothing to prove.
Assume $i<j$ and let $p\coloneqq j-i$.
Since $S$ appears in $T$ at positions $i$ and $j$, for every $k\in\fragment{1}{|S|}$ we have $S[k]=T[i+k-1]=T[j+k-1]$.
For every $k\in\fragment{1}{|S|-p}$, we have $S[k+p]=T[i+k+p-1]=T[j+k-1]=S[k]$.
Hence $p=|i-j|$ is a period of $S$.
\end{proof}

\para{Bounded Distances.}
For a threshold $k\ge0$, we define
\[
\kDTW_\delta(S,T)\coloneqq
\begin{cases}
\DTW_\delta(S,T),&\DTW_\delta(S,T)\le k,\\
\infty,&\DTW_\delta(S,T)>k.
\end{cases}
\]
For unit costs, we write $\kDTW(S,T)\coloneqq\kDTW_\unit(S,T)$.
We define $\kED(S,T)$ analogously from $\ED(S,T)$.
Thus, a bounded distance equals the true distance when it is at most $k$ and is $\infty$ otherwise.
We also use $\infty$ as the failure output of algorithms that return an alignment when the distance is at most $k$.

The following lemma shows that the computation of $\kDTW_\delta$ can focus on a band of blocks in $\AG_\delta(S,T)$.
This is an adaptation of the band idea defined for the $\kED$ problem by Ukkonen~\cite{Ukkonen85}.

\begin{lemma}[{\cite[{Claim 16}]{BGMW24}}]\label{lem:kband}
Let $(x,y)$ be a vertex in the alignment graph.
Let $B_{i,j}$ be the block containing $(x,y)$.
If $|j-i| > 2k$, then $\dist((1,1),(x,y)) > k$.
\lipicsEnd
\end{lemma}

The following lemma states that one can compute $\kDTW_\delta(S,T)$ in $\Otild(nk)$ time for $k\ge1$, where $n$ bounds the number of runs in each string.
\begin{lemma}[{\cite{BGMW24}}]\label{lem:nk-computation-ICALP}
    There exists an algorithm that, given two run-length encoded strings $S$ and $T$ with at most $n$ runs each, a threshold $k\ge1$, and oracle access to a mismatch-cost function $\delta$, computes $\kDTW_\delta(S,T)$ in $\Otild(nk)$ time.
\lipicsEnd
\end{lemma}

The following lemma is the celebrated Landau-Vishkin~\cite{LV88} algorithm for computing $\kED$ between two strings.
\begin{lemma}[\cite{LV88}]\label{lem:LV}
    There is an algorithm that, given strings $X, Y \in \Sigma^*$ with $\max\{|X|,|Y|\}\le n$ along with a parameter $k$, in time $\Oh(n+k^2)$ computes $\kED(X,Y)$.
    Moreover, if $\kED(X,Y)\le k$ the algorithm also returns an optimal (unweighted) alignment of $X$ onto $Y$.
\lipicsEnd
\end{lemma}

\para{Orthogonal Vectors.}
An instance of the Orthogonal Vectors problem, denoted $\OV$, consists of two sets
\[
X=\{\vec u_1,\ldots,\vec u_n\},\qquad Y=\{\vec v_1,\ldots,\vec v_m\}\subseteq \{0,1\}^d.
\]
The goal is to decide whether there exist vectors $\vec u\in X$ and $\vec v\in Y$ such that $\vec u\cdot \vec v=0$, where
\[
\vec u\cdot \vec v\coloneqq\sum_{i=1}^d \vec u[i]\cdot \vec v[i]
\]
is the standard inner product.
We say that the instance is a \Yes instance if such a pair exists, and a \No instance otherwise.
We also allow $X$ and $Y$ to be represented as indexed lists with repetitions, with $|X|$ and $|Y|$ denoting their lengths.
Repetitions do not affect whether an orthogonal pair exists.

We use the following logarithmic-dimension form of the Orthogonal Vectors Hypothesis (OVH); see, e.g.,~\cite{KW19}. It is implied by SETH via a reduction of Williams~\cite{Wil05}.
\begin{hypothesis}[Orthogonal Vectors Hypothesis (OVH)]\label{hyp:ovh}
For every $\eps>0$, there is $c > 0$ such that no deterministic algorithm for $\OV$, restricted to $|X|=|Y|=n$ and $d = \ceil{c \log n}$, runs in time $\Oh(n^{2-\eps})$.
\lipicsEnd
\end{hypothesis}

The \emph{randomized Orthogonal Vectors Hypothesis} is the same statement with bounded-error randomized algorithms in place of deterministic algorithms.

\section{Lower Bound for the \ref{prob:dtw} Problem}\label{sec:lb}
We first prove the reduction from $\OV$ in \cref{thm:ov-to-dtw} and then use padding to derive \cref{thm:lb-dtw}.
The reduction encodes vectors by strings whose distances reveal their inner products.
We combine these strings with shorter gadgets and separating runs, then analyze which comparisons an alignment must make.

\subsection{Reduction}
Fix a vector dimension $d\ge1$.
We use the alphabet $\Sigma\coloneqq\{\dchara,\dcharb,\dchare,\dcharf,\dcharw\}$ and the unit-cost function.
For each bit, a coordinate gadget contains two short runs encoding its value, surrounded by long runs of $\dchare$ and $\dcharf$.
For $x,y\in\{0,1\}$, define
\[
E_X(x)\coloneqq\dchare^{4d}\dchara^{1+x}\dcharb^1\dcharf^{4d},
\qquad
E_Y(y)\coloneqq\dchare^{4d}\dcharb^{1+y}\dchara^1\dcharf^{4d}.
\]
The short runs have the following distance.
\begin{observation}\label{obs:coordination-dist}
For every $x,y\in\{0,1\}$, we have $\DTW(\dchara^{1+x}\dcharb,\dcharb^{1+y}\dchara)=2+x\cdot y$.
\end{observation}
\begin{proof}
The first and last vertices of every alignment are mismatches, giving a lower bound of $2$.
When $x=y=1$, every vertex immediately following the first vertex is also a mismatch, so the cost is at least $3$.
These bounds are attained by direct alignments of the four possible pairs of strings.
\end{proof}

For vectors $\vec u,\vec v\in\{0,1\}^d$, concatenating coordinate gadgets gives the vector gadgets
\[
G_X(\vec u)\coloneqq\bigodot_{i=1}^d E_X(\vec u[i]),
\qquad
G_Y(\vec v)\coloneqq\bigodot_{i=1}^d E_Y(\vec v[i]).
\]
Their \emph{narrow gadgets}, obtained by shortening every run to one character, are $\bot_X\coloneqq(\dchare\dchara\dcharb\dcharf)^d$ and $\bot_Y\coloneqq(\dchare\dcharb\dchara\dcharf)^d$.
Thus, $\shrink(G_X(\vec u))=\bot_X$ and $\shrink(G_Y(\vec v))=\bot_Y$ for every pair of vectors.
Each narrow gadget has length $4d$, and each vector gadget has length between $8d^2+2d$ and $8d^2+3d$.

The long $\dchare$- and $\dcharf$-runs force a sufficiently cheap alignment to compare corresponding coordinates.
Consequently, the coordinate costs add exactly.
\begin{lemma}\label{lem:dist_vect_gadget}
For every $\vec u,\vec v\in\{0,1\}^d$, we have $\DTW(G_X(\vec u),G_Y(\vec v))=2d+\vec u\cdot\vec v$.
\end{lemma}
\begin{proof}
Align corresponding coordinate gadgets, matching their outer runs at cost $0$.
By \cref{obs:coordination-dist,fact:dtw_partition}, this gives cost at most $2d+\vec u\cdot\vec v\le3d$.
For the lower bound, let $\A$ be an optimal alignment, whose cost is therefore at most $3d$.
We first show that it matches the long runs of each coordinate to their counterparts, then bound the cost between these matches.

For $i,j\in[1\dd d]$, let $B^{\dchare}_{i,j}$ be the block induced by the $\dchare$-run of coordinate $i$ in $G_X(\vec u)$ and the $\dchare$-run of coordinate $j$ in $G_Y(\vec v)$.
Define $B^{\dcharf}_{i,j}$ analogously for the $\dcharf$-runs, and write $B^{\dchare}_i\coloneqq B^{\dchare}_{i,i}$ and $B^{\dcharf}_i\coloneqq B^{\dcharf}_{i,i}$.

Among the blocks $B^{\dchare}_{i,j}$, the alignment must visit a block in every row $i$ and every column $j$.
Otherwise, an entire run of $4d$ copies of $\dchare$ in one string would be matched only against other symbols, exceeding the cost bound.
It cannot visit two such blocks in the same row either: between any two $\dchare$-runs of the other string lies an $\dcharf$-run of length $4d$, all of which would then mismatch the fixed $\dchare$-run.
The same argument excludes two distinct blocks in one column.
Hence, exactly one block in each row and column is visited.
Monotonicity forces these blocks to be $B^{\dchare}_1,\ldots,B^{\dchare}_d$.
Applying the same argument to the $\dcharf$-runs shows that $\A$ also visits $B^{\dcharf}_1,\ldots,B^{\dcharf}_d$.

For each coordinate $i\in[1\dd d]$, consider the subpath between a visited vertex in $B^{\dchare}_i$ and a visited vertex in $B^{\dcharf}_i$.
In each string, the corresponding substring consists of that coordinate's two short runs, preceded by a non-empty $\dchare$-run and followed by a non-empty $\dcharf$-run.
By \cref{obs:sub_DTW,lem:sig_antsig_sig}, this subpath has cost at least $\DTW(\dchara^{1+\vec u[i]}\dcharb,\dcharb^{1+\vec v[i]}\dchara)=2+\vec u[i]\cdot\vec v[i]$.
These subpaths occur in disjoint coordinate intervals and in increasing order along $\A$.
Their costs therefore sum to at least $2d+\vec u\cdot\vec v$, proving the lower bound.
\end{proof}

For the global construction, let $X\coloneqq(\vec u_0,\ldots,\vec u_{n-1})$ and $Y\coloneqq(\vec v_0,\ldots,\vec v_{m-1})$ be non-empty lists of binary vectors of dimension $d\ge1$.

We separate the gadgets by copies of the \emph{wall} $\wall\coloneqq\dcharw^{12d}$.
In the first string, each vector gadget is followed by a narrow gadget of the opposite type; the second string uses these two types in the opposite order.
Define
\[
\begin{aligned}
\hat X&\coloneqq\wall\,\bigodot_{i=0}^{n-1}[G_X(\vec u_i)\,\wall\,\bot_Y\,\wall\,],\\
\hat Y&\coloneqq\wall\,\bigodot_{i=0}^{n+m-1}[\bot_X\,\wall\,G_Y(\vec v_{i\bmod m})\,\wall\,].
\end{aligned}
\]
The first string contains $n$ two-gadget units, while the second contains $n+m$ such units and repeats the list $Y$ cyclically.
Since a vector gadget and its narrow version have the same run compression, their corresponding runs can be aligned at cost $0$.
We next describe the resulting zero-cost paths through the alignment graph.

\subsubsection{The Alignment Graph}
We work in the unit-cost alignment graph $\AG(\hat X,\hat Y)$, with vertex set $V\coloneqq[1\dd|\hat X|]\times[1\dd|\hat Y|]$.
For any directed path $\A$, its cost $\cost(\A)$ is the sum of its vertex costs, counting each visited vertex once, including the endpoints.
For vertices $u,v\in V$, write $\dist(u,v)$ for the minimum cost of a directed path from $u$ to $v$; set this distance to $\infty$ if no such path exists.

We number the walls of $\hat X$ from $0$ to $2n$, and those of $\hat Y$ from $0$ to $2(n+m)$.
A \emph{portal} is the block induced by one wall in each string.
For wall indices $a\in[0\dd2n]$ and $b\in[0\dd2(n+m)]$, let $P_{a,b}$ denote the corresponding portal.
The \emph{wall offset} of $P_{a,b}$ is $a-b$.
Every wall consists of $12d$ copies of $\dcharw$, so every portal is a $0$-block.

For $a\in[0\dd2n)$, the gadget after wall $a$ of $\hat X$ is $G_X(\vec u_{a/2})$ when $a$ is even and $\bot_Y$ when $a$ is odd.
For $b\in[0\dd2(n+m))$, the gadget after wall $b$ of $\hat Y$ is $\bot_X$ when $b$ is even and $G_Y(\vec v_{((b-1)/2)\bmod m})$ when $b$ is odd.
Every gadget has $4d$ runs.

We now extend the portal indexing to full diagonals of blocks, including the blocks inside gadget comparisons.
For each integer $s$, begin with the union of all portals $P_{a,b}$ with $a-b=s$.
For every such pair with $a<2n$ and $b<2(n+m)$, also include the blocks pairing run $i$ of the gadget after wall $a$ with run $i$ of the gadget after wall $b$, for all $i\in[1\dd4d]$.
The added blocks are exactly those between $P_{a,b}$ and $P_{a+1,b+1}$ when advancing by one run in each string.
Let $D_s$ denote the resulting set of vertices.
Thus, whenever it is non-empty, $D_s$ is a consecutive chain of blocks through all portals with wall offset $a-b=s$.
Each such chain begins and ends in a portal, since both strings begin and end with a wall.
We call $D_s$ a \emph{zero diagonal} when $s$ is even.

On a zero diagonal, each vector gadget is paired with its narrow version, so every vertex has cost $0$.
Whenever $u,v$ lie on the same zero diagonal and $v$ is reachable from $u$, monotone paths within blocks and the diagonal edges joining consecutive blocks give $\dist(u,v)=0$.

\subsubsection{YES Instances}
We construct an alignment with $m$ crossings between consecutive zero diagonals, using an orthogonal pair for one crossing and vectors from the ends of $X$ for the others.
\begin{lemma}\label{lem:ov-yes-to-dtw}
Suppose that $(X,Y)$ is a \Yes instance of $\OV$ with $n>2m$.
If each of the first and last $m$ vectors of $X$ has inner product $1$ with every vector of $Y$,
then $\DTW(\hat X,\hat Y)<(10d+1)m$.
\end{lemma}
\begin{proof}
Call the first and last $m$ vectors of $X$ the \emph{boundary vectors}.
Choose an orthogonal pair $\vec u_{i^*},\vec v_{j^*}$.
The assumed inner products exclude the boundary vectors, giving $i^*\in [m\dd n-m)$ and $j^*\in [0\dd m)$.
The alignment will follow zero diagonals between $m$ paths that each move to the next zero diagonal.

For wall indices $a\in[0\dd 2n]$ and $b\in[0\dd 2(n+m)]$, let $p_{a,b}$ be the first vertex of the portal $P_{a,b}$.
For even wall indices $a\in[0\dd 2n-2]$ and $b\in[0\dd 2(n+m)-4]$, we will construct a \emph{switching path} $\A_{a,b}$ from $p_{a,b}$ to $p_{a+1,b+3}$ of cost $10d+\vec u_{a/2}\cdot\vec v_{(b/2)\bmod m}$.
This path moves from $D_{a-b}$ to $D_{a-b-2}$.
We first describe how to assemble the switching paths into the full alignment, and then construct each switching path.

After following $r$ switching paths, the alignment is on $D_{-2r}$.
A switching path using a vector $\vec u_q$ therefore starts at $p_{2q,2(q+r)}$ and compares it with $\vec v_{(q+r)\bmod m}$.
To compare the chosen orthogonal pair, set $t\coloneqq(j^*-i^*)\bmod m$.
We use $t$ boundary vectors before the chosen orthogonal pair and $m-t-1$ boundary vectors after it.
Form a sequence $q_0,\ldots,q_{m-1}$ by listing the indices in $[0\dd t)$, then $i^*$, and finally the indices in $[n-m+t+1\dd n)$, with each range in increasing order.
The inequalities $t<m\le i^*\le n-m-1<n-m+t+1$ show that the sequence is strictly increasing, with $q_t=i^*$.

For each $r\in[0\dd m)$, use the switching path $\A_{2q_r,2(q_r+r)}$.
These indices are in the required ranges because $0\le q_r<n$ and $0\le q_r+r\le n+m-2$.
The path starts on $D_{-2r}$ and ends on $D_{-2(r+1)}$.
Consecutive switching paths can therefore be connected along zero diagonals: the end of one can reach the start of the next because $q_{r+1}\ge q_r+1$.
Similarly, the start vertex $p_{0,0}=(1,1)$ can reach the start of the first switching path along $D_0$, and the end of the last switching path can reach $p_{2n,2(n+m)}$ along $D_{-2m}$, since $q_0\ge0$ and $q_{m-1}<n$.
All these connecting paths stay on zero diagonals and therefore cost $0$.
After the last connection, traverse the final portal to $(|\hat X|,|\hat Y|)$ at cost $0$.

The switching path with $r=t$ compares $\vec u_{i^*}$ with $\vec v_{j^*}$, because $q_t=i^*$ and $(i^*+t)\bmod m=j^*$, so it costs $10d$.
Every other switching path uses a boundary vector of $X$ and therefore costs $10d+1$.
The endpoints shared by consecutive path pieces lie in portals and have cost $0$, so the full alignment costs $(m-1)(10d+1)+10d=(10d+1)m-1$.
It remains to supply the switching paths with the promised costs.

Construct the switching path by concatenating three paths between consecutive vertices in the sequence
\[
p_{a,b},\qquad p_{a,b+1},\qquad p_{a+1,b+2},\qquad p_{a+1,b+3}.
\]
The first and third pieces keep the first coordinate fixed and traverse a narrow gadget of $\hat Y$ between two walls.
Their only mismatches are the $4d$ characters of the narrow gadget, so each costs $4d$.
For the second piece, traverse the starting portal to its last vertex at cost $0$, then take a diagonal edge into the two vector gadgets.
Follow an optimal alignment of $G_X(\vec u_{a/2})$ and $G_Y(\vec v_{(b/2)\bmod m})$, and enter the next portal by a diagonal edge.
By \cref{lem:dist_vect_gadget}, this piece costs $2d+\vec u_{a/2}\cdot\vec v_{(b/2)\bmod m}$.
All shared endpoints belong to portals and hence have cost $0$, so the three costs add to $10d+\vec u_{a/2}\cdot\vec v_{(b/2)\bmod m}$, as required.
\end{proof}

\subsubsection{NO Instances}
We now show that a \No instance admits no alignment of cost below $(10d+1)m$.
We first study paths of cost at most $10d$, which will let us bound the cost of each crossing between consecutive zero diagonals.
A \emph{transition} is the subpath from the last visited vertex of one portal to the first visited vertex of the next distinct portal.
A horizontal transition keeps the first wall index fixed, whereas a diagonal transition increases both wall indices by one.

\begin{lemma}\label{lem:dtw-portal-transitions}
Suppose that $(X,Y)$ is a \No instance of $\OV$, and let $\A$ be a path in $\AG(\hat X,\hat Y)$ with $\cost(\A)\le10d$.
Consider a transition of $\A$ from $P_{a,b}$ to $P_{a',b'}$.
Then $a'-a,b'-b\in\{0,1\}$, and the following hold:
\begin{enumerate}
    \item If $(a',b')=(a,b+1)$, then $b$ is even and the transition costs at least $4d$.
    \item If $(a',b')=(a+1,b+1)$, with $a$ even and $b$ odd, then the transition costs at least $2d+1$.
\end{enumerate}
\end{lemma}
\begin{proof}
Since $(X,Y)$ is a \No instance with non-empty lists, every vector is nonzero.
Each vector gadget therefore has length at least $8d^2+2d+1>10d$, and each wall has length $12d>10d$.
Monotonicity gives $a'\ge a$ and $b'\ge b$.
If $a'\ge a+2$, the transition traverses the entire wall of $\hat X$ with index $a+1$.
No vertex visited in that wall can belong to a portal, since the transition joins consecutive portal visits.
All $12d$ characters of the wall are therefore mismatched, contradicting the cost bound.
The same argument for $\hat Y$ excludes $b'\ge b+2$, so each wall index increases by at most $1$.

If only the second wall index increases, the transition crosses the gadget between walls $b$ and $b+1$ of $\hat Y$ while staying within one wall of $\hat X$.
The gadget contains no $\dcharw$, so each of its characters contributes a mismatch.
For odd $b$, this is a vector gadget, whose length exceeds the cost of $\A$.
Thus, $b$ is even and the crossed gadget is $\bot_X$, of length $4d$, proving the first bound.

Finally, suppose that both indices increase, with $a$ even and $b$ odd.
The intervening gadgets are $G_X(\vec u)$ and $G_Y(\vec v)$ for some $\vec u\in X$ and $\vec v\in Y$.
In each string, the substring between the transition's endpoints consists of the corresponding gadget preceded and followed by non-empty $\dcharw$-runs.
By \cref{obs:sub_DTW,lem:sig_antsig_sig}, the transition costs at least $\DTW(G_X(\vec u),G_Y(\vec v))$.
This distance is $2d+\vec u\cdot\vec v\ge2d+1$ by \cref{lem:dist_vect_gadget}, since $(X,Y)$ is a \No instance.
\end{proof}

We now lower-bound the cost of moving from one zero diagonal to the next with smaller wall offset.
After extending its endpoints to portals at cost $0$, any such path of cost at most $10d$ must contain two horizontal transitions, each crossing a narrow gadget, with a vector-gadget comparison between them.
For a \No instance, these three transitions already cost at least $4d+(2d+1)+4d=10d+1$.

\begin{lemma}\label{lem:movediagonal}
Let $t$ be an integer, and let $u\in D_{2t}$ and $v\in D_{2t-2}$ be vertices such that $v$ is reachable from $u$.
If $(X,Y)$ is a \No instance of $\OV$, then $\dist(u,v)\ge10d+1$.
\end{lemma}
\begin{proof}
Suppose, for a contradiction, that $\dist(u,v)\le10d$.
We first extend the path endpoints to portals without increasing the cost.
If $u$ lies in a portal, let $u^-$ be its first vertex; otherwise, let $u^-$ be the first vertex of the preceding portal on $D_{2t}$.
Define $v^+$ analogously, using the last vertex of the portal containing $v$ or of the following portal on $D_{2t-2}$.
These portals exist because each diagonal chain begins and ends in a portal.
The vertex $u^-$ can reach $u$, and $v$ can reach $v^+$, with both extensions following zero diagonals at cost $0$.
We can therefore choose a path $\A$ from $u^-$ to $v^+$ with $\cost(\A)\le10d$.

By \cref{lem:dtw-portal-transitions}, each wall index increases by $0$ or $1$ between consecutive distinct portals visited by $\A$.
Thus, the wall offset changes by at most $1$.

List the visited portals as $Q_0,\ldots,Q_L$, and write $\delta_\ell$ for the wall offset of $Q_\ell$.
The endpoint offsets are $\delta_0=2t$ and $\delta_L=2t-2$.
Let $j$ be the first index with $\delta_j=2t-2$, and let $i<j$ be the last preceding index with $\delta_i=2t$.
The step bound gives $j\ge i+2$ and $\delta_\ell=2t-1$ for every $\ell\in[i+1\dd j-1]$.
Indeed, a larger intermediate offset would force a later return to $2t$ on the way to $Q_j$, and a smaller one would reach $2t-2$ before $j$.
Thus, the transitions from $Q_i$ to $Q_{i+1}$ and from $Q_{j-1}$ to $Q_j$ are horizontal, while all transitions between the intermediate portals are diagonal.

By \cref{lem:dtw-portal-transitions}, both selected horizontal transitions start at walls of $\hat Y$ with even indices and cost at least $4d$ each.
Let $P_{a,b}$ be the portal $Q_i$.
The index $b$ is even, and so is $a$, because $a-b=2t$.
The first horizontal transition reaches $Q_{i+1}=P_{a,b+1}$.
The next transition cannot already be the second horizontal one, since $b+1$ is odd.
Hence, $i+2<j$, and the next portal is the diagonal neighbor $Q_{i+2}=P_{a+1,b+2}$.
This diagonal transition starts at an even first wall index and an odd second wall index, so it costs at least $2d+1$ by \cref{lem:dtw-portal-transitions}.

The two selected horizontal transitions and this diagonal transition have disjoint interiors, and any shared endpoints lie in portals and cost $0$.
Their lower bounds therefore add, yielding
\[
\cost(\A)\ge4d+(2d+1)+4d=10d+1,
\]
contrary to our choice of $\A$.
Since all path costs are integers, this proves $\dist(u,v)\ge10d+1$.
\end{proof}

We can now compare the costs of complete alignments.
Every alignment starts on $D_0$ and ends on $D_{-2m}$, so it must cross $m$ successive pairs of zero diagonals.
For a \No instance, each crossing costs at least $10d+1$ by \cref{lem:movediagonal}.

\begin{lemma}\label{lem:ov-no-to-dtw}
If $(X,Y)$ is a \No instance of $\OV$, then $\DTW(\hat X,\hat Y)\ge(10d+1)m$.
\end{lemma}
\begin{proof}
Let $\A$ be an optimal alignment of $\hat X$ and $\hat Y$.
Its first vertex is in $P_{0,0}\subseteq D_0$, while its last vertex is in $P_{2n,2(n+m)}\subseteq D_{-2m}$.

Number the runs of both strings from $0$.
Each wall contributes one run and each gadget contributes $4d$.
Thus, wall $a$ has run index $(4d+1)a$, and corresponding gadget runs have equal offsets from their preceding walls.
It follows that $D_s$ consists precisely of the vertices whose run index in $\hat X$ minus that in $\hat Y$ is $(4d+1)s$.
This difference changes by at most $1$ along any edge of $\A$ and goes from $0$ to $-2m(4d+1)$.
Hence, $\A$ first visits $D_0,D_{-2},\ldots,D_{-2m}$ in this order.
For each $h\in[0\dd m]$, let $v_h$ be the first vertex of $\A$ on $D_{-2h}$.

For each $h\in[0\dd m)$, the subpath from $v_h$ to $v_{h+1}$ costs at least $10d+1$ by \cref{lem:movediagonal}, applied with $t=-h$.
These subpaths have disjoint interiors, and their shared endpoints have cost $0$ because they lie on zero diagonals.
We may therefore add their costs, obtaining $\cost(\A)\ge(10d+1)m$.
\end{proof}

\subsubsection{Completing the Reduction}
The first direction (\cref{lem:ov-yes-to-dtw}) requires prescribed inner products at the ends of $X$.
The following observation shows how to obtain them from an arbitrary instance.
\begin{observation}\label{obs:dtw-boundary-preprocessing}
Let $X,Y$ be non-empty lists of binary vectors of dimension $d\ge1$, with $n\coloneqq|X|$ and $m\coloneqq|Y|$.
In $\Oh((n+m)d)$ time, one can construct an equivalent $\OV$ instance $(X',Y')$ of dimension $d+1$, with $|X'|=n+2m$ and $|Y'|=m$, such that each of the first and last $m$ vectors of $X'$ has inner product $1$ with every vector of $Y'$.
\end{observation}
\begin{proof}
Append one coordinate to every vector, setting it to $0$ in $X$ and to $1$ in $Y$.
Then, prepend and append $m$ copies of the vector $(0,\ldots,0,1)$ to the first list.
The old inner products are unchanged, and every added vector has inner product $1$ with every vector of the second list.
Thus, the answer is preserved, and the resulting lists have the required sizes, dimension, and inner products.
The construction takes $\Oh((n+m)d)$ time.
\end{proof}

\restatementwithproof{\thmovdtw*}
\begin{proof}\label{pg:proof}
Let $n\coloneqq|X|$ and $m\coloneqq|Y|$.
By \cref{obs:dtw-boundary-preprocessing}, we obtain an equivalent instance $(X',Y')$ with $|X'|=n+2m>2m$, $|Y'|=m$, and dimension $d'\coloneqq d+1$.
It satisfies the inner-product condition of \cref{lem:ov-yes-to-dtw}.
Construct $S\coloneqq\hat X$ and $T\coloneqq\hat Y$ from this instance, and set $k\coloneqq(10d'+1)m-1$.
By \cref{lem:ov-no-to-dtw,lem:ov-yes-to-dtw}, the original instance is a \Yes instance if and only if $\DTW(S,T)\le k$.

It remains to bound the sizes and construction time.
Each coordinate gadget has length at most $8d'+3$, so each vector gadget has length at most $8(d')^2+3d'$.
Each concatenated unit of either string consists of one vector gadget, one narrow gadget of length $4d'$, and two walls of length $12d'$ each.
Thus, a unit has length at most $8(d')^2+31d'$.
The string $S$ has an initial wall followed by $n+2m$ units, and $T$ has an initial wall followed by $n+3m$ units.
Since $d'\le2d$, each unit has length at most $94d^2$, and the initial wall has length at most $24d^2$.
Consequently,
\[
|S|,|T|\le 24d^2+94(n+3m)d^2\le400(n+m)d^2.
\]
The threshold satisfies
\[
k=(10d+11)m-1\le21md.
\]
Both strings use the five-symbol alphabet $\Sigma$, and the preprocessing and construction take $\Oh((n+m)d^2)$ time.
These bounds prove the theorem.
\end{proof}

\subsection{Proof of \Cref{thm:lb-dtw}}
The reduction in \cref{thm:ov-to-dtw} gives upper bounds on the string lengths and threshold, whereas \cref{thm:lb-dtw} requires equal, prescribed string lengths and a prescribed threshold.
We first establish the padding property needed to reach these values.
The construction uses a common fresh symbol to increase the string lengths without changing their distance, and two distinct fresh symbols to increase the distance by a prescribed amount.

\begin{lemma}[Padding to Prescribed Parameters]\label{lem:dtw-padding}
Let $S$ and $T$ be non-empty strings, and let $\kappa\ge0$ be an integer.
Given integers $n,k$ satisfying $n>\max\{|S|,|T|\}+k-\kappa$ and $k\ge\kappa$,
one can construct in $\Oh(n)$ time strings $\widehat S$ and $\widehat T$ of length $n$, using at most three additional alphabet symbols, such that, for the unit-cost function, $\DTW(\widehat S,\widehat T)\le k$ if and only if $\DTW(S,T)\le\kappa$.
\end{lemma}
\begin{proof}
Set $\Delta\coloneqq k-\kappa$, and let $\ell_S\coloneqq n-|S|$ and $\ell_T\coloneqq n-|T|$ be the required padding lengths.
The hypotheses give $\Delta\ge0$ and $\ell_S,\ell_T>\Delta$.
Choose three pairwise distinct symbols $\hash,\dol_S,\dol_T$ not occurring in $S$ or $T$, and define
\[
\widehat S\coloneqq\dol_S^\Delta\hash^{\ell_S-\Delta}S,
\qquad
\widehat T\coloneqq\dol_T^\Delta\hash^{\ell_T-\Delta}T.
\]
Both strings have length $n$ and can be constructed in $\Oh(n)$ time.
We prove the stronger identity $\DTW(\widehat S,\widehat T)=\Delta+\DTW(S,T)$.

For the upper bound, align the $\dol_S$- and $\dol_T$-prefixes position by position at cost $\Delta$, then match the two non-empty $\hash$-runs at cost $0$, and finally follow an optimal alignment of $S$ and $T$.
When $\Delta=0$, the first step is simply omitted.

For the lower bound, consider any alignment $\A$ of $\widehat S$ and $\widehat T$.
Delete its initial portion in which both coordinates lie in the added prefixes, that is, all vertices $(i,j)$ with $i\le\ell_S$ and $j\le\ell_T$.
The coordinate that first enters its original string has already traversed all $\Delta$ positions of its $\dol_S$- or $\dol_T$-prefix.
Each of these positions is mismatched, since its symbol is absent from the entire opposite string.
Thus, the deleted portion has cost at least $\Delta$.

Replace each remaining vertex $(i,j)$ by $(\max\{1,i-\ell_S\},\max\{1,j-\ell_T\})$, and delete consecutive repeated vertices.
Each projected coordinate is nondecreasing and increases by at most one per step.
At the first retained vertex, one coordinate has just entered its original string, while the other has not passed its first original character.
The resulting path starts at $(1,1)$ and ends at $(|S|,|T|)$, so it is an alignment of $S$ and $T$.
Pairs of original characters retain their costs.
Every other remaining pair originally compared a padding symbol with a character of $S$ or $T$, so it already had cost $1$ and its replacement cannot cost more.
The resulting alignment therefore has cost at most $\cost(\A)-\Delta$, proving the lower bound.

The two bounds give the claimed identity, and $k=\kappa+\Delta$ then gives the required equivalence.
\end{proof}

We also need a lower bound on how many input symbols an algorithm must inspect.
This will ensure that the assumed running time pays for constructing each padded instance and that the target threshold is large enough to accommodate the reduction.
\begin{lemma}\label{lem:dtw-linear-time}
Consider integers $0\le k<n-2$.
Any Monte Carlo algorithm for the unit-cost binary \ref{prob:dtw} problem, restricted to instances satisfying $|S|=|T|=n$ and threshold $k$, and having error probability at most $1/3$, must inspect $\Omega(n)$ input symbols.
In particular, its worst-case running time is $\Omega(n)$.
\end{lemma}
\begin{proof}
Fix distinct symbols $\ca$ and $\cb$, and let $S\coloneqq\cb\ca^{n-2}\cb$ and $T_0\coloneqq\cb^n$.
Every $\ca$ in $S$ must be mismatched, and the two $\cb$-runs can be matched at cost $0$, so $\DTW(S,T_0)=n-2>k$.
For each $j\in[2\dd n-1]$, let $T_j\coloneqq\cb^{j-1}\ca\cb^{n-j}$.
Aligning the corresponding runs of $S$ and $T_j$ gives cost $0$.
Thus, $(S,T_0,k)$ is a \No instance and each $(S,T_j,k)$ is a \Yes instance.

Run the algorithm on $(S,T_0,k)$ and $(S,T_j,k)$ using the same random choices.
Until position $j$ of the second string is inspected, the two executions receive identical answers and therefore make identical queries.
Their outputs can thus differ only if the execution on $(S,T_0,k)$ inspects that position.
The error bound implies that the probability of a \Yes output is at least $2/3$ on $(S,T_j,k)$ and at most $1/3$ on $(S,T_0,k)$.
The difference between these probabilities is at most the probability that the two outputs disagree under the shared random choices.
Consequently, position $j$ must be inspected on $(S,T_0,k)$ with probability at least $1/3$.

By linearity of expectation, summing over $j\in[2\dd n-1]$ gives an expected number of at least $(n-2)/3=\Omega(n)$ inspected positions.
The worst-case running time is therefore $\Omega(n)$ as well.
\end{proof}

We now combine the reduction with \cref{lem:dtw-padding} to prove the lower bound for prescribed string lengths and thresholds.
We divide one OV vector list into groups, reduce each group together with the other list, and pad the resulting instances to the required parameters.
\restatementwithproof{\mthmlb*}
\begin{proof}
Fix $0<\eps<1$.
Let $c_{\mathrm{red}}\ge1$ be a common upper bound on the constants hidden in the size and threshold bounds of \cref{thm:ov-to-dtw}, and choose $c_{\mathrm{ov}}>0$ as promised by \cref{hyp:ovh}, or its randomized version for randomized algorithms, with $\eps/2$ in place of $\eps$.

Let $(k_n)_{n=1}^\infty$ be an integer sequence computable in $\Oh(n)$ time and satisfying $1\le k_n\le n/2$ for all sufficiently large $n$.
Suppose, towards a contradiction, that an algorithm $\mathbf A$ solves \cref{prob:dtw} on instances with $|S|=|T|=n$ and threshold $k_n$ in time $\Oh(n^{1-\eps}\cdot k_n)$.
For sufficiently large $n$, we have $k_n\le n/2<n-2$, so \cref{lem:dtw-linear-time} gives a running-time lower bound of $\Omega(n)$.
Comparing this with the assumed upper bound gives $n^{1-\eps}\cdot k_n=\Omega(n)$, or equivalently $k_n=\Omega(n^\eps)$.

For sufficiently large $s$, consider an instance of $\OV$ with $|X|=|Y|=s$ and $d\coloneqq\ceil{c_{\mathrm{ov}}\log s}$.
By \cref{thm:ov-to-dtw}, keeping all of $X$ and any non-empty subset of $Y$ produces two strings of length at most $2c_{\mathrm{red}}sd^2$.
We choose the target length to be four times this bound, leaving room for padding, and take the corresponding threshold:
\[
n\coloneqq\ceil{8c_{\mathrm{red}}sd^2},
\qquad
k\coloneqq k_n.
\]
The sequence's computability lets us obtain $k$ in $\Oh(n)$ time.
Since $d=\Theta(\log s)$ and $n=\Theta(sd^2)$, the bound $k=\Omega(n^\eps)$ implies $k\ge4c_{\mathrm{red}}d$ for sufficiently large $s$.

We will partition $Y$ into groups and apply the reduction to each group together with all of $X$.
A group of at most $r$ vectors yields a threshold of at most $c_{\mathrm{red}}dr$.
To bound this threshold by $k/2$, we set
\[
r\coloneqq\min\!\left\{s,\left\lfloor\tfrac{k}{2c_{\mathrm{red}}d}\right\rfloor\right\}.
\]
Since $k\ge4c_{\mathrm{red}}d$, rounding down loses at most a factor of two, so $r\ge1$ and $r=\Theta(\min\{s,k/d\})$.
Set $t\coloneqq\ceil{s/r}$, and partition $Y$ into $t$ non-empty groups $Y_1,\ldots,Y_t$ of size at most $r$.
The original instance $(X,Y)$ is a \Yes instance if and only if $(X,Y_i)$ is a \Yes instance for at least one $i\in[1\dd t]$.

Fix a group $Y_i$.
Applying \cref{thm:ov-to-dtw} to $(X,Y_i)$ yields an equivalent instance $(S_i,T_i,\kappa_i)$ of the unit-cost \hyperref[prob:dtw]{DTW} problem with an integer threshold $\kappa_i$ and
\[
\kappa_i\le c_{\mathrm{red}}dr\le\tfrac12k,
\qquad
|S_i|,|T_i|\le c_{\mathrm{red}}(s+r)d^2\le\tfrac14n.
\]
Since $k\le n/2$, these bounds give $\max\{|S_i|,|T_i|\}+k-\kappa_i\le3n/4<n$.
Thus, all hypotheses of \cref{lem:dtw-padding} hold, and it produces strings $\widehat S_i$ and $\widehat T_i$ of length $n$ such that $\DTW(\widehat S_i,\widehat T_i)\le k$ if and only if $\DTW(S_i,T_i)\le\kappa_i$.
By the reduction, the latter condition holds if and only if $(X,Y_i)$ is a \Yes instance.
We can therefore use $\mathbf A$ on $(\widehat S_i,\widehat T_i,k)$ to decide whether $(X,Y_i)$ is a \Yes instance.
The construction uses unit costs, and padding adds at most three symbols to the reduction's constant-size alphabet.

We perform this construction for every group and answer \Yes if at least one group is accepted.
If $\mathbf A$ is a bounded-error randomized algorithm, repeat each call independently $\Oh(\log(t+1))$ times and take the majority answer to reduce its error to at most $1/(3t)$.
By the union bound, all group answers, and hence the final answer, are correct with probability at least $2/3$.

For each group, the reduction takes $\Oh((s+r)d^2)=\Oh(n)$ time, and \cref{lem:dtw-padding} takes $\Oh(n)$ time.
As $n^{1-\eps}\cdot k=\Omega(n)$, this construction time and the initial computation of $k$ are absorbed by the time for the calls to $\mathbf A$.
Allowing for error amplification, the total running time is therefore $\Oh(t\cdot n^{1-\eps}\cdot k\log(t+1))$.

Since $1\le r\le s$ and $r=\Theta(\min\{s,k/d\})$, the number of groups satisfies $t=\Oh(1+sd/k)$ and $t\le s$.
Since $k\le n/2$ and $sd=\Oh(n)$, we have $tk=\Oh(k+sd)=\Oh(n)$.
The total running time is therefore $\Oh(n^{2-\eps}\log s)$.
Using $n=\Theta(s\log^2s)$ and absorbing the resulting polylogarithmic factor into $s^{\eps/2}$, we obtain $\Oh(s^{2-\eps/2})$.
This contradicts the choice of $c_{\mathrm{ov}}$, which rules out such an algorithm for $\OV$ on sets of size $s$ and dimension $d=\ceil{c_{\mathrm{ov}}\log s}$ under the corresponding version of OVH.
\end{proof}

\section{Lower Bound for the \ref{prob:pmdtw} Problem}\label{sec:pmlb}
This section is dedicated to proving \cref{thm:lb-pmdtw}.
We first provide a reduction from $\mathsf{OV}$ to \cref{prob:pmdtw}, proving \cref{thm:ov-to-pmdtw}, and only then derive \cref{thm:lb-pmdtw} using appropriate padding.

\subsection{Reduction}

\subsubsection{OV Variant}
The following preprocessing places a suitable orthogonal pair, if one exists, at interior indices that agree modulo $3$, leaving room to align the parts of the pattern before and after the pair.
It also amplifies nonzero inner products to at least $9$, so their contribution exceeds the construction's constant alignment overhead.

\begin{lemma}\label{lem:pmlb-ov-variant}
Given an instance $(X,Y)$ of $\OV$ with $n\coloneqq|X|\ge1$, $m\coloneqq|Y|\ge1$, and dimension $d\ge1$, one can in $\Oh((n+m)d)$ time construct an equivalent instance $(X',Y')$, represented by the indexed lists
\[
X'\coloneqq(\Vec{u}'_0,\ldots,\Vec{u}'_{n'-1}),\qquad
Y'\coloneqq(\Vec{v}'_0,\ldots,\Vec{v}'_{m'-1})
\]
of vectors in $\{0,1\}^{d'}$, where $n'=\Oh(n+m)$, $m'=\Oh(m)$, $d'\coloneqq9d$, and $m'\equiv1\pmod3$, such that:
\begin{enumerate}
    \item If $(X,Y)$ is a \Yes instance, then $\Vec{u}'_i\cdot\Vec{v}'_j=0$ for some $i\in[m'\dd n'-m']$ and $j\in[1\dd m'-2]$ with $i\equiv j\pmod3$.
    \item If $(X,Y)$ is a \No instance, then $\Vec{u}'\cdot\Vec{v}'\ge9$ for every $\Vec{u}'\in X'$ and $\Vec{v}'\in Y'$.
\end{enumerate}
\end{lemma}
\begin{proof}
We first repeat each coordinate of every vector in $X$ and $Y$ nine times, increasing the dimension to $d'=9d$ and multiplying every inner product by nine.
We use copies of the all-ones vector $1^{d'}$ for padding; this preserves the answer because $1^{d'}$ is orthogonal only to $0^{d'}$, and an instance containing $0^{d'}$ is already a \Yes instance.

We pad $Y$ to place all its original vectors in the interior and make its length congruent to $1$ modulo $3$.
To this end, we prepend one copy of $1^{d'}$ and append between one and three copies, choosing their number so that the resulting list $Y'$ has length $m'\equiv1\pmod3$.

Next, we pad $X$ with at most two copies of $1^{d'}$ to obtain a list $X_0$ whose length $q$ is divisible by $3$.
For $r\in\{1,2\}$, define the cyclic shift $X_r[i]\coloneqq X_0[(i+r)\bmod q]$ for $i\in[0\dd q)$.
In $X_0\cdot X_1\cdot X_2$, every vector of $X$ has a copy at each index residue modulo $3$.
We then prepend and append $m'$ copies of $1^{d'}$, obtaining a list $X'$ of length $n'=3q+2m'$.
The prepended copies shift all indices by the same amount, so each original vector still has a copy at every residue modulo $3$.
Thus, every original orthogonal pair yields a pair in $X',Y'$ satisfying the required index conditions.

All output vectors, including the padding vectors, have each coordinate repeated nine times, so every inner product is a multiple of $9$.
Since the transformation preserves the answer, a \No input therefore yields inner products of at least $9$.
Finally, $m'\le m+4=\Oh(m)$ and $n'\le3(n+2)+2m'=\Oh(n+m)$, and the construction takes $\Oh((n+m)d)$ time.
\end{proof}

We apply \cref{lem:pmlb-ov-variant} to the input instance and denote the resulting lists by $X\coloneqq(\Vec{u}_0,\ldots,\Vec{u}_{n-1})$ and $Y\coloneqq(\Vec{v}_0,\ldots,\Vec{v}_{m-1})$ in $\{0,1\}^d$, with $n\coloneqq|X|$ and $m\coloneqq|Y|$.
By construction, $n\ge m$, $d\ge9$, and $m\equiv1\pmod3$.
We now construct a triple $(\hat P,\hat T,k)$ from these lists.

\subsubsection{Gadgets}
Set $k\coloneqq10d+8$.
We reuse the vector gadgets $G_X,G_Y$ and their narrow versions $\bot_X,\bot_Y$ from \cref{sec:lb}.
Recall that the narrow gadgets are the run compressions of the corresponding vector gadgets.

For the wall gadgets, we reuse the base wall $\wall=\dcharw^{12d}$ and introduce two fresh symbols $\dol_1$ and $\dol_2$.
We define the \emph{small wall} $\wall^{\dol}_{\smallz}\coloneqq\wall\dol_1\dol_2\wall$ and the \emph{big wall} $\wall^{\dol}_{\bigz}\coloneqq\wall\dol_1^{2k}\dol_2^{2k}\wall$.
For the pattern boundaries, we also define the \emph{reversed wall} $\wall^{\dol}_{\mathsf{rev}}\coloneqq\dcharw^{2k}\dol_2^{2k}\dol_1\dcharw^{2k}$.
Small and big walls have the same run compression, so they can be aligned at cost $0$.
The reversed boundary walls ensure that any alignment of cost at most $k$ from the pattern to a substring of the text can be adjusted, without increasing its cost, to start and end inside small text walls; see \cref{lem:pmlb-start-end-portal}.

For every integer $i$, let $\wall_i\coloneqq\wall^{\dol}_{\smallz}$ if $i\equiv0\pmod3$, and $\wall_i\coloneqq\wall^{\dol}_{\bigz}$ otherwise.
We define the text by
\[\hat T\coloneqq\bigodot_{i=0}^{n-1}\left(\wall_i\,G_X(\Vec{u}_i)\,\wall^{\dol}_{\bigz}\,\bot_Y\,\right).\]

All internal pattern walls are small, and the two boundary walls are reversed:

\[\hat P \coloneqq \wall^{\dol}_{\mathsf{rev}} \bot_X \wall^{\dol}_{\smallz} G_Y(\Vec{v}_0) \bigodot_{i=1}^{m-1}\left(\wall^{\dol}_{\smallz}\, \bot_X \,\wall^{\dol}_{\smallz} \, G_Y(\Vec{v}_i)\,\right) \wall^{\dol}_{\mathsf{rev}}.
\]

We call a non-empty substring $S$ of $\hat T$ a \emph{$k$-DTW occurrence} of $\hat P$ if $\DTW(\hat P,S)\le k$.
For a non-empty substring $\hat T\fragment{x}{y}$, we view an alignment from $\hat P$ to this substring as a path from $(1,x)$ to $(|\hat P|,y)$ in $\AG(\hat P,\hat T)$.
Thus, text coordinates refer to the full text, and path costs count each visited vertex once, including the endpoints.
We next prove that $(X,Y)$ is a \Yes instance of OV if and only if there is a $k$-DTW occurrence of $\hat P$ in $\hat T$.

\subsubsection{YES Instances}
We will show that if $(X,Y)$ is a \Yes instance of OV, there is a $k$-$\DTW$ occurrence of $\hat P$ in $\hat T$.

\begin{lemma}\label{lem:pmlb-yes-so-yes}
If $(X,Y)$ is a \Yes instance, then there is a $k$-DTW occurrence of $\hat P$ in $\hat T$.
\end{lemma}
\begin{proof}
By \cref{lem:pmlb-ov-variant}, choose an orthogonal pair $\Vec{u}_i\in X,\Vec{v}_j\in Y$ with $i\in[m\dd n-m]$, $j\in[1\dd m-2]$, and $i\equiv j\pmod3$.
Set $\ell\coloneqq i-j$.
Before and after the selected pair, we will align vector gadgets with narrow gadgets at cost $0$.
The central comparison will consume two pattern vector gadgets but only one text vector gadget, shifting the subsequent correspondence by one.

To describe the alignment, define the text and pattern units
\[
\begin{aligned}
T_z&\coloneqq \wall_{\ell+z}G_X(\Vec{u}_{\ell+z})\wall^{\dol}_{\bigz}\bot_Y
    &&\text{for }z\in[0\dd m-2],\\
P_z&\coloneqq \wall^{\dol}_{\smallz}\bot_X\wall^{\dol}_{\smallz}G_Y(\Vec{v}_z)
    &&\text{for }z\in[1\dd m-1].
\end{aligned}
\]
We align $\hat P$ to the text substring
\[
    \tilde T\coloneqq T_0T_1\cdots T_{m-2}\wall_{\ell+m-1}.
\]
The bounds on $i$ and $j$ give $\ell\ge0$ and $\ell+m-1<n$, so this substring exists.
Since $\ell\equiv0\pmod3$ and $m\equiv1\pmod3$, its initial wall $\wall_\ell$ and final wall $\wall_{\ell+m-1}$ are both small.
The pattern consists of its initial reversed wall, the substring $\bot_X\wall^{\dol}_{\smallz}G_Y(\Vec{v}_0)$, the units $P_1,\ldots,P_{m-1}$, and its final reversed wall.

We construct the alignment in three phases:
\begin{enumerate}
    \item Align the initial reversed wall of $\hat P$ to $\wall_\ell$, and the following substring $\bot_X\wall^{\dol}_{\smallz}G_Y(\Vec{v}_0)$ to the rest of $T_0$.
    Then align $P_z$ to $T_z$ for each $z\in[1\dd j)$.
    \item For the central comparison, align $P_jP_{j+1}$ to $T_j$.
    These units exist because $1\le j\le m-2$, and $T_j$ contains the selected gadget $G_X(\Vec{u}_i)$ because $\ell+j=i$.
    \item The central comparison has consumed two pattern units but only one text unit, so the next unused units, if any, are $P_{j+2}$ and $T_{j+1}$.
    Continue by aligning $P_z$ to $T_{z-1}$ for each $z\in[j+2\dd m)$, and align the final reversed wall of $\hat P$ to $\wall_{\ell+m-1}$.
\end{enumerate}
These pieces cover both strings $\hat P$ and $\tilde T$ in order.
We now describe how to carry out each comparison and bound its cost.

Each reversed pattern wall can be aligned to a small text wall at cost $2$.
Match their outer $\dcharw$-runs, align the first pattern $\dol_2$ with the text $\dol_1$, and align all remaining pattern $\dol_2$'s and the trailing $\dol_1$ with the text $\dol_2$.
The only mismatches are the two pairs involving $\dol_1$.
This implements the initial and final wall comparisons, since both endpoint walls of $\tilde T$ are small.

All other comparisons in the first and third phases have cost $0$.
Recall that $\shrink(G_X(\Vec{u}))=\bot_X$ and $\shrink(G_Y(\Vec{v}))=\bot_Y$ for every $\Vec{u}\in X$ and $\Vec{v}\in Y$.
The small and big walls also have the same run compression, namely $\dcharw\dol_1\dol_2\dcharw$.
To align a pattern unit $P_z$ with a text unit $T_{z'}$, pair their first walls, then $\bot_X$ with $G_X(\Vec{u}_{\ell+z'})$, then their second walls, and finally $G_Y(\Vec{v}_z)$ with $\bot_Y$.
Each pair has identical run compressions, so we can align corresponding runs of equal symbols at cost $0$.
Omitting the first wall pair gives a zero-cost alignment of $\bot_X\wall^{\dol}_{\smallz}G_Y(\Vec{v}_0)$ with $G_X(\Vec{u}_\ell)\wall^{\dol}_{\bigz}\bot_Y$, the part of $T_0$ after its initial wall.

It remains to implement the central comparison of $P_jP_{j+1}$ with $T_j$.
We align $G_Y(\Vec{v}_j)$ with the selected gadget $G_X(\Vec{u}_i)$ and $G_Y(\Vec{v}_{j+1})$ with the following $\bot_Y$.
We align the remaining pieces of $P_jP_{j+1}$ to the two text walls, giving the following partition (pattern above, text below):
\[
\begin{array}{c|c|c|c}
\wall^{\dol}_{\smallz}\bot_X\wall^{\dol}_{\smallz}
    &G_Y(\Vec{v}_j)
    &\wall^{\dol}_{\smallz}\bot_X\wall^{\dol}_{\smallz}
    &G_Y(\Vec{v}_{j+1})\\
\wall_i&G_X(\Vec{u}_i)&\wall^{\dol}_{\bigz}&\bot_Y
\end{array}
\]
In each of the first and third columns, align the first small pattern wall and $\bot_X$ to the first $\dcharw$ of the text wall.
Keep this text position fixed while traversing the pattern substring.
All $4d$ symbols of $\bot_X$ mismatch $\dcharw$, and the small pattern wall contributes exactly two further mismatches, from its $\dol_1$ and $\dol_2$.
The cost is therefore $4d+2$.
Then move to the remaining small pattern wall while keeping the text at the same position, and align these two walls at cost $0$ by matching corresponding runs.
In the second column, \cref{lem:dist_vect_gadget} gives cost $2d+\Vec{u}_i\cdot\Vec{v}_j=2d$.
The fourth column costs $0$ by the run-compression identity above.
Thus, the central comparison costs $(4d+2)+2d+(4d+2)=10d+4$.

Concatenating these local alignments in the prescribed order yields an alignment of $\hat P$ to $\tilde T$.
The two boundary comparisons cost $2$ each, the central comparison costs $10d+4$, and all other comparisons cost $0$.
The total cost is therefore $2+(10d+4)+2=10d+8=k$, as required.
\end{proof}

\subsubsection{NO Instances}
It remains to show that if $(X,Y)$ is a \No instance of OV, there is no $k$-$\DTW$ occurrence of $\hat P$ in $\hat T$.
We number all wall occurrences from left to right, starting from $0$: the pattern walls have indices $0,\ldots,2m$, and the text walls have indices $0,\ldots,2n-1$.
In particular, the copy of $\wall_i$ preceding $G_X(\Vec{u}_i)$ has text wall index $2i$, so the small text walls are exactly those whose indices are divisible by $6$.

Between consecutive pattern walls, the gadget is $\bot_X$ after an even wall and $G_Y(\Vec{v})$ after an odd wall, for some $\Vec{v}\in Y$.
Between consecutive text walls, it is $\bot_Y$ after an odd wall and $G_X(\Vec{u})$ after an even wall, for some $\Vec{u}\in X$.
Every $\dcharw$-run has length at least $12d$.
The narrow gadgets have length $4d$, whereas every vector gadget has length at least $8d^2$.

As in \cref{sec:lb}, we call a block of $\AG(\hat P,\hat T)$ corresponding to a $\dcharw$-run in each string a \emph{portal}.
A portal has \emph{wall indices} $(a,b)$ if its pattern run belongs to wall $a$ and its text run belongs to wall $b$.
Each pair of walls contains four portals, one for each choice of their two $\dcharw$-runs.
We call $a-b$ the \emph{wall offset} of such a portal.

We first use the boundary walls to move the endpoints of a hypothetical alignment of cost at most $k$ into portals whose text walls are small, without increasing its cost.
Since the pattern boundary walls have indices $0$ and $2m\equiv2\pmod6$ and small text walls have indices divisible by $6$, the resulting wall offsets are distinct even integers.
We will then use bounds on the costs of moving between consecutive portals to show that such a low-cost alignment must align two narrow gadgets against walls and compare a pair of vector gadgets, costing at least $4d+(2d+9)+4d=10d+9>k$.
\begin{lemma}\label{lem:pmlb-start-end-portal}
Let $\A$ be an alignment of cost at most $k$ from $\hat P$ to a substring of $\hat T$.
There is an alignment $\A'$ from $\hat P$ to a substring of $\hat T$, with $\cost(\A')\le\cost(\A)$, that starts in a portal with wall indices $(0,b_0)$ and ends in a portal with wall indices $(2m,b_1)$, where $b_0\equiv b_1\equiv0\pmod6$.
\end{lemma}
\begin{proof}
Consider either the initial or the final copy of
$\wall^{\dol}_{\mathsf{rev}}=\dcharw^{2k}\dol_2^{2k}\dol_1\dcharw^{2k}$ in $\hat P$.
Its first $\dcharw$-run and the subsequent $\dol_2$-run each contain a symbol matched by $\A$: otherwise, all $2k$ positions of one of these runs would incur mismatches, exceeding the cost bound $k$.
Fix one such match in each run.
If the matched $\dol_2$ in $\hat T$ belonged to a big wall $\wall^{\dol}_{\bigz}=\dcharw^{12d}\dol_1^{2k}\dol_2^{2k}\dcharw^{12d}$, the earlier matched $\dcharw$ would lie before that wall's $\dol_1^{2k}$-run: there is no $\dcharw$ between this run and the matched $\dol_2$.
The subpath between the two fixed matches would therefore traverse every position of that $\dol_1^{2k}$-run.
The corresponding pattern substring contains only $\dcharw$ and $\dol_2$, so all $2k$ positions of this text run would mismatch, again exceeding the cost bound.
Hence, the matched $\dol_2$ in $\hat T$ belongs to a small wall $\wall^{\dol}_{\smallz}=\dcharw^{12d}\dol_1\dol_2\dcharw^{12d}$.

Let $u$ and $v$ be the vertices of $\A$ realizing the chosen $\dol_2$-matches in the initial and final reversed walls, respectively.
We will replace the prefix of $\A$ ending at $u$ and the suffix starting at $v$.
The run compressions of these pattern pieces are $\dcharw\dol_2$ and $\dol_2\dol_1\dcharw$, respectively.
A cost-$0$ alignment would require the corresponding text pieces to have the same run compressions.
However, neither $\dcharw\dol_2$ nor $\dol_2\dol_1$ occurs consecutively in the run compression of $\hat T$.
Thus, both pieces have cost at least $1$.

To construct replacements of cost $1$, we use the following identities, valid for every integer $r\ge1$:
\[
\begin{aligned}
\DTW(\dcharw^{2k}\dol_2^r,\;\dcharw^{12d}\dol_1\dol_2)&=1,\\
\DTW(\dol_2^r\dol_1\dcharw^{2k},\;\dol_2\dcharw^{12d})&=1.
\end{aligned}
\]
In each identity, match the corresponding $\dcharw$- and $\dol_2$-runs, with a single mismatch between the sole $\dol_1$ and a~$\dol_2$.
In the first identity, take $r$ to be the number of $\dol_2$ positions in the pattern prefix ending at $u$.
This gives a replacement prefix starting at the beginning of the small text wall matched at~$u$ and ending at the same vertex $u$.
The second identity applies to the pattern suffix starting at $v$ and gives a replacement suffix ending at the end of the small text wall matched there.

The replacements join the unchanged subpath from $u$ to $v$ at its original endpoints, yielding the required alignment $\A'$.
Each replacement costs~$1$, while each original piece costs at least $1$.
Both joining vertices have cost $0$, so these comparisons give $\cost(\A')\le\cost(\A)$.
The new endpoints lie in portals belonging to the first and last pattern walls and to small text walls, whose indices are divisible by $6$.
\end{proof}

We now follow a path through its successive portal visits.
A \emph{transition} is the subpath from the last visited vertex in one portal to the first visited vertex in the next distinct portal.
Transitions within the same pair of walls keep both wall indices fixed; the following lemma gives the bounds needed below for transitions that change at least one wall index.

\begin{lemma}\label{lem:pmlb-portal-transitions}
Suppose that $(X,Y)$ is a \No instance of OV, and let $\A$ be a path in $\AG(\hat P,\hat T)$ with $\cost(\A)<10d+9$.
Consider a transition of $\A$ between portals with wall indices $(a,b)$ and $(a',b')$.
Then $a'-a,b'-b\in\{0,1\}$, and the following hold:
\begin{enumerate}
    \item If $(a',b')=(a+1,b)$, then $a$ is even and the transition costs at least $4d$.
    \item If $(a',b')=(a,b+1)$, then $b$ is odd and the transition costs at least $4d$.
    \item If $(a',b')=(a+1,b+1)$, with $a$ odd and $b$ even, then the transition costs at least $2d+9$.
\end{enumerate}
\end{lemma}
\begin{proof}
Since $d\ge9$, the assumed cost of \(\A\) is less than both $12d$ and $8d^2$.
Suppose that, in either string, a $\dcharw$-run lay strictly between the runs containing the endpoints of the transition.
The path would traverse every position of this intermediate run.
No such position could be matched: a $\dcharw$-match there would place the path in another portal between the two consecutive portal visits.
The whole run would therefore be mismatched, giving a cost of at least $12d$, a contradiction.
Thus, in each string, the transition either stays in the same $\dcharw$-run or advances to the next one, and each wall index increases by at most $1$.

If only the pattern wall index increases, the transition traverses the gadget between pattern walls $a$ and $a+1$, while its text segment lies within one wall and contains only $\dcharw$, $\dol_1$, and $\dol_2$.
The pattern gadget contains none of these three wall symbols, so every position of the gadget contributes a mismatch.
For odd~$a$, the gadget is a vector gadget of length at least $8d^2$, which exceeds the cost of the entire path.
Thus, $a$ must be even, and the crossed gadget is $\bot_X$ of length $4d$, giving the first bound.

If only the text wall index increases, the transition crosses the gadget between text walls $b$ and $b+1$, while the pattern segment stays inside wall $a$.
Again, the gadget contains none of the wall symbols, so every position of the text gadget is mismatched.
For even $b$, the gadget is $G_X(\Vec{u})$ for some $\Vec{u}\in X$ and has length at least $8d^2$, exceeding the cost bound.
Hence $b$ is odd, and the gadget is $\bot_Y$, giving a cost of at least $4d$.

Finally, suppose that both indices increase, with $a$ odd and $b$ even.
Because no $\dcharw$-run is skipped, the transition goes from the last $\dcharw$-run of each source wall to the first $\dcharw$-run of the next wall.
The parities of $a$ and $b$ ensure that the intervening gadget is a vector gadget in each string.
The pattern and text segments of the transition therefore consist of $G_Y(\Vec{v})$ and $G_X(\Vec{u})$, respectively, each preceded and followed by a nonempty $\dcharw$-run, for some $\Vec{v}\in Y$ and $\Vec{u}\in X$.
Since neither vector gadget contains $\dcharw$, \cref{obs:sub_DTW,lem:sig_antsig_sig} identify the distance between the transition's endpoints with $\DTW(G_Y(\Vec{v}),G_X(\Vec{u}))$.
By \cref{lem:dist_vect_gadget} and symmetry of DTW, this distance is $2d+\Vec{u}\cdot\Vec{v}$.
The transition costs at least this distance, and the \No-instance guarantee in \cref{lem:pmlb-ov-variant} gives $\Vec{u}\cdot\Vec{v}\ge9$.
This proves the last bound.
\end{proof}

The local bounds prevent a low-cost path from changing between different even wall offsets.
\begin{lemma}\label{lem:pmlb-wall-offset}
Suppose that $(X,Y)$ is a \No instance of OV.
Every path in $\AG(\hat P,\hat T)$ whose endpoints lie in portals with distinct even wall offsets has cost at least $10d+9$.
\end{lemma}
\begin{proof}
Suppose, for a contradiction, that such a path $\A$ has cost less than $10d+9$.
List the distinct portals visited by $\A$, in order, as $Q_1,\ldots,Q_z$.
By \cref{lem:pmlb-portal-transitions}, each wall index either stays fixed or increases by $1$ in a transition, so the wall offset changes by at most $1$.

Write $2t$ for the wall offset of $Q_1$.
If the final offset is greater than $2t$, let $Q_j$ be the first portal with offset $2t+2$; otherwise, let it be the first portal with offset $2t-2$.
Such a portal exists because the final offset is a different even integer and each transition changes the offset by at most $1$.
Let $Q_i$ be the last portal before $Q_j$ with offset $2t$.
Because offsets change by at most $1$, the intermediate list $Q_{i+1},\ldots,Q_{j-1}$ is nonempty.
The choices of $Q_i$ and $Q_j$ exclude an intermediate visit to either endpoint offset; together with the bound on offset changes, this keeps all intermediate offsets strictly between them.
Thus, all intermediate offsets equal $2t+1$ in the first case and $2t-1$ in the second.
We consider the two cases separately.

Suppose first that $Q_j$ has offset $2t+2$.
The transitions entering and leaving the intermediate list each increase the offset by $1$.
Since neither wall index decreases and each increases by at most $1$, these transitions advance only the pattern wall index.
By \cref{lem:pmlb-portal-transitions}, both cost at least $4d$ and start at an even pattern wall index.
The entering transition ends at an odd pattern wall index, while the leaving transition starts at an even one.
These are the pattern wall indices of the first and last intermediate portals, respectively.
Since this index increases only in steps of $1$, some transition between intermediate portals must increase it from odd to even.
The wall offset is $2t+1$ at both ends of that transition, so the text wall index must also increase by $1$.
At its source, the odd wall offset and odd pattern wall index imply that the text wall index is even.
By \cref{lem:pmlb-portal-transitions}, this transition costs at least $2d+9$.

Suppose now that $Q_j$ has offset $2t-2$.
The entering and leaving transitions each decrease the offset by $1$, so they advance only the text wall index.
By \cref{lem:pmlb-portal-transitions}, both cost at least $4d$ and start at an odd text wall index.
The entering transition ends at an even text wall index, while the leaving transition starts at an odd one.
Thus, the first intermediate text wall index is even and the last is odd.
Some transition between intermediate portals must therefore increase it from even to odd.
The wall offset is $2t-1$ at both ends of that transition, so the pattern wall index must also increase by $1$.
At its source, the odd wall offset and even text wall index imply that the pattern wall index is odd.
By \cref{lem:pmlb-portal-transitions}, this transition again costs at least $2d+9$.

In either case, the entering transition, the identified internal transition, and the leaving transition occur in order along $\A$ and have disjoint interiors.
Any shared endpoints lie in portals and have cost $0$, so adding the three lower bounds counts no positive cost twice.
This gives $\cost(\A)\ge4d+(2d+9)+4d=10d+9$, a contradiction.
\end{proof}

We can now conclude the soundness of the reduction.
\begin{lemma}\label{lem:pmlb-no-so-no}
If $(X,Y)$ is a \No instance of OV, there is no $k$-DTW occurrence of $\hat P$ in $\hat T$.
\end{lemma}
\begin{proof}
Suppose, for a contradiction, that there is an alignment $\A$ of cost at most $k$ from $\hat P$ to a substring of~$\hat T$.
By \cref{lem:pmlb-start-end-portal}, we may assume that its endpoint portals have wall indices $(0,b_0)$ and $(2m,b_1)$, where $b_0\equiv b_1\equiv0\pmod6$.
The starting wall offset is $-b_0\equiv0\pmod6$.
For the ending offset, $m\equiv1\pmod3$ gives $2m\equiv2\pmod6$, and hence $2m-b_1\equiv2\pmod6$.
Thus, both offsets are even, and their different residues modulo $6$ ensure that they are distinct.
By \cref{lem:pmlb-wall-offset}, the alignment must have cost at least $10d+9$, exceeding $k=10d+8$ and giving a contradiction.
\end{proof}

\subsubsection{Completing the Reduction}
\restatementwithproof{\thmovpm*}
\begin{proof}
Let $n\coloneqq|X|$ and $m\coloneqq|Y|$.
By \cref{lem:pmlb-ov-variant}, we obtain an equivalent instance $(X',Y')$ with $n'\coloneqq|X'|=\Oh(n+m)$, $m'\coloneqq|Y'|=\Oh(m)$, and dimension $d'\coloneqq9d$.
Construct $P\coloneqq\hat P$ and $T\coloneqq\hat T$ from this instance as above, with threshold $k\coloneqq10d'+8$.
By \cref{lem:pmlb-yes-so-yes,lem:pmlb-no-so-no}, the original instance is a \Yes instance if and only if $P$ has a $k$-DTW occurrence in $T$.

It remains to verify the parameter bounds and construction time.
Each narrow gadget has length $4d'$, and each vector gadget has length at most $8(d')^2+3d'$.
Since $k\le18d'$, each wall gadget has length at most $\max\{4k+24d',6k+1\}\le110d'$.
The text consists of $n'$ units, each containing two walls, one narrow gadget, and one vector gadget.
Hence,
\[
|T|\le n'\bigl(220d'+4d'+8(d')^2+3d'\bigr)\le300(d')^2n'.
\]
The pattern contains $m'$ vector gadgets, $m'$ narrow gadgets, $2m'-1$ small walls, and two reversed walls.
Thus,
\[
|P|\le m'\bigl(8(d')^2+7d'\bigr)+110d'(2m'+1)\le300(d')^2m',
\]
where the last inequality uses $m'\ge1$ and $d'\ge9$.
Substituting the bounds on $n',m',d'$ gives $|P|=\Oh(md^2)$, $|T|=\Oh((n+m)d^2)$, and $k=\Oh(d)$.
The initial reversed wall begins with $2k\ge k+1$ copies of $\dcharw$, so the first $k+1$ characters of $P$ are equal.

Both strings use the seven-symbol alphabet $\{\dchara,\dcharb,\dchare,\dcharf,\dcharw,\dol_1,\dol_2\}$ and can be constructed in time proportional to their lengths.
Including the preprocessing, the construction takes $\Oh((n+m)d^2)$ time.
\end{proof}

\subsection{Proof of \cref{thm:lb-pmdtw}}
The reduction in \cref{thm:ov-to-pmdtw} gives upper bounds on the string lengths and threshold, whereas \cref{thm:lb-pmdtw} requires prescribed values for all three parameters.
We first establish the padding properties needed to reach those values.
For non-empty strings $P$ and $T$, let
\[
\PMDTW(P,T)\coloneqq \min_{S\ \text{non-empty substring of }T}\DTW(P,S).
\]
The following lemma combines three padding operations: repeating the first text character, extending the initial pattern run, and prepending copies of a fresh symbol to the pattern.
\begin{lemma}[Padding to Prescribed Parameters]\label{lem:pmdtw-padding}
Let $P$ and $T$ be non-empty strings, and let $\kappa\ge0$ be an integer such that the first $\kappa+1$ characters of $P$ are equal.
Given integers $n,m,k$ satisfying
\[
n\ge |T|,
\qquad m\ge |P|+k-\kappa,
\qquad k\ge\kappa,
\]
one can construct strings $P'$ and $T'$ of lengths $m$ and $n$, respectively, in $\Oh(m+n)$ time, using at most one additional alphabet symbol, such that, for the unit-cost function, $\PMDTW(P',T')\le k$ if and only if $\PMDTW(P,T)\le\kappa$.
\end{lemma}
\begin{proof}
We first pad the text by repeating its first character.
Let $b\coloneqq T[1]$, set $q\coloneqq n-|T|\ge0$, and define $T'\coloneqq b^q T$.
Every substring of $T$ remains available in $T'$, so $\PMDTW(P,T')\le\PMDTW(P,T)$.
Conversely, consider an alignment of $P$ to a substring of $T'$, with text coordinates indexed in the whole string $T'$.
Replace each visited vertex $(i,j)$ by $(i,\max\{1,j-q\})$ and delete consecutive repeated vertices.
The resulting path is an alignment to a non-empty substring of $T$, and the mapped text characters are unchanged, so its cost does not increase.
Hence, $\PMDTW(P,T')=\PMDTW(P,T)$.

We next extend the initial pattern run, reserving room for the fresh prefix that will raise the threshold.
Set $\Delta\coloneqq k-\kappa$ and $r\coloneqq m-|P|-\Delta$; both are nonnegative by the hypotheses.
Let $a\coloneqq P[1]$ and define $Q\coloneqq a^rP$.
Fix a non-empty substring $S$ of $T'$.
Every expansion of $Q$ is also an expansion of $P$, since the added copies only lengthen its initial run of $a$'s.
Thus, any pair of same-length expansions witnessing $\DTW(Q,S)\le\kappa$ also witnesses $\DTW(P,S)\le\kappa$.

For the converse, let $X$ and $Y$ be same-length expansions of $P$ and $S$ with $\Ham(X,Y)\le\kappa$.
The first $\kappa+1$ positions of $X$ all contain $a$, so at least one of them must be matched by an $a$ in $Y$; otherwise these positions alone would contribute $\kappa+1$ mismatches.
Duplicate this matched pair $r$ additional times at that position in $X$ and $Y$.
The insertion lengthens the initial run of $X$ by $r$, making it an expansion of $Q$, while the modified $Y$ remains an expansion of $S$.
All inserted pairs match, so the cost remains at most $\kappa$.
Since $S$ was arbitrary, we have $\PMDTW(Q,T')\le\kappa$ if and only if $\PMDTW(P,T')\le\kappa$.

Finally, we raise the threshold by adding a fresh pattern prefix.
Choose a symbol $\hash$ not occurring in $P$ or $T$, and define $P'\coloneqq\hash^\Delta Q$.
In any alignment of $P'$ to a substring of $T'$, each of the $\Delta$ added pattern positions contributes at least one mismatch, since $\hash$ does not occur in $T'$.
Removing all vertices in these first $\Delta$ pattern rows leaves a suffix of the path that still visits every position of $Q$.
Its text coordinates form a non-empty substring of $T'$, so the original cost is at least $\Delta+\PMDTW(Q,T')$.
Conversely, extend an optimal alignment of $Q$ to a substring of $T'$ by aligning the $\Delta$ added copies of $\hash$ to its first text position.
This adds exactly $\Delta$ mismatches, proving $\PMDTW(P',T')=\Delta+\PMDTW(Q,T')$.

Together with $k=\kappa+\Delta$ and the preceding equivalences, this proves the required equivalence.
The constructed strings have lengths $|P'|=\Delta+r+|P|=m$ and $|T'|=n$.
The construction takes $\Oh(m+n)$ time, and $\hash$ is its only additional symbol.
\end{proof}

We also need a lower bound on how many text symbols an algorithm must inspect.
This will ensure that the assumed running time pays for constructing each padded instance and that the target pattern has enough room for the reduction's gadgets.
\begin{lemma}\label{lem:pmdtw-linear-time}
Consider integers $0\le k<m\le n$.
Any Monte Carlo algorithm for the unit-cost \ref{prob:pmdtw} problem, restricted to instances satisfying $|P|=m$, $|T|=n$, and threshold $k$, and having error probability at most $1/3$, must inspect $\Omega(n)$ text symbols.
In particular, its worst-case running time is $\Omega(n)$.
\end{lemma}
\begin{proof}
Fix distinct symbols $\ca$ and $\cb$, and let $P\coloneqq\ca^m$ and $T_0\coloneqq\cb^n$.
Every alignment of $P$ to a substring of $T_0$ mismatches all $m$ pattern positions, so $(P,T_0,k)$ is a \No instance.
For each $j\in[1\dd n]$, let $T_j\coloneqq\cb^{j-1}\ca\cb^{n-j}$.
Each $(P,T_j,k)$ is a \Yes instance, witnessed by the single-character substring $\ca$.

Run the algorithm on $T_0$ and $T_j$ using the same random choices.
Until position $j$ is inspected, the two executions receive identical answers and therefore make identical queries.
In particular, their outputs can differ only if the execution on $T_0$ inspects position $j$.
The error bound implies that the probability of a \Yes output is at least $2/3$ on $T_j$ and at most $1/3$ on $T_0$.
The difference between these probabilities is at most the probability that the two outputs disagree under the shared random choices.
Thus, the two outputs differ with probability at least $1/3$, and position $j$ must be inspected on $T_0$ with probability at least $1/3$.

By linearity of expectation, summing these inspection probabilities over all $j\in[1\dd n]$ gives an expected number of at least $n/3$ distinct text positions inspected on $T_0$.
The worst-case running time is therefore $\Omega(n)$ as well.
\end{proof}

We now combine the reduction with \cref{lem:pmdtw-padding} to prove the lower bound for prescribed pattern lengths and thresholds.
We divide one OV vector list into groups and pad the reduced instances to the required parameters.
\restatementwithproof{\mthmlbpm*}
\begin{proof}
Fix $0<\eps<1$.
Let $c_{\mathrm{red}}\ge1$ be a common upper bound on the constants hidden in the size and threshold bounds of \cref{thm:ov-to-pmdtw}.
Choose $c_{\mathrm{ov}}>0$ as promised by \cref{hyp:ovh}, or its randomized version for randomized algorithms, with $\eps/2$ in place of $\eps$, and set
$
c_\eps \coloneqq 8c_{\mathrm{red}}c_{\mathrm{ov}}.
$

Let $(m_n)_{n=1}^\infty$ and $(k_n)_{n=1}^\infty$ be integer sequences computable in $\Oh(n)$ time and satisfying
$
c_\eps\log n \le k_n \le \tfrac12 m_n \le \tfrac12 n
$
for all sufficiently large $n$.
Suppose, towards a contradiction, that an algorithm $\mathbf A$ solves \cref{prob:pmdtw} on instances with $|P|=m_n$, $|T|=n$, and $k=k_n$ in time $\Oh(n^{1-\eps}\cdot m_n)$.
By \cref{lem:pmdtw-linear-time}, the running time of $\mathbf A$ on instances with text length $n$ is $\Omega(n)$.
Comparing this lower bound with the assumed upper bound gives $n^{1-\eps}\cdot m_n=\Omega(n)$, or equivalently $m_n=\Omega(n^\eps)$.

For sufficiently large $s$, consider an instance of $\OV$ with $|X|=|Y|=s$ and
$
d\coloneqq\ceil{c_{\mathrm{ov}}\log s}.
$
By \cref{thm:ov-to-pmdtw}, keeping all of $X$ and any non-empty subset of $Y$ produces a text of length at most $2c_{\mathrm{red}}sd^2$.
We therefore choose the target text length and the corresponding pattern length and threshold as
\[
n\coloneqq \ceil{2c_{\mathrm{red}}sd^2},
\qquad
m\coloneqq m_n,
\qquad
k\coloneqq k_n.
\]
The sequences' computability lets us obtain $m$ and $k$ in $\Oh(n)$ time.
Since $d=\Theta(\log s)$ and $n=\Theta(sd^2)$, the bound $m=\Omega(n^\eps)$ implies $m=\omega(d^2)$.

We allocate at most $m/4$ positions to each pattern produced by the reduction, leaving the remaining positions for padding.
The reduction's pattern length is at most $c_{\mathrm{red}}d^2$ times the number of vectors retained from $Y$, so we choose the maximum group size
\[
r\coloneqq \left\lfloor \tfrac{m}{4c_{\mathrm{red}}d^2}\right\rfloor.
\]
Because $m=\omega(d^2)$, we have $r=\Theta(m/d^2)$ and $r\ge1$ for sufficiently large $s$.
Moreover, $m\le n=\ceil{2c_{\mathrm{red}}sd^2}\le4c_{\mathrm{red}}sd^2$ for such $s$, so $r\le s$.
Set $t\coloneqq\ceil{s/r}$, and partition $Y$ into $t$ non-empty groups $Y_1,\ldots,Y_t$ of size at most $r$.
The original instance $(X,Y)$ is a \Yes instance if and only if $(X,Y_i)$ is a \Yes instance for at least one $i\in\fragment{1}{t}$.

Fix a group $Y_i$.
Applying \cref{thm:ov-to-pmdtw} to $(X,Y_i)$ yields an equivalent instance $(P_i,T_i,\kappa_i)$ of the unit-cost \ref{prob:pmdtw} problem with
\[
\kappa_i \le c_{\mathrm{red}}d,
\qquad
|P_i| \le c_{\mathrm{red}}rd^2 \le \tfrac14m,
\qquad\text{and}\qquad
|T_i|\le c_{\mathrm{red}}(s+r)d^2\le n.
\]
The reduction also guarantees that the first $\kappa_i+1$ characters of $P_i$ are equal.

For sufficiently large $s$, we have $n\ge s$ and $d=\ceil{c_{\mathrm{ov}}\log s}\le 2c_{\mathrm{ov}}\log n$.
Together with $k\ge c_\eps\log n=8c_{\mathrm{red}}c_{\mathrm{ov}}\log n$, this gives $\kappa_i\le c_{\mathrm{red}}d\le k$.
Since $|P_i|\le m/4$ and $k\le m/2$, we also have $|P_i|+k-\kappa_i\le m/4+m/2\le m$.
Thus, all hypotheses of \cref{lem:pmdtw-padding} hold, and it produces strings $\widehat P_i$ and $\widehat T_i$ of lengths $m$ and $n$, respectively, such that $\PMDTW(\widehat P_i,\widehat T_i)\le k$ if and only if $\PMDTW(P_i,T_i)\le\kappa_i$.
By the reduction, the latter condition holds if and only if $(X,Y_i)$ is a \Yes instance.
We can therefore use $\mathbf A$ on $(\widehat P_i,\widehat T_i,k)$ to decide whether $(X,Y_i)$ is a \Yes instance.
The construction uses unit costs, and padding adds at most one symbol to the reduction's constant-size alphabet.

We perform this construction for every group and answer \Yes if at least one group is accepted.
If $\mathbf A$ is a bounded-error randomized algorithm, repeat each call independently $\Oh(\log(t+1))$ times and take the majority answer to reduce its error to at most $1/(3t)$.
By the union bound, all group answers, and hence the final answer, are correct with probability at least $2/3$.

For each group, the reduction takes $\Oh((s+r)d^2)=\Oh(n)$ time, and \cref{lem:pmdtw-padding} takes $\Oh(m+n)=\Oh(n)$ time.
As $n^{1-\eps}\cdot m=\Omega(n)$, this construction time and the initial computation of $m$ and $k$ are absorbed by the time for the calls to $\mathbf A$.
Allowing for error amplification, the total running time is therefore $\Oh(t\cdot n^{1-\eps}\cdot m\log(t+1))$.

Since $r\le s$ and $r=\Theta(m/d^2)$, the number of groups satisfies $t=\ceil{s/r}=\Oh(s/r)=\Oh(sd^2/m)$.
Together with $t\le s$ and $n=\Theta(sd^2)$, this gives
\[
t\cdot n^{1-\eps}\cdot m\log(t+1)
=\Oh\!\left(s^{2-\eps}d^{4-2\eps}\log s\right).
\]
Since $d=\Theta(\log s)$, the factor $d^{4-2\eps}\log s$ is polylogarithmic in $s$ and is therefore $\Oh(s^{\eps/2})$.
The resulting running time is $\Oh(s^{2-\eps/2})$.
This contradicts the choice of $c_{\mathrm{ov}}$, which rules out such an algorithm for $\OV$ on sets of size $s$ and dimension $d=\ceil{c_{\mathrm{ov}}\log s}$ under the corresponding version of OVH.
\end{proof}

\section{Upper Bound}\label{sec:UB}

In this section we introduce an upper bound for the $\kDTW_\delta$ problem.
We show an algorithm that supports weighted DTW.
Recall that in this setting, we have oracle access to some mismatch-cost function $\delta:\Sigma^2\to \R_{\ge 0}$ such that for every $\sigma\in\Sigma$ we have $\delta(\sigma,\sigma)=0$ and for every $\sigma'\neq \sigma$ we have $\delta(\sigma,\sigma')\ge 1$.
Let $n\coloneqq\max\{|S|,|T|\}$.
Recall that an alignment $\A$ from $S$ to $T$ is a path from $(1,1)$ to $(|S|,|T|)$ in $\AG_{\delta}(S,T)$.
We use $\cost(\A)\coloneqq \sum_{(i,j)\in\A}\delta(S[i],T[j])$.

\MainUB*

\tnote{Moved what was here to the first 10 pages.}

\para{Assumptions and Notations.}
The auxiliary lemmas below assume an integer threshold $1\le k\le n$; the proof of \cref{thm:main-UB} handles arbitrary real thresholds.
When considering the input strings $S$ and $T$, as well as their substrings, we assume that they are given in run-length encoded form.
Notice that this assumption can be ensured after linear-time preprocessing.
In addition, when considering a path $\A$ in the alignment graph $\AG_\delta(S,T)$ with cost $\cost(\A)$, we assume that it is represented in compressed form.
The representation is a sequence of segment descriptors.
A subpath contained in one non-zero block is represented by its
endpoints and, if needed, the turning point described in \cref{lem:path-inside-block}.
A maximal zero-cost subpath is represented by its entry and exit
vertices together with the first and last zero blocks that it visits.
Since two orthogonally adjacent blocks cannot both be zero blocks,
the zero blocks visited by such a subpath form a block diagonal;
we traverse these blocks in some canonical order, e.g., preferring diagonal steps until only horizontal or vertical steps are possible (see \cref{lem:path-inside-block}).
Thus the representation determines a specific alignment and, in
particular, the sequence of blocks visited by the alignment.
Since every visited non-zero block contributes at least one to the
cost, an alignment $\A$ has a compressed representation of
$\Oh(\cost(\A)+1)$ segment descriptors.
As specified in Section~\ref{sec:computational-model}, each query to the mismatch-cost function $\delta$ takes constant time.


Let $A$ be a string.  
Assume the run-length encoding of $A$ is
\[
A = a_1^{x_1} a_2^{x_2} a_3^{x_3} \dots a_m^{x_m},
\]
where $a_i \ne a_{i+1}$ for every $1 \le i < m$.  
Thus, $x_i$ is the length of the $i$-th run and $a_i$ is the character appearing in that run.

For an index $i \in [|A|]$, we denote by $R_A(i)$ the index of the run of the run-length encoding that contains position $i$.  
Formally,
\[
R_A(i) \coloneqq \min \left\{ t \;\middle|\; \sum_{j=1}^{t} x_j \ge i \right\}.
\]
In particular, the character at position $i$ satisfies $A[i] = a_{R_A(i)}$.

\para{Example.}
Let
\[
A \coloneqq \mathtt{aaabbbccaaaa}  = \ca^3 \cb^3 \cc^2 \ca^4.
\]
We have:
\[
R_A(2)=1,\qquad
R_A(4)=2,\qquad
R_A(7)=3,\qquad
R_A(10)=4.
\]
Indeed, position $7$ lies in the block $\cc^2$, hence $R_A(7)=3$ and $A[7]=c=a_3$.

\subsection{An $\Otild(nk)$ Algorithm Computing Optimal Alignment}\label{sec:nk-alignment}

Boneh et al.~\cite{BGMW24} proved that one can compute $\kDTW_\delta(S,T)$ in $\Otild(nk)$ for every mismatch-cost function $\delta$ and two strings $S$ and $T$ with $|\shrink(S)|=\Oh(n)$ (see \cref{lem:nk-computation-ICALP}).
However, their algorithm outputs only the numerical value of $\kDTW_\delta(S,T)$.
For our algorithm we also need to find, in cases $\kDTW_\delta(S,T)\le k$, an optimal alignment $\A$ from $S$ to $T$.
In this section we describe how to extend \cref{lem:nk-computation-ICALP} to output an optimal alignment as well.

We begin with two structural lemmas regarding shortest paths in $\AG_\delta(S,T)$.

\begin{lemma}\label{lem:path-inside-block}
    Given $S$, $T$, and two vertices $u$ and $v$ in the same block $B_{i,j}$ such that $u$ can reach $v$, we can compute $\dist(u,v)$ and a constant-size segment descriptor of a shortest path from $u$ to $v$ in $\Oh(1)$ time.
\end{lemma}
\begin{proof}
    Let $u = (x_u,y_u)$ and $v= (x_v,y_v)$.
    Since every edge increases the $x$ and $y$ coordinates by at most 1, we have that the path from $u$ to $v$ contains at least $1+\max(x_v-x_u,y_v-y_u)$ vertices.
    Every vertex on a path from $u$ to $v$ lies in the same block $B_{i,j}$ and therefore has cost $c(u)$.
    It follows that $\dist(u,v) \ge (1+\max(x_v-x_u,y_v-y_u)) \cdot c(u)$.

    We can greedily construct such a path, by making diagonal steps from $u$ until either the $x$- or the $y$-coordinate matches the corresponding coordinate of $v$.
    From this point on, we proceed only with orthogonal paths in the corresponding direction.
    Formally, the path proceeds in diagonal edges from $u$ to $u'\coloneqq (x_u + \delta,y_u+\delta)$ for $\delta \coloneqq \min(x_v-x_u,y_v-y_u)$.
    Then, the path proceeds from $u'$ to $v$ in orthogonal edges (horizontal if $\delta= x_v-x_u$, vertical otherwise).
    Thus, this path is represented by the segment descriptor $(u,u',v)$, which can be computed in $\Oh(1)$ time given $u$ and $v$.
\end{proof}


Let $v=(x,y)$ be a vertex in block $B_{i,j}$.
For a vertex $u=(x_u,y_u)$ in block $B_{i,j}$, we call an edge $(u',u)$ entering $u$ an external edge if $u'\notin B_{i,j}$. 
We say that $u'$ is an external neighbor of $u$.
We also denote for every input vertex $u=(x_u,y_u)$ from which $v$ is reachable, the cost of a shortest path from $u$ to $v$, $\delta_u \coloneqq (1+\max(x-x_u,y-y_u))\cdot c(u)$.
Notice that $\dist(u,v) = \delta_u$.

\begin{lemma}\label{lem:rle-return-alignment}
Given two strings $S$ and $T$ with $\min(|\shrink(S)|,|\shrink(T)|) = n$, oracle access to a mismatch-cost function $\delta$, if $\kDTW_\delta(S,T) \le k$, we can compute a compressed representation of an optimal alignment $\A^*$ from $S$ to $T$ in $\Otild(n k)$ time.
\end{lemma}
\begin{proof}
    Let us recall the framework of \cref{lem:nk-computation-ICALP}, which was proved by Boneh et al.~\cite{BGMW24}\footnote{The paper of \cite{BGMW24} defines the alignment graph as an edge-weighted graph, whereas in this paper we use a vertex-weighted formulation. While this differs from their presentation, their algorithm can be interpreted within our formulation without any essential changes, and the two viewpoints are interchangeable for our purposes.}.
    For every  $i,j$ such that $|i-j| \le 2k$, the $i,j$ frontier of $\AG_\delta(S,T)$, denoted as $F_{i,j}$, is defined based on the block $B_{i,j} = [i_1\dd i_2]\times [j_1\dd j_2]$.
    The frontier $F_{i,j}$ consists of all the vertices $v=(x,y)$ such that if $v\in B_{a,b}$ then $|a-b|\le 2k$  and either $x < i_2$ and $y= j_2$, or $x= i_2$ and $y \in [j_1-1\dd j_2]$, or $x > i_2$ and $y=j_1-1$.
    For convenience, we regard the distance to any vertex belonging to a block $B_{i,j}$ with $|i-j|>2k$ as infinity; such a vertex is therefore never queried in $D$.

    The algorithm of \cite{BGMW24} initiates a data structure $D$ and iterates all pairs $i,j$ with $|i-j| \le 2k$ in increasing counting order (i.e., all pairs with $j=1$ in increasing order of $i$, followed by all pairs with $j=2$ in increasing order of $i$, and so on).
    When the pair $(i,j)$ is processed, the data structure $D$ is modified to support distance queries from $(1,1)$ to any vertex in $F_{i,j}$ in $\Otild(1)$ time.
    
    In order to recover an alignment, we run a partially-persistent version of the algorithm of \cite{BGMW24}.\footnote{
Although \textsc{Lookup} in~\cite{BGMW24} invokes \textsc{Flush}, Invariant~3 shows that the queried value can instead be obtained from the maintained segment and pending-ray representations by two predecessor queries and evaluating their minimum, without modifying the data structure.
We make the ordered structures storing these representations partially persistent by standard path copying, and save their roots after every processed block.
Since the algorithm of~\cite{BGMW24} performs only $\Otild(nk)$ primitive operations on these structures, persistence increases their total cost by only a polylogarithmic factor, and each historical \textsc{Lookup} also takes polylogarithmic time.
Thus all versions needed for backtracking, as well as all historical distance queries, are supported within $\Otild(nk)$ total time.

}
    We back-track the alignment as follows.
    Assume that we already know that there is an optimal alignment visiting block $B_{i,j}$ in output vertex $v$, and we also know $\dist((1,1),v)$.
    We will show how to find a block $B_{i',j'}$ with $i'+j' < i + j$ and an output vertex $v'$ of $B_{i',j'}$ such that there is an optimal alignment from $(1,1)$ to $v$ visiting $v'$.
    We also find an optimal path from $v'$ to $v$, and the value of $\dist((1,1),v')$.
    We store each such path using the segment descriptors defined above.
    Consecutive zero-cost portions are merged into a single zero-diagonal descriptor.
    It is easy to see that if we start this process with $v\coloneqq(|S|,|T|)$ in the block $B_{R_S(|S|),R_T(|T|)}$, the concatenation of all $v'$-to-$v$ paths is an optimal alignment.

    Let us describe how to find $v'$ and $B_{i',j'}$.
    Denote $B_{i,j} = [i_1\dd i_2]\times [j_1\dd j_2]$.
    We retrieve the version $D$ just before the pair $(i,j)$ was processed.
    Notice that this version of $D$ supports distance queries from $(1,1)$ to every node of the form $(i_1-1,j_3)$ for $j_3\in[j_1\dd j_2]$, to the node $(i_1-1,j_1-1)$, and to every node of the form $(i_3,j_1-1)$ for $i_3\in[i_1\dd i_2]$ (except for non-existent vertices with one of the coordinates being $0$).
    If $B_{i,j} = B_{1,1}$, we know that the path must start in $(1,1)$.
    We create such a path using \cref{lem:path-inside-block}.
    
    If $B_{i,j}$ is a zero block, we know from \cref{lem:first_and_last_in_0_block} that there is an optimal path to $v$ visiting $(i_1,j_1)$.
    We check each of the 3 candidates for the previous step from a block outside of $B_{i,j}$ to $(i_1,j_1)$ (specifically, $(i_1-1,j_1)$, $(i_1,j_1-1)$, and $(i_1-1,j_1-1)$), all candidates belonging to blocks $B_{a,b}$ with $|a-b|\le 2k$ are available at the version of $D$ we have at hand, while the remaining candidates are regarded as having distance infinity.
    So, we can obtain the distance from $(1,1)$ to each of them in $\Otild(1)$ time and add the cost of the extra step into $(i_1,j_1)$.
    We pick as $v'$ any candidate for which this sum is exactly $\dist((1,1),v)$.
    We retrieve an optimal path from $v'$ to $v$ via $(i_1,j_1)$ by taking the edge connecting $v'$ to $(i_1,j_1)$ and attaching an in-block optimal path computed via \cref{lem:path-inside-block}.

    Otherwise, $B_{i,j}= [i_1\dd i_2]\times [j_1\dd j_2]$ is a non-zero block.
    Denote $v = (x,y)$.
    We run a different procedure depending on the relationship between $x - i_1$ and $y-j_1$, i.e., between the horizontal and the vertical distances of $v$ from the input boundary of $B_{i,j}$.
    We present the algorithm for the case $x-i_1 \le y-j_1$, the complementary case is symmetrical.
    Iterate the input vertices of $B_{i,j}$ starting with $(i_1, y)$ in clockwise order (i.e. decreasing the second coordinate until $(i_1,j_1)$ is reached, and then increasing the first coordinate until $(x,j_1)$ is reached).
    Let $u$ be a vertex iterated in this manner.

    When iterating the vertex $u$, we check every external neighbor $u'$ of $u$ as a candidate for being the predecessor of $u$ in a shortest path from $(1,1)$ to $v$ going through $u$.
    Namely, we check if $\dist((1,1),u')  +\dist(u,v) = \dist((1,1),v)$ (recall that $\AG_\delta(S,T)$ is a vertex-weighted graph, so the cost of a path with edge $(u',u)$ is exactly the sum of the weight of the prefix up to $u'$ and the suffix from $u$).
    If a candidate $u'$ satisfying this equation is found, return $u'$ with the block containing $u'$, and the distance $\dist((1,1),u')$.
    Note that for every external neighbor $u'$ of $u$ that belongs to a block $B_{a,b}$ with $|a-b|\le 2k$, we have access to $\operatorname{dist}((1,1),u')$ via $D$; all other external neighbors are regarded as having distance infinity.
    The path from $u'$ to $v$ is constructed by starting with the edge $(u',u)$ and proceeding with an in-block shortest path from $u$ to $v$ computed by \cref{lem:path-inside-block}.
    
    The correctness of the backtracking procedure is clear: for a zero block, it follows directly from \cref{lem:first_and_last_in_0_block} and for a non-zero block we simply iterate the vertices until a valid candidate is found (there must be some input vertex that is first visited by an optimal path to $v$).

    \para{Complexity.}
    On top of the $\Otild(nk)$ time for running the algorithm of \cref{lem:nk-computation-ICALP}, we apply the backtracking.
    Since we backtrack from a block $B_{i,j}$ to a block $B_{i',j'}$ that precedes it (i.e. $i'\le i$, $j'\le j$, and $B_{i,j} \neq B_{i',j'}$), it is clear that we never process the same block twice.
    A zero block is clearly processed in $\Otild(1)$ time.
    As for a non-zero block $B_{i,j}=[i_1\dd i_2]\times [j_1\dd j_2]$, let $v=(x,y)$ and let $u^*=(x',y')$ be the first input vertex of $B_{i,j}$ visited by the optimal alignment to $v$ found by the algorithm.
    The algorithm iterates the input vertices of $B_{i,j}$ in clockwise order (recall that we analyze the case where $x-i_1 \le y-j_1$).
    Each iterated vertex is processed in $\Otild(1)$ time by applying $\Oh(1)$ queries to $D$ and computing simple arithmetic after each query.
    Clearly, when the iteration over the input vertices of $B_{i,j}$ reaches $u^*$, the processing of $B_{i,j}$  terminates without processing further vertices.
    We claim that the number of input vertices inspected before $u^*$ when $B_{i,j}$ is processed is $\Oh(\delta_{u^*})$ where (as defined above) $\delta_u=\dist(u,v)$.
    The $u^*$-to-$v$ paths found throughout the backtracking process belong to distinct non-zero blocks and form subpaths of the returned optimal alignment.
    Hence, the sum of their costs, and therefore the sum of the corresponding values $\delta_{u^*}$, is at most $k$.
    Thus, the total time spent processing non-zero blocks is $\Otild(k)$. 
    Moreover, the returned alignment contains $\Oh(k+1)$ segment descriptors: there are at most $k$ visited non-zero blocks, and maximal zero-cost portions between them are represented by one descriptor each.
    Hence constructing the compressed output does not affect the claimed running time.

    It only remains to prove that, indeed, for every input vertex $u$ the iteration will reach $u$ (or terminate) within $\Oh(\delta_u)$ steps.
    Let $u_1,u_2,\dots, u_I$ be the input vertices starting from $u_0 \coloneqq (i_1,y)$ in clockwise order.
    We show that $\delta_{u_z} \ge z/2=\Omega(z)$ for every input vertex $u_z$ from which $v$ is reachable.
    Let $c:=c(v)$, which is the common vertex cost of $B_{i,j}$.
    Since $B_{i,j}$ is non-zero, $c\ge1$.

    For $z \in [1\dd x-i_1]$, we have $u_z = (i_1,y-z)$ , and therefore $\delta_{u_z} = (1+\max\{x-i_1,z\})c \ge x-i_1\ge z$.

    For $z\in [x-i_1\dd y-j_1]$ we have $u_z = (i_1,y-z)$, and therefore $\delta_{u_z} = (1+\max\{x-i_1,z\})c \ge z$.

    Finally, for $z\in [y-j_1\dd y-j_1 +x-i_1]$ we have $u_z = (i_1 + z-(y-j_1) ,j_1)$ and hence  $\delta_{u_z} = \bigl(1+\max\{x-i_1-(z-(y-j_1)),\,y-j_1\}\bigr)c \ge y-j_1 \ge z/2$, where the last inequality follows from $x-i_1\le y-j_1$.
    Also $\delta_{u_0}\ge1$, so the initial inspected vertex is covered as well.
    \end{proof}

\subsection{Approximating $L_p(S)$}

In this section we show how to compute, given a string $S$ of length $n$ and some integer $p$, a $2$-approximation of $L_p(S)$ in $\Otild(n)$ time.
Recall that $L_p(S)$ is the length of the longest substring of $\shrink(S)$ with period at most $p$.
A \emph{maximal repetition} in a string $X$ is a substring occurrence $R$ with $|R|\ge2\per(R)$ that cannot be extended by one character to either side in $X$ while retaining $\per(R)$ as a period.
We use a near-linear time algorithm that computes all maximal repetitions of a given string, proved by Ellert et al.~\cite{EGG23} (see also Bannai et al.~\cite{BIINTT17}, Ellert and Fischer~\cite{EF21}).


\begin{lemma}[\cite{EGG23}]\label{lem:all_runs}
Given a string $X$ of length $n$, all maximal repetitions of $X$, together with their
periods, can be computed in $\Otild(n)$ time.
\lipicsEnd
\end{lemma}

We show that one can compute a $2$-approximation of $L_p(S)$ in $\Otild(n)$ time.

\begin{lemma}\label{lem:approx_LpS}
Given a string $S$ of length $n$ and an integer $p\ge1$, one can compute  a value $\widehat L_p(S)$ such that
$L_p(S)\le \widehat L_p(S)\le 2L_p(S)$. 
The algorithm takes
$\Otild(n)$ time.
\end{lemma}

\begin{proof}
Let $X\coloneqq\shrink(S)$ and $m\coloneqq|X|$.
If $p\ge m$, return $m$.
Otherwise, compute all maximal repetitions of $X$ using \cref{lem:all_runs}.
Let $M$ be the maximum length of a maximal repetition of $X$ with period at most $p$, and set $M\coloneqq0$ if no such repetition exists.
Return
$\widehat L_p(S)\coloneqq\max\{M,\min\{2p,m\}\}$.

Let $Y$ be a longest substring of $X$ with period $q\le p$, hence $|Y|=L_p(S)$.
Clearly $|Y|\le m$.
If $|Y|<2q$, then $|Y|<2p$ and therefore $L_p(S)=|Y|\le \min\{2p,m\}\le\max\{M,\min\{2p,m\}\}$.
Otherwise, $Y$ is contained in a maximal repetition of $X$ with period at most $q\le p$, and hence, $L_p(S)=|Y|\le M\le\max\{M,\min\{2p,m\}\}$.
Thus $L_p(S)\le\widehat L_p(S)$.

On the other hand, if $p\ge m$ we have $\widehat L_p(S)=m$ and $L_p(S)=m$ as well, and in particular $\widehat L_p(S)\le2L_p(S)$.
Otherwise, $p<m$.
In this case $M\le L_p(S)$. 
Moreover, since $p$ is an integer and $p<m$, any substring of $X$ of length $p$ has period at most $p$, and hence $p\le L_p(S)$.
Therefore, $\widehat L_p(S)=\max\{M,\min\{2p,m\}\}\le \max\{L_p(S),2L_p(S)\}=2L_p(S)$.
\end{proof}

\subsection{From Approximate Alignment to an Optimal Alignment}\label{sec:approx-to-optimal}

In this section we show how to efficiently compute an optimal alignment from $S$ to $T$, given some alignment $\A$. 
The time complexity depends on $\cost(\A)$.
This algorithm is the most technical part of our algorithm. 

We begin by showing the monotonicity of $L_k(S)$ with respect to both the parameter $k$ and the input string (with respect to the superstring relation).

\begin{observation}\label{lem:Lk-monotone}
    Let $k_1<k_2$ be two integers and let $S$ be a string.
    Then $L_{k_1}(S)\le L_{k_2}(S)$.
\end{observation}
\begin{proof}
Let $R$ be a longest substring of $\shrink(S)$ with period at most $k_1$.
Then $|R|=L_{k_1}(S)$.
Notice that the period of $R$ is at most $k_1<k_2$.
Hence, $R$ is a substring of $\shrink(S)$ with period at most $k_2$.
Therefore, $L_{k_2}(S)\ge |R| = L_{k_1}(S)$, as required.
\end{proof}

\begin{observation}\label{lem:Lk-substring}
    Let $S$ be a string and let $S[i\dd j]$ be a substring of $S$.
    For every $k$, we have $L_{k}(S[i\dd j])\le L_{k}(S)$.
\end{observation}
\begin{proof}
Let $R$ be a longest substring of $\shrink(S[i\dd j])$ with period at most $k$, so $|R|=L_k(S[i\dd j])$.
Since $\shrink(S[i\dd j])$ is a substring of $\shrink(S)$, the string $R$ is also a substring of $\shrink(S)$.
Hence, $R$ is also a substring of $\shrink(S)$ with period at most $k$.
Therefore, $L_k(S)\ge |R|=L_k(S[i\dd j])$.
\end{proof}

In the following lemma we show that two cheap long paths in $\AG_\delta(S,T)$ where the corresponding substring of $S$ has a long period, must either intersect at some zero-block or induce a mismatch.
Later, we will use this to show that two cheap paths which are even longer must intersect at some zero-block, since they cannot afford to have too many mismatches, as each of them has at most $k$ mismatches.

\begin{lemma}\label{lem:aper-intersect-or-mismatch}
    Let $S$ and $T$ be two strings.
    Let $T[\alpha_1\dd \beta_1]$ and $T[\alpha_2\dd \beta_2]$ be two substrings of $T$ such that $|R_T(\alpha_1) -R_T(\alpha_2)|\le 4k$.
    Let $\A_1$ and $\A_2$ be alignments of cost at most $2k$, aligning $S$ to $T[\alpha_1\dd \beta_1]$ and to $T[\alpha_2\dd \beta_2]$, respectively.
    Let $S'\coloneqq\shrink(S)$ and let $i_1,i_2$ be indices such that $P\coloneqq S'[i_1\dd i_2]$ is a string with $\per(P) > 12k$.

    There is a block $B_{i^*,j^*}$ with $i^* \in [i_1\dd i_2]$ satisfying one of the following.
    \begin{enumerate}
        \item $B_{i^*,j^*}$ is a zero block and $B_{i^*,j^*}\cap \A_1\cap \A_2 \ne \emptyset$.
        \item $B_{i^*,j^*}$ is a non-zero block and $B_{i^*,j^*}\cap (\A_1\cup \A_2) \ne \emptyset$.
    \end{enumerate}
\end{lemma}
\begin{proof}
    Assume that both $\A_1$ and $\A_2$ do not visit a non-zero block $B_{i,j}$ with $i \in [i_1\dd i_2]$.
    Let $B_{i_1,j_1}$ be the first block visited by $\A_1$ with $i$ value $i_1$ and let $B_{i_1,j_2}$ be the first block visited by $\A_2$ with $i$ value $i_1$.
    We will prove that $j_1=j_2$.

    We claim that $j_1$ is an occurrence of $P$ in $T'\coloneqq\shrink(T)$.
    The proof that $j_2$ is an occurrence of $P$ in $T'$ is identical.
    Recall that every block $B_{i,j}$ visited by $\A_1$ with $i\in[i_1\dd i_2]$ is a zero block.
    Since two orthogonally adjacent blocks cannot be both zero blocks, it holds that all blocks visited by $\A_1$ with $i$ value in $[i_1\dd i_2]$ are zero blocks of the form $B_{i_1 + x,j_1+x}$, i.e., $\A_1$ only makes diagonal block steps in the horizontal block strip $[i_1\dd i_2]$.

    
    It follows that $T'[j_1\dd j_1 + |P|)=S'[i_1\dd i_2]=P$ as required.
    Similarly, $j_2$ is also an occurrence of $P$ in $T'$.
    Since $\A_1$ aligns $S$ and $T[\alpha_1\dd \beta_1]$ for a cost of at most $2k$, it follows from \cref{lem:kband} that $|j_1 - (i_1 + R_T(\alpha_1))+1| \le  4k$.
    Due to the same reasoning, we have $|j_2 - (i_1 + R_T(\alpha_2))+1| \le  4k$.
    It follows that $|j_1 - j_2| \le |R_T(\alpha_1) - R_T(\alpha_2)| + 8k \le 12k$.
    We have shown that $j_1$ and $j_2$ are occurrences of $P$ in $T'$ with $|j_1 -j_2| \le 12k$.
    It follows from \cref{obs:two-occs-period} that either $j_1 = j_2$ or $\per(P)\le |j_1-j_2| \le 12k$. 
    Since we know that $\per(P) > 12k$, it must be the case that $j_1 =j_2$, as required. 
    Finally, since $\A_1$ and $\A_2$ only make diagonal block steps in the horizontal block strip $[i_1\dd i_2]$, it must be that both $\A_1$ and $\A_2$ visit the right-bottom vertex of $B_{i_1,j_1}$, and therefore $B_{i_1,j_1}\cap\A_1\cap\A_2\ne\emptyset$.
\end{proof}

Here we show that a string $S$ contains $4k+1$ disjoint substrings, each 
with period larger than $12k$.
These substrings satisfy the requirements on $P$ in \cref{lem:aper-intersect-or-mismatch}.
\begin{lemma}\label{lem:aperpartition}
  Let $S$ be a string with no consecutive equal symbols (i.e., $S=\shrink(S)$) and let $k$ be an integer such that $|S| \ge (4k+2) (L_{12k}(S)+24k)$.
  There are substrings $S[i_1\dd j_1],S[i_2\dd j_2],\ldots ,S[i_{4k+1}\dd j_{4k+1}]$ of $S$ satisfying the following.
  For every $x\neq y \in [1\dd 4k+1]$ we have $[i_x\dd j_x] \cap [i_y\dd j_y] = \emptyset$.
  For every $x\in [1\dd 4k+1]$ it holds that $\per(S[i_x\dd j_x])> 12k$.
\end{lemma}
\begin{proof}
    Let $L\coloneqq L_{12k}(S)$.
    Since $|S|\ge(4k+2)(L+24k)>(4k+1)(L+1)$, one can choose $4k+1$ consecutive pairwise disjoint substrings of $S$, each of length $L+1$.
    By definition of $L_{12k}(S)$, each of these substrings has period larger than $12k$.
\end{proof}

Combining \cref{lem:aper-intersect-or-mismatch,lem:aperpartition}, we show that two alignments of $S$ with a small number of mismatches must intersect in some zero block, provided that their length is sufficiently large (compared to $L_{\Oh(k)}(S)$).

\begin{lemma}\label{lem:large-s-alignments-cross}
    Let $S$ and $T$ be two strings and let $k$ be an integer such that $k\ge 1$.
    Let $T[\alpha_1\dd \beta_1]$ and $T[\alpha_2\dd \beta_2]$ be two substrings of $T$ with $|R_T(\alpha_1) - R_T(\alpha_2)|\le 4k$.
    Let $\A_1$ and $\A_2$ be two alignments of cost at most $2k$ each, aligning $S$ to $T[\alpha_1\dd \beta_1]$ and to $T[\alpha_2\dd \beta_2]$, respectively.
    Let $S'\coloneqq\shrink(S)$.

    If $|S'| \ge (4k+2)(L_{12k}(S) + 24k)$,
    it must be that there is a zero block $B_{i^*,j^*}$ such that $B_{i^*,j^*}\cap \A_1\cap \A_2 \ne \emptyset$.
\end{lemma}
\begin{proof}
    By applying \cref{lem:aperpartition} on $S'$, we obtain $4k+1$ substrings $A_1\coloneqq S' [i'_1\dd j'_1],A_2\coloneqq S'[i'_2\dd j'_2],\ldots, A_{4k+1}\coloneqq S' [i'_{4k+1}\dd j'_{4k+1}]$ such that for every $x\neq y \in [4k+1]$ it holds that $\per(A_x) > 12k$ and $ [i'_x\dd j'_x] \cap [i'_y\dd j'_y] =\emptyset$.

    For every $x\in [4k+1]$, it follows from \cref{lem:aper-intersect-or-mismatch} that there is a block $B_{i^*_x,j^*_x}$ with $i^*_x \in [i'_x\dd j'_x]$ that is either a zero block visited both by $\A_1$ and $\A_2$, or a non-zero block visited by $\A_1$ or by $\A_2$.
    We claim that for some $y\in [4k+1]$ it must be the former that holds, rather than the latter.

    Assume to the contrary that for every $x\in [4k+1]$ we have that $B_{i^*_x,j^*_x}$ is a non-zero block visited by either $\A_1$ or $\A_2$.
    Notice that the blocks $B_{i^*_x,j^*_x}$ are pairwise disjoint, since each has its $i$  value contained in a different interval $[i'_x\dd j'_x]$, disjoint from other intervals.
    It follows that for each $x\in [4k+1]$, the block $B_{i^*_x,j^*_x}$ induces a unique cost of at least 1 either to $\A_1$ or to $\A_2$.

    We obtain that the total cost of $\A_1$ and $\A_2$ is at least $4k+1$, which leads to the cost of one of the two being more than $2k$, a contradiction.

    We have shown that for some $y\in [4k+1]$ there is a zero block $B_{i^*_y,j^*_y}$ with $i^*_y\in [i'_y\dd j'_y] \subseteq [1\dd |S'|]$ satisfying $B_{i^*_y,j^*_y}\cap \A_1 \cap \A_2 \neq \emptyset$, as required.
\end{proof}


We are now ready to prove the main lemma of this subsection, which shows how to use an approximate alignment from $S$ to $T$ to compute an optimal alignment.

\begin{lemma}\label{lem:dtw-given-approx}
    There exists an algorithm that, given two strings $S$ and $T$ of total length $n$, oracle access to a mismatch-cost function $\delta$, an alignment $\A$ from $S$ to $T$  (given in the compressed representation defined above) of cost at most $t\in\mathbb N_{>0}$ and a bound $M\ge L_{12t}(S)$,  computes $\DTW_\delta(S,T)$ and outputs an optimal alignment $\A^*$  from $S$ to $T$ (in the compressed representation defined above).
    The algorithm runs in $\Otild(t^2 M)$ time.
\end{lemma}
\begin{proof}
\newcommand{\tDTW}{t\textit{-}\mathsf{DTW}}


    If $\cost(\A)=0$, return $\A$, which is clearly optimal.
    From now on, we assume $\cost(\A)>0$.

    Let $S' \coloneqq \shrink(S)$ and $T' \coloneqq \shrink(T)$ with respective lengths $n'$ and $m'$.
    Denote $L\coloneqq (4t+2)(M+24t)$.
    Consider the indices $i_1 \coloneqq \floor{n'/2 - L -3t}$ and $i_2 \coloneqq \ceil{n'/2 + L + 3t}$.

    If $n' \le 2L+8t$, we apply \cref{lem:rle-return-alignment}, with $k\coloneqq t$ to compute $\DTW_\delta(S,T)$ and an optimal alignment.
    If $n'\le 12t$, \cref{lem:rle-return-alignment} takes $\Otild(tn')=\Otild(t^2)=\Otild(t^2M)$.
    Otherwise, $n'>12t$, and hence $L_{12t}(S)\ge 12t=\Omega(t)$, so $M=\Omega(t)$. Therefore the running time is $\Otild(t(L+t))=\Otild(t^2M+t^3)=\Otild(t^2 M)$.
    %
    %
    Notice that since $\cost(\A)\le t$ it must be that $\DTW_\delta(S,T)\le t$ and therefore the algorithm of \cref{lem:rle-return-alignment} will return an alignment.

    Otherwise, consider the blocks $B_{i_1,j_1}$ and $B_{i_2,j_2}$ as the first blocks visited by $\A$ with $i$-values $i_1$ and $i_2$, respectively. 
    Furthermore, let    $B_{i'_1,j'_1}$ be the first zero block visited by $\A$ weakly after $B_{i_1,j_1}$ (i.e., if $B_{i_1,j_1}$ is a zero block then $B_{i'_1,j'_1} = B_{i_1,j_1}$).
    Similarly, let $B_{i'_2,j'_2}$ be the first zero block visited by $\A$ weakly after $B_{i_2,j_2}$.
    Notice that since  $\cost(\A)\le t$, it must hold that $B_{i'_1,j'_1}$ and $B_{i'_2,j'_2}$ are within $t+1\le 2t$ blocks from $B_{i_1,j_1}$ and from $B_{i_2,j_2}$, respectively.
    In particular, $i'_1 - i_1 \le 2t$ and $i'_2 - i_2 \le 2t$, since along $\A$ the $i$-index of consecutive blocks can increase by at most $1$.
    Notice that all these blocks can be found by scanning the compressed representation of $\A$, since every segment descriptor specifies the sequence of blocks that it visits.

    Let $S_c$ be the substring of $S$ implied by $S'[i'_1\dd i'_2]$.
    Similarly, let $T_c$ be the substring of $T$ implied by $T'[j'_1\dd j'_2]$.
Let $v_1$ be the first vertex of the zero block $B_{i'_1,j'_1}$,
and let $v_2$ be the last vertex of the zero block $B_{i'_2,j'_2}$.
Let $w_1$ and $w_2$ be vertices of $\A$ in
$B_{i'_1,j'_1}$ and $B_{i'_2,j'_2}$, respectively.
By \cref{lem:first_and_last_in_0_block}, there is a path
$P_1$ from $(1,1)$ to $v_1$ whose cost is at most the cost of
the prefix of $\A$ ending at $w_1$.
Symmetrically, there is a path $P_2$ from $v_2$ to
$(|S|,|T|)$ whose cost is at most the cost of the suffix of
$\A$ starting at $w_2$.
Since $B_{i'_1,j'_1}$ and $B_{i'_2,j'_2}$ are zero blocks,
there is a zero-cost path inside $B_{i'_1,j'_1}$ from $v_1$ to
$w_1$, and a zero-cost path inside $B_{i'_2,j'_2}$ from $w_2$
to $v_2$.
Together with the subpath of $\A$ from $w_1$ to $w_2$, these
paths form an alignment from $S_c$ to $T_c$ of cost at most $t$.
Notice that using the compressed representation of $\A$ one can easily compute $P_1$ and $P_2$ using only local edits on the segment-descriptors.

We use \cref{lem:rle-return-alignment}, with $k\coloneqq t$, to compute
an optimal alignment $\A_c^*$ from $S_c$ to $T_c$.
In particular, $\cost(\A_c^*)\le t$.
We then define $\bar\A\coloneqq P_1\circ\A_c^*\circ P_2$.
The alignment $\bar\A$ is an alignment from $S$ to $T$ satisfying $\cost(\bar\A)\le\cost(\A)\le t$.
The alignment $\bar\A$ is represented by concatenating the corresponding segment-descriptor sequences.



     Now, let $B_{i_c,j_c}$ be the first zero block visited by $\A_c^*$ with $i$-index in $[n'/2-t\dd n'/2+t]$.
     Again, notice that such a block exists because the total cost of $\A_c^*$ is bounded by $t$ and $\A_c^*$ visits at least $2t$ blocks with $i$ values in $[n'/2 - t\dd n'/2 + t]$.
     Since every visited non-zero block contributes at least $1$ to the cost and $2t>t$, at least one of these blocks must be a zero block.

     We make the following claim.
     \begin{claim}\label{clm:visit-center}
         There is an optimal alignment from $S$ to $T$ that visits the block $B_{i_c,j_c}$.
     \end{claim}
     \begin{claimproof}
Let $\A^{\mathrm{opt}}$ be an optimal alignment from $S$ to $T$.
Since $\cost(\A^{\mathrm{opt}})\le\cost(\A)\le t$, both
$\A^{\mathrm{opt}}$ and $\bar\A$ have cost at most $t$.
Recall that the subpath of $\bar\A$ between $v_1$ and $v_2$
is precisely $\A_c^*$.

We first show that $\A^{\mathrm{opt}}$ and $\A_c^*$ have a
common zero vertex strictly before the $i_c$th run of $S$.
Let $S_L$ be the substring of $S$ implied by $S'[i'_1\dd i_c-1]$.
Consider the restrictions of $\A^{\mathrm{opt}}$ and $\bar\A$
to $S_L$.
These restrictions are alignments of $S_L$ to two substrings
of $T$, and the latter restriction is a sub-alignment of
$\A_c^*$.

Let $\alpha_1$ and $\alpha_2$ be the first positions of the
two corresponding substrings of $T$.
Since both alignments are subpaths of global alignments of cost
at most $t$, \cref{lem:kband} gives
\[
    |R_T(\alpha_1)-i'_1|\le 2t
    \qquad\text{and}\qquad
    |R_T(\alpha_2)-i'_1|\le 2t.
\]
Hence, $|R_T(\alpha_1)-R_T(\alpha_2)|\le 4t$.

Furthermore, since $i'_1-i_1\le2t$, we have $|\shrink(S_L)|
    =i_c-i'_1\ge (n'/2-t)-(i_1+2t)\ge L$.
By \cref{lem:Lk-substring}, $L_{12t}(S_L)\le L_{12t}(S)\le M$,
and therefore
$|\shrink(S_L)|\ge(4t+2)\bigl(L_{12t}(S_L)+24t\bigr)$.
Thus, we can apply \cref{lem:large-s-alignments-cross}, with
$k\coloneqq t$, to the two restricted alignments.
It follows that there is a zero block $B_{i_L,j_L}$ with $i_L\in[i'_1\dd i_c-1]$ containing a vertex $u_L$ visited by both
$\A^{\mathrm{opt}}$ and $\A_c^*$.

One can apply symmetrical reasoning to prove that there is also a zero block $B_{i_R,j_R}$ with $i_R \in [i_c+1\dd i'_2]$ containing a vertex $u_R$ visited by both $\A^{\mathrm{opt}}$ and $\A_c^*$.


Let $P$ be the subpath of $\A^{\mathrm{opt}}$ from $u_L$ to
$u_R$, and let $Q$ be the subpath of $\A_c^*$ from $u_L$ to
$u_R$.
Since $\A^{\mathrm{opt}}$ is an optimal alignment, $P$ is a
shortest path from $u_L$ to $u_R$.
Similarly, since $\A_c^*$ is an optimal alignment from $S_c$
to $T_c$, $Q$ is a shortest path from $u_L$ to $u_R$.
Indeed, by monotonicity, every path from $u_L$ to $u_R$ is
contained in the subgraph induced by $S_c$ and $T_c$.
Consequently, $\cost(P)=\cost(Q)$.

We replace $P$ in $\A^{\mathrm{opt}}$ by $Q$.
The resulting alignment has the same cost as
$\A^{\mathrm{opt}}$, and is therefore optimal.
Finally, since $i_L<i_c<i_R$ and $Q$ is a subpath of
$\A_c^*$, the path $Q$ visits $B_{i_c,j_c}$.
Hence, the resulting optimal alignment visits
$B_{i_c,j_c}$, as required.
\end{claimproof}

    We recursively apply the algorithm to find $\DTW_\delta(S_1,T_1)$ and $\DTW_\delta(S_2,T_2)$ for the strings $S_1$, $S_2$ , $T_1$, and $T_2$ defined as follows.
    \begin{enumerate}
        \item $S_1$ is the substring of $S$ implied by $S'[1\dd i_c]$,
        \item $S_2$ is the substring of $S$ implied by $S'[i_c\dd n']$,
        \item $T_1$ is the substring of $T$ implied by $T'[1\dd j_c]$, and
        \item $T_2$ is the substring of $T$ implied by $T'[j_c\dd m']$.
    \end{enumerate}
    We need to equip each recursive call with an approximate alignment.
    We partition $\bar\A$ at $B_{i_c,j_c}$ into alignments $\A_1$ and $\A_2$, aligning $S_1$ to $T_1$ and $S_2$ to $T_2$, respectively, as follows.
    Let $f$ and $\ell$ be the first and last vertices of $B_{i_c,j_c}$, respectively, and let $w$ be some vertex in $\bar A\cap B_{i_c,j_c}$.
    Define $\A_1$ as the prefix of $\bar\A$ ending at $w$, followed by a zero-cost path from $w$ to $\ell$ inside $B_{i_c,j_c}$, and $\A_2$ as a zero-cost path from $f$ to $w$ inside $B_{i_c,j_c}$, followed by the suffix of $\bar\A$ starting at $w$.
    Since $B_{i_c,j_c}$ is a zero-block, $\cost(\A_1)+\cost(\A_2)=\cost(\bar \A )\le\cost(\A)\le t$, where $\A_1$ aligns $S_1$ and $T_1$ and $\A_2$ aligns $S_2$ and $T_2$ (notice that $B_{i_c,j_c}$ is a zero-block, and therefore contributes nothing to  $\cost(\A_1)$ and to $\cost(\A_2)$).
    Notice that $\A_1$ and $\A_2$ can be obtained by local edits on segment-descriptors.
    We make each of the recursive calls using the respective alignment, and use $\ceil{\cost(\A_1)}$ and $\ceil{\cost(\A_2)}$, respectively as the upper bound on the cost of each alignment. 
    Notice that $\ceil{\cost(\A_1)}\le t$ and  $\ceil{\cost(\A_2)}\le t$ since $t$ is an integer.
    In the real RAM model, we compute these ceilings by binary search over the integers in $[0\dd t]$.
    Since this is a recursive case, $n'>2L+8t$, so the searches take $\Oh(\log n)$ time, absorbed by the stated bound.
    If $\cost(\A_\ell)=0$ (for $\ell\in\{1,2\}$) the algorithm does not recurse since $A_\ell$ is already optimal.  
    Finally, we pass $M$ as the upper bound to both recursive calls since $L_{12\ceil{\cost(\A_\ell)}}(S_\ell)\le L_{12t}(S_\ell)\le L_{12t}(S)\le M$  for every $\ell\in\{1,2\}$, due to \cref{lem:Lk-monotone,lem:Lk-substring}.

    Let $\A^*_1$ and $\A^*_2$ be the two returned paths, respectively, and assume both visit the first and last vertices of $B_{i_c,j_c}$, $f$ and $\ell$ (see \cref{lem:corners_0_block}), and share the canonical path $Z$  from $f$ to $\ell$.
    We compute $\A^*_1[s\dd f]\circ Z \circ \A^*_2[\ell\dd e]$ (where $s$ and $e$ are the first and last vertices of $\A$).
    We return the obtained alignment as an optimal alignment from $S$ to $T$.

    \para{Correctness.}
    If $|\shrink(S)| \le 2L + 8t$, the algorithm simply applies \cref{lem:rle-return-alignment} and the correctness follows from the correctness of \cref{lem:rle-return-alignment}.

    Otherwise, the algorithm finds a zero block $B_{i_c,j_c}$ such that if $\DTW_\delta(S,T) \le t$, there is an optimal alignment visiting $B_{i_c,j_c}$ (due to \cref{clm:visit-center}).
    Let $\A'$ be such an optimal alignment (hence, $\cost(\A')=\DTW_\delta(S,T)$) visiting $B_{i_c,j_c}$ at some vertex $w$.
    Let $\A'_1\coloneqq\A'[s,w]\circ(w\rightsquigarrow \ell)$ and $\A'_2\coloneqq(f\rightsquigarrow w)\circ\A'[w,e]$.
    The algorithm recursively finds optimal alignments $\A^*_1$ and $\A^*_2$ from $B_{1,1}$ to $B_{i_c,j_c}$ and from $B_{i_c,j_c}$ to $B_{n',m'}$ respectively.
    It follows from the optimality of $\A^*_1$ and $\A^*_2$ that $\cost(\A^*_1) \le \cost(\A'_1)$ and $\cost(\A^*_2) \le \cost(\A'_2)$.
    Therefore, the path obtained from combining $\A^*_1$ and $\A^*_2$ has cost at most $\cost(\A'_1) + \cost(\A'_2)= \cost(\A') = \DTW_\delta(S,T)$, as required.

    \para{Complexity.}
    If $|S'| \le 2L + 8t$, the algorithm applies \cref{lem:rle-return-alignment} which runs in $\Otild(t |S'|) = \Otild(t (L+t)) = \Oh(t^2M)$ time.

    Otherwise, the algorithm applies \cref{lem:rle-return-alignment} on the pair of strings $S_c$ and $T_c$ such that $|\shrink(S_c)|=\Oh(L+t)$, which as stated before takes $\Otild(t\cdot  |\shrink(S_c)|)=\Otild(t(L+t))=\Otild(t^2M)$ time.
    The algorithm then constructs $\bar \A$, splits it into $\A_1$ and $\A_2$ where all paths are in a compressed representation in $\Oh(t+1)$ time. 
    The algorithm runs two recursive calls, $(S_1,T_1)$ and $(S_2,T_2)$.
    Since $i_c\in[n'/2-t\dd n'/2+t]$ and $t\ge1$, both $S_1$ and $S_2$ contain at most $n'/2+t+1\le n'/2+2t$ runs.
    Moreover, since the base case does not apply, $n'>2L+8t>8t$.
    Thus, each recursive instance contains fewer than $3n'/4$ runs.
    Hence, the recursion depth is $\Oh(\log n)$.
    



    For every recursive call with a path $\A_\ell$, let $c'\coloneqq\cost(\A_\ell)$ and $t'\coloneqq\ceil{c'}$.
    Since every path costs either $0$ or at least $1$, we have $t'\le 2c'$.
    Since the sum of the costs $c'$ at every recursion level is at most
    the initial budget $t$, we have $\sum t' \le 2\sum c' \le 2t$. 
    So, the internal running time of each recursive call on $(S',T')$ with path of cost $c'$ is dominated by the application of \cref{lem:rle-return-alignment} which takes $\Otild((t')^2 M)= \Otild((t') \cdot t \cdot M)$ time (the inequality is by \cref{lem:Lk-monotone,lem:Lk-substring}).
    Since the sum of all $t'$ budgets in a level is at most $2t$, this sums up to $\Otild(t ^2M)$ per layer, for a total of $\Otild(t ^2M)$ time across all layers.
\end{proof}

\subsection{Using Edit-Distance to Approximate $\DTW$}\label{sec:ED-to-approx-DTW}

In the following two lemmas we show a close relation between the $\DTW$ and $\ED$ of two strings $S$ and $T$ with $S=\shrink(S)$ and $T=\shrink(T)$.

\begin{lemma}\label{lem:DTW>ED/2}
Let $S$ and $T$ be two strings such that $S=\shrink(S)$ and $T=\shrink(T)$ and let $\delta$ be a mismatch-cost function.
Then $\ED(S,T)\le 2\DTW_{\delta}(S,T)$. 
\end{lemma}
\begin{proof}
We note that both $S$ and $T$ have the property that every two consecutive characters are distinct.

Let $\A$ be a shortest path from $(1,1)$ to $(|S|,|T|)$ in $\AG_\delta(S,T)$.
Consider the path $\A'$ in $\AGed(S,T)$, starting from the edge $(0,0) \xrightarrow{} (1,1)$ and proceeding using exactly the same edges as $\A$.
Consider an edge $e=(u,v)$ in $\A'$ such that the weight of $e$ in $\AGed$ is $1$ and the cost of $v$ in $\AG_\delta(S,T)$ is $0$.
We consider three cases regarding the orientation of $e$: either vertical, horizontal, or diagonal.
We start by observing that $e$ cannot be a diagonal edge.
That is since a diagonal edge with weight $1$ in $\AGed(S,T)$ must enter a vertex $(i,j)$ such that $S[i]\ne T[j]$.
This contradicts our assumption that $v=(i,j)$ has cost $0$ in $\AG_\delta(S,T)$.
Consider the case in which $e$ is a horizontal edge from $(i,j-1)$ to $(i,j)$.
We assume that the cost of $(i,j)$ in $\AG_\delta(S,T)$ is $0$, and therefore $S[i]=T[j]$.
Since every two consecutive symbols in $T$ are distinct, it must hold that $S[i]\neq T[j-1]$.
Therefore, the vertex $(i,j-1)$ that precedes $(i,j)$ in $\A$ has cost at least $1$.
We charge the excess cost of $\A'$ compared to $\A$ contributed by $e$ on the vertex $(i,j-1)$.
Notice that this mapping is injective.
The same arguments can be applied for the case in which $e$ is a vertical edge.

We have shown that every non-zero edge in $\A'$ either corresponds to a non-zero vertex in $\A$ or it can be uniquely mapped to a non-zero vertex in $\A$.
It follows that the cost of $\A'$ is at most twice the cost of $\A$.
In particular, $\ED(S,T) \le \cost(\A') \le 2\cdot\cost(\A) = 2\cdot\DTW_\delta(S,T)$.
\end{proof}

\begin{lemma}\label{lem:ED-alignment-to-DTW-alignment}
Let $S$ and $T$ be two strings and let  $\A$ be an edit-distance alignment in $\AGed(S,T)$.
Then, there exists a unit-cost $\DTW$ alignment $\A'$ in $\AG_\unit(S,T)$ with $\cost(\A')\le\cost(\A)$.
Moreover, there exists an algorithm that, given $\A$ (in compressed representation) computes (a compressed representation of) $\A'$ in $\Oh(|S|+|T|)$ time.
\end{lemma}
\begin{proof}

We construct a path $\A'$ in $\AG(S,T)$ as follows.
Let $(i,j)$ be the first vertex of $\A$ with $i\ge 1$ and $j\ge 1$.
Since $\A$ is monotone, either $i=1$ or $j=1$.
We start $\A'$ at $(1,1)$, move horizontally to $(1,j)$ if $i=1$ (or vertically to $(i,1)$ if $j=1$), and from there continue using exactly the same moves as $\A$.
Thus, $\A'$ is a valid $\DTW$ alignment.

We compare the costs.
For every vertex $(p,q)$ visited by $\A'$ from $(i,j)$ onward, the cost in $\AG(S,T)$ is at most the cost paid by $\A$ when entering $(p,q)$:
if $\A$ enters $(p,q)$ diagonally, the costs are equal, and if it enters horizontally or vertically, $\A$ pays $1$ while the vertex cost is at most $1$.

It remains to bound the added prefix.
If $(i,j)=(1,j)$, then the prefix of $\A'$ has cost at most $j-1$, while $\A$ must pay at least $j-1$ to reach $(1,j)$ along the boundary of $\AGed(S,T)$.
The case $(i,j)=(i,1)$ is symmetric.
Therefore, $\cost(\A')\le \cost(\A)$.

The construction scans $\A$ once and outputs $\A'$, hence runs in $\Oh(|\A|)=\Oh(|S|+|T|)$ time.
\end{proof}

Using \cref{lem:DTW>ED/2,lem:ED-alignment-to-DTW-alignment} we are ready to prove the following lemma, which is used as the first step in the proof of \cref{thm:main-UB}.


\begin{lemma}\label{lem:ub-base-DTW-unit-for-shrinked}
    Given two strings $S$ and $T$, and some $M\ge L_{24k}(S)$, there exists an algorithm that runs in $\Otild(n+k^2M)$ time and computes an optimal alignment in $\AG_\unit(\shrink(S),\shrink(T))$, or reports $\DTW(\shrink(S),\shrink(T))>k$.
\end{lemma}
\begin{proof}
Let $S'\coloneqq\shrink(S)$ and $T'\coloneqq\shrink(T)$.
We use \cref{lem:LV} to compute $2\kED(S',T')$ and an optimal edit distance alignment  from $S'$ to $T'$, if $\ED(S',T') \le 2k$.
If \cref{lem:LV} reports that $\ED(S',T') > 2k$, we report that $\DTW(S',T') > k$ (this is correct by \cref{lem:DTW>ED/2}).
Otherwise, \cref{lem:LV} returns an optimal edit distance alignment $\A''$ of cost at most $2k$.
We apply \cref{lem:ED-alignment-to-DTW-alignment} with $\A''$ to obtain a $\DTW$ alignment $\A'$  from $S'$ to $T'$ with cost at most $2k$.
If $\cost_{\unit}(\A')=0$, we return $\A'$, which is already optimal.
Otherwise, if $\cost_{\unit}(\A')>0$, we apply \cref{lem:dtw-given-approx}
on $S'$ and $T'$, the mismatch-cost function $\unit$, the alignment $\A'$, the parameter $t\coloneqq\cost(\A')$ and the upper bound $M$ to obtain an optimal alignment $\A^*$.
Notice that $L_{12t}(S')\le L_{24k}(S')=L_{24k}(S)\le M$, since $t=\cost(\A')\le2k$ and $S'=\shrink(S)$.
Finally, if $\cost(\A^*)>k$ the algorithm returns $\infty$.

\para{Complexity.}
Applying \cref{lem:LV} takes $\Otild(n +k^2)$ time to compute $\A''$.
Computing $\A'$ by \cref{lem:ED-alignment-to-DTW-alignment} takes additional $\Otild(n)$ time.
Finally, the algorithm of \cref{lem:dtw-given-approx} computes $\A^*$ (with $t=\cost(\A')\le2k$ and $M$) in $\Otild(k^2M)$ time.
\end{proof}

\subsection{Scaling the Mismatch-Cost Function}\label{sec:scaling-delta}

Let $\delta$ be a mismatch-cost function.
We define $\half_\delta:\Sigma^2\to\R_{\ge0}$ to be the function that halves every value of $\delta$, but still preserves the requirement that $\delta(\sigma,\sigma')\ge 1$ for every $\sigma\ne\sigma'$.
Formally, $\half_\delta(\sigma,\sigma')$ is $0$ if $\sigma=\sigma'$ and $\max\{\delta(\sigma,\sigma')/2,1\}$ otherwise.
Notice that, given oracle access to $\delta$, one can also obtain oracle access to $\half_\delta$ with an additional constant time per query.

For two strings $S$ and $T$, recall that $\AG_\delta(S,T)$ is a vertex-weighted graph, where the weights of the vertices are a function of $\delta$.
Further recall that $\AG_\delta(S,T)$ and $\AG_{\half_\delta}(S,T)$ are both vertex-weighted graphs, with the same vertices and edges, which are different only in the weights of the vertices.
For a path $\A$, we denote by $\cost_\delta(\A)$ the cost of $\A$ in $\AG_\delta(S,T)$, and by $\cost_{\half_\delta}(\A)$ the cost of $\A$ in $\AG_{\half_\delta}(S,T)$.
The following lemma shows that $\cost_{\half_\delta}(\A)$ is an approximation for $\cost_\delta(\A)$.
\begin{lemma}\label{lem:half-delta-approx}
Let $S$ and $T$ be two strings, and let $\delta$ be a mismatch-cost function.
For every alignment  $\A$  from $S$ to $T$ we have $\cost_{\delta}(\A)\in[\cost_{\half_\delta}(\A), 2\cost_{\half_\delta}(\A)]$.
\end{lemma}
\begin{proof}
We have $\cost_{\delta}(\A)=\sum_{(i,j)\in \A} \delta(S[i],T[j])$ and $\cost_{\half_\delta}(\A)=\sum_{(i,j)\in \A} \half_\delta(S[i],T[j])$.
For every $\sigma,\sigma'\in\Sigma$, if $\sigma=\sigma'$ then both sides are $0$.
Otherwise, $\half_\delta(\sigma,\sigma')=\max\{\delta(\sigma,\sigma')/2,1\}\ge \delta(\sigma,\sigma')/2$, and hence $\delta(\sigma,\sigma')\le 2\half_\delta(\sigma,\sigma')$.
On the other hand, we also have $\half_\delta(\sigma,\sigma')=\max\{\delta(\sigma,\sigma')/2,1\}\le \delta(\sigma,\sigma')$.
Summing over all vertices visited by $\A$ gives $\cost_{\delta}(\A)\in[\cost_{\half_\delta}(\A) ,2\cost_{\half_\delta}(\A)]$.
\end{proof}

The following lemma shows that given an optimal alignment with respect to the unit-cost $\DTW$, we can compute an optimal alignment with respect to $\DTW_\delta$.

\begin{lemma}\label{lem:scale-delta}
There exists an algorithm that, given two strings $S$ and $T$, some $M\ge L_{24k}(S)$, oracle access to a mismatch-cost function $\delta$ and an optimal alignment $\A$ in $\AG_\unit(S,T)$ of cost at most $k$, computes in $\Otild(k^2M)$ time an optimal alignment in $\AG_\delta(S,T)$ or reports that $\kDTW_\delta(S,T)>k$.
\end{lemma}
\begin{proof}
Assume without loss of generality that $\delta(\sigma,\sigma')\le k+1$, since if this is not the case we can replace every mismatch cost larger than $k$ by $k+1$, without affecting $\kDTW_\delta(S,T)$.
Now, we define a sequence of mismatch-cost functions $\delta=\delta_0,\delta_1,\delta_2,\dots,\delta_{\ceil{\log (k+1)}}$, as follows.
$\delta_0\coloneqq\delta$, and for every $i>0$ we define $\delta_i\coloneqq\half_{\delta_{i-1}}$.
Recall that we have oracle access to a mismatch-cost function $\delta$.
For integer input weights, every cost at stage $i$ is a multiple of $2^{-i}$.
Multiplying alignment costs by $2^i$ gives integers with $\Oh(w)$ bits, since $i=\Oh(\log k)=\Oh(\log n)$ and alignments have at most $2n-1$ vertices.
Thus, the scaled costs are represented exactly using a constant number of words on the word RAM.
Moreover, it is clear that $\delta_{\ceil{\log (k+1)}}$ is the function $\unit$, since  $\delta(\sigma,\sigma')\in[1,k+1]$ (for $\sigma\ne\sigma'$) we have $\delta_{\ceil{\log(k+1)}}(\sigma,\sigma')\in[1,\max\{1,(k+1)/2^{\ceil{\log (k+1)}}\}]=\{1\}$.

Let $\A_{\ceil{\log (k+1)}}\coloneqq\A$.
For every $i$, from $\ceil{\log (k+1)}-1$ downward to $0$ we define $\A_i$ to be an optimal alignment of $\AG_{\delta_i}(S,T)$.
For every $i$, the algorithm computes $\A_i$, using \cref{lem:dtw-given-approx} with $\A_{i+1}$, $M$ and $t\coloneqq2k$.
If at some point during the computation we have $\DTW_{\delta_i}(S,T)=\cost_{\delta_i}(\A_i)>k$ the algorithm returns that $\kDTW_\delta(S,T)>k$.
This is true since, by definition of $\half_\delta$ we have $\DTW_\delta(S,T)\ge\DTW_{\half_\delta}(S,T)$ and therefore $\DTW_\delta(S,T)=\DTW_{\delta_0}(S,T)\ge\DTW_{\delta_i}(S,T)>k$.

\para{Complexity.}
There are $\Oh(\log k)$ iterations.
At the $i$th iteration, we begin with $\A_{i+1}$, which by \cref{lem:half-delta-approx} satisfies $\cost_{\delta_i}(\A_{i+1})\le 2\cost_{\delta_{i+1}}(\A_{i+1})\le 2k$.
We use \cref{lem:dtw-given-approx} that takes $\Otild(k^2 M)$.
Hence, the total running time of the algorithm is $\Otild(k^2M)$.
\end{proof}

\subsection{Scaling the RLE Exponents}\label{sec:scaling-exponents}

In this section we show how to compute  an optimal alignment for $\AG_\delta(S,T)$ given an optimal alignment for $\AG_\delta(\shrink(S),\shrink(T))$.

For a string $S$ with run length encoding $(\sigma_1,a_1),(\sigma_2,a_2), \ldots (\sigma_\ell,a_\ell)$, we define $\tshrink(S)$ as the string with run length encoding $(\sigma_1,\ceil{ a_1/2}),(\sigma_2,\ceil{a_2/2}), \ldots (\sigma_\ell,\ceil{a_\ell/2})$.
In words, $\tshrink$ is obtained from $S$ by shrinking each run of $S$ to half its original size.

\begin{lemma}\label{lem:tshrink-approx-dtw}
    For strings $S$ and $T$ and a mismatch-cost function $\delta$, it holds that $\DTW_\delta(\tshrink(S),\tshrink(T)) \in [\DTW_\delta(S,T) / 2 \dd \DTW_\delta(S,T)]$.
    Furthermore, given an alignment $\A$ from $\tshrink(S)$ to $\tshrink(T)$ (in compressed representation), there is an algorithm that outputs (a compressed representation of) an alignment $\A'$ from $S$ to $T$ with cost at most $2\cdot\cost(\A)$ in $\Oh(|S| + |T|)$ time.
\end{lemma}
\begin{proof}
In this proof we follow the definition of $\DTW_\delta$ via expansions of $S$ and $T$ (see the discussion before \cref{lem:expensions-alignments}).

\para{Upper bound.}
First, the upper bound $\DTW_\delta(\tshrink(S),\tshrink(T))\le\DTW_\delta(S,T)$ is straightforward.
Let $S'$ and $T'$ be equal-length expansions of $S$ and $T$, respectively, such that
\[
\DTW_\delta(S,T)=\sum_{i=1}^{|S'|}\delta(S'[i],T'[i]).
\]
Since every run of $\tshrink(S)$ is obtained from the corresponding run of $S$ by decreasing its length, every expansion of $S$ is also an expansion of $\tshrink(S)$.
Thus $S'$ and $T'$ are also expansions of $\tshrink(S)$ and $\tshrink(T)$, respectively.
Hence,
\begin{align*}
\DTW_\delta(\tshrink(S),\tshrink(T))
&=
\min_{\substack{
S'' \text{ is an expansion of } \tshrink(S)\\
T'' \text{ is an expansion of } \tshrink(T)\\
|S''|=|T''|
}}
\sum_{i=1}^{|S''|}\delta(S''[i],T''[i])\\
&\le \sum_{i=1}^{|S'|}\delta(S'[i],T'[i])\\
&=\DTW_\delta(S,T).
\end{align*}

\para{Lower bound.}
We now prove that $\DTW_\delta(\tshrink(S),\tshrink(T))\ge \DTW_\delta(S,T)/2$ by showing that $\DTW_\delta(S,T)\le 2\cdot\DTW_\delta(\tshrink(S),\tshrink(T))$.
We prove this together with the algorithmic claim that, given an alignment $\A$ from $\tshrink(S)$ to $\tshrink(T)$, we can compute an alignment $\A'$ from $S$ to $T$ whose cost is at most $2\cost(\A)$.

First, we compute strings $S'$ and $T'$ which are same-length expansions of $\tshrink(S)$ and $\tshrink(T)$, respectively, such that
\[
\cost(\A)=\sum_{i=1}^{|S'|}\delta(S'[i],T'[i]).
\]
Then let $S''$ be the string obtained by duplicating each character of $S'$, that is, for every $i\in[|S'|]$ we define both $S''[2i-1]$ and $S''[2i]$ to be $S'[i]$.
Similarly, we define $T''$ to be the string obtained by duplicating each character of $T'$.
Clearly $|S''|=|T''|$, and
\begin{align*}
\sum_{i=1}^{|S''|}\delta(S''[i],T''[i])
&=
\sum_{i=1}^{|S'|} \big(\delta(S''[2i-1],T''[2i-1]) + \delta(S''[2i],T''[2i])\big) \\
&=
\sum_{i=1}^{|S'|} 2\delta(S'[i],T'[i]) \\
&=
2\sum_{i=1}^{|S'|}\delta(S'[i],T'[i])
=
2\cost(\A).
\end{align*}

Finally, we claim that $S''$ is an expansion of $S$.
Recall that if the run-length encoding of $S$ is $(\sigma_1,a_1),(\sigma_2,a_2), \ldots,(\sigma_\ell,a_\ell)$, then the run-length encoding of $\tshrink(S)$ is $(\sigma_1,\ceil{a_1/2}),(\sigma_2,\ceil{a_2/2}), \ldots,(\sigma_\ell,\ceil{a_\ell/2})$.
Since $S'$ is an expansion of $\tshrink(S)$, the run-length encoding of $S'$ is $(\sigma_1,b_1),(\sigma_2,b_2), \ldots,(\sigma_\ell,b_\ell)$ with $b_i\ge \ceil{a_i/2}$ for every $i\in[\ell]$.
Thus the run-length encoding of $S''$ is $(\sigma_1,c_1),(\sigma_2,c_2), \ldots,(\sigma_\ell,c_\ell)$ with $c_i\coloneqq2b_i\ge a_i$ for every $i\in[\ell]$.
Hence $S''$ is indeed an expansion of $S$.
By the same argument, $T''$ is an expansion of $T$.

Finally, by \cref{lem:expensions-alignments} we can compute an alignment $\A'$ from $S$ to $T$ according to $S''$ and $T''$ such that
\[
\cost(\A')\le \sum_{i=1}^{|S''|}\delta(S''[i],T''[i])=2\cost(\A).
\]

Applying the above construction to an optimal alignment $\A$ from $\tshrink(S)$ to $\tshrink(T)$, for which $\cost(\A)=\DTW_\delta(\tshrink(S),\tshrink(T))$, we obtain an alignment $\A'$ from $S$ to $T$ such that
\[
\DTW_\delta(S,T)\le \cost(\A')\le 2\cost(\A)=2\DTW_\delta(\tshrink(S),\tshrink(T)).
\]
Rearranging gives $\DTW_\delta(\tshrink(S),\tshrink(T))\ge \DTW_\delta(S,T)/2$.

\para{Complexity.}
Computing $S'$ and $T'$ from $\A$ using \cref{lem:expensions-alignments} takes $\Oh(|\A|)=\Oh(|S|+|T|)$ time.
Computing $S''$ and $T''$ takes $\Oh(|S'|+|T'|)=\Oh(|\A|)=\Oh(|S|+|T|)$ time.
Finally, computing $\A'$ from $S''$ and $T''$ using \cref{lem:expensions-alignments} takes $\Oh(|S''|)=\Oh(|\A|)=\Oh(|S|+|T|)$ time.
Thus the total running time of the algorithm is $\Oh(|S|+|T|)$.
\end{proof}

\begin{lemma}\label{lem:scale-RLE}
There exists an algorithm that, given two strings $S$ and $T$ of total length $n$, some $M\ge L_{24k}(S)$, oracle access to a mismatch-cost function $\delta$ and an  optimal alignment $\A$ in $\AG_\delta(\shrink(S),\shrink(T))$ of cost at most $k$, computes in $\Otild(n+k^2M)$ time an optimal alignment in $\AG_\delta(S,T)$ or reports that $\kDTW_\delta(S,T)>k$.
\end{lemma}
\begin{proof}
Given $S$ and $T$, we define a sequence $P \coloneqq (S_1,T_1),(S_2,T_2),\dots,(S_{|P|},T_{|P|})$ as follows.
$(S_1,T_1)$ are simply $S$ and $T$.
For $i \ge1$, if $S_{i} = \tshrink(S_{i})$ and $T_{i} = \tshrink(T_{i})$, then $(S_i,T_i)$ is the last pair of the sequence.
Otherwise, $(S_{i+1},T_{i+1}) \coloneqq (\tshrink(S_{i}),\tshrink(T_{i}))$.
Notice that $S_{|P|}=\shrink(S)$ and $T_{|P|}=\shrink(T)$.
Hence, $\A$ is an optimal alignment for $\AG_\delta(S_{|P|},T_{|P|})$.
Let $\A_{|P|}\coloneqq\A$.

The algorithm iterates $i$ from $|P|-1$ to $1$ in decreasing order.
At every iteration, we assume that we already have an optimal $\DTW_\delta$ alignment $\A_{i+1}$, or that we have found that $\DTW_\delta(S_{i+1},T_{i+1}) > k$.
Due to \cref{lem:tshrink-approx-dtw}, if $\DTW_\delta(S_{i+1},T_{i+1}) > k$ we can report that $\DTW_\delta(S_i,T_i) > k$ as well.
Otherwise, we have the alignment $\A_{i+1}$ with cost at most $k$ at hand.
Due to \cref{lem:tshrink-approx-dtw}, we can use the alignment $\A_{i+1}$ from $S_{i+1} = \tshrink(S_i)$ to $T_{i+1} = \tshrink(T_i)$ to construct an alignment $\A'_i$ from $S_i$ to $T_i$ of cost at most $2k$.
Since $\shrink(S_i)=\shrink(S)$ we have $L_{24k}(S_i)=L_{24k}(S)\le M$, using $t\coloneqq2k$ it holds that $L_{12t}(S_i)\le M$.
Then, we can apply \cref{lem:dtw-given-approx} on $(S_i,T_i)$ with the alignment $\A'_i$, the parameter $M$  and $t\coloneqq2k$ to either report that $\DTW_\delta(S_i,T_i) >k$ or return an optimal $\DTW_\delta$ alignment $\A_i$ from $S_i$ to $T_i$, with cost at most $k$.

After running the above for all $i$ from $|P|-1$ down to $1$, we have either reported that $\DTW_\delta(S,T)> k$ or have found an optimal alignment $\A_1$ from $S$ to $T$, as required.

\para{Complexity.} 
Since $|P|=\Oh(\log n)$, the computation of the sequence $P$ takes $\Otild(n)$ time.

For every iteration, the algorithm processes $(S_i,T_i)$ given $\A_{i+1}$, to compute $\A_i$.
If $\DTW_\delta(S_{i+1},T_{i+1})>k$, the algorithm terminates, and this  iteration takes $\Otild(1)$ time.
Otherwise, the computation of $\A'_{i}$ from $\A_{i+1}$ takes $\Otild(n)$ time by \cref{lem:tshrink-approx-dtw}.
Then, the computation of $\A_i$ from $\A'_i$ takes $\Otild(k^2M)$ time, by \cref{lem:dtw-given-approx}.
Since we have only $|P|=\Oh(\log n)$ iterations, the total running time of the algorithm is $\Otild(n+k^2M)$.
\end{proof}

\subsection{Wrapping Up}
The proof of \cref{thm:main-UB} is straightforward using \cref{lem:ub-base-DTW-unit-for-shrinked,lem:scale-delta,lem:scale-RLE}.
\restatementwithproof{\MainUB*}
\begin{proof}
For $0\le k<1$, normalization implies that an alignment of cost at most $k$ must have cost $0$.
Compare $\shrink(S)$ and $\shrink(T)$ in $\Oh(n)$ time: if they are equal, matching corresponding runs yields an optimal alignment of cost $0$; otherwise, return \No.
For $k\ge n$, use the standard $\Oh(n^2)$-time dynamic program, which meets the claimed bound since $L_p(S)\ge1$.
It remains to consider $1\le k<n$.
We first describe the algorithm for integer thresholds.
First, the algorithm computes $M$, a $2$-approximation for $L_{24k}(S)$, using \cref{lem:approx_LpS}.
Then, the algorithm applies \cref{lem:ub-base-DTW-unit-for-shrinked}. 
If \cref{lem:ub-base-DTW-unit-for-shrinked} returns $\infty$, the algorithm returns $\infty$ for $\kDTW_\delta(S,T)$.
Otherwise, let $\A_1$ be the optimal alignment in $\AG_\unit(\shrink(S),\shrink(T))$ returned by \cref{lem:ub-base-DTW-unit-for-shrinked}.
The algorithm uses $\A_1$ to apply \cref{lem:scale-delta}.
The algorithm of \cref{lem:scale-delta} either returns $\infty$, indicating that $\DTW_\delta(\shrink(S),\shrink(T))>k$, or computes an optimal alignment $\A_2$ in $\AG_\delta(\shrink(S),\shrink(T))$.
In the former case, the algorithm returns $\infty$ for $\kDTW_\delta(S,T)$.
In the latter case, the algorithm uses $\A_2$ to apply the algorithm of \cref{lem:scale-RLE}, and returns its output.
Finally, since the alignment is computed in compressed representation, the algorithm converts it into an explicit optimal alignment in $\Oh(n)$ time.

\para{Correctness.}
The correctness of the algorithm follows directly from \cref{lem:ub-base-DTW-unit-for-shrinked,lem:scale-delta,lem:scale-RLE}, and from the inequality 
\[
\DTW_\delta(S,T)\ge \DTW_\delta(\shrink(S),\shrink(T))\ge \DTW(\shrink(S),\shrink(T)).
\]
Notice that the first inequality follows from \cref{lem:tshrink-approx-dtw} and the second inequality follows from \cref{lem:half-delta-approx} (both by recursive application).

\para{Complexity.}
Computing $M$ using \cref{lem:approx_LpS} takes $\Otild(n)$ time.
The time complexity of \cref{lem:ub-base-DTW-unit-for-shrinked} is $\Otild(n+k^2M)$.
The time complexity of \cref{lem:scale-delta} is $\Otild(k^2M)$ and the time complexity of \cref{lem:scale-RLE} is $\Otild(n+k^2M)$.
Hence, the total time complexity of the algorithm is $\Otild(n+k^2M)$.
Since $M\le 2L_{24k}(S)$ the total running time of the algorithm is $\Otild(n+k^2M)=\Otild(n+k^2L_{24k}(S))$.

For a non-integral threshold $1\le k<n$, compute $K\coloneqq\ceil{k}$ by binary search over the integers in $[1\dd n]$, using $\Oh(\log n)$ comparisons.
Run the integer-threshold algorithm with $K$.
If it reports that the distance exceeds $K$, return \No; otherwise, compare the cost of the returned optimal alignment with $k$ and output the alignment if the answer is \Yes.
Since $K\le2k$ and $p=24K$, the running time is $\Otild(n+K^2L_{24K}(S))=\Otild(n+k^2L_p(S))$, as claimed.
\end{proof}

\subparagraph*{AI Disclosure:}
We used GPT 5.4 to assist with generating figures (producing TikZ code), content refactoring (e.g., extracting a repeated argument to a separate lemma, with instructions what the lemma should state), content summarization (e.g., explaining how \cref{thm:ov-to-dtw} follows from earlier constructions and lemmas; see \cpageref{pg:proof}), and drafting a small number of proofs that could otherwise be left as exercises for a knowledgeable reader (e.g., that DTW behaves predictably with respect to padding fresh characters; \cref{lem:dtw-padding}).
We subsequently verified and adjusted the content generated this way to ensure correctness, simplicity, and clarity.
We did not use any AI assistance for conceptual tasks such as designing gadgets, algorithms, or analysis strategies.

\bibliography{bib.bib}


\end{document}